\documentclass[10pt,oneside,english,reqno]{amsart}
\usepackage[T1]{fontenc}
\usepackage[latin9]{inputenc}
\usepackage{prettyref,enumitem}
\usepackage{amstext}
\usepackage{amsthm}
\usepackage{amssymb}
\usepackage{esint}
\usepackage{xcolor}
\usepackage{comment}
\usepackage{float}
 \usepackage[foot]{amsaddr}

\usepackage{tikz}
\usepackage{tikz-3dplot}
\usetikzlibrary{decorations.pathreplacing}
\usepackage{pgfplots}

\makeatletter
\numberwithin{equation}{section}
\theoremstyle{plain}
\newtheorem{thm}{\protect\theoremname}
\theoremstyle{plain}
\newtheorem{prop}[thm]{\protect\propositionname}
\theoremstyle{plain}

\theoremstyle{plain}
\newtheorem{rem}[thm]{\protect\remarkname}
\theoremstyle{plain}
\newtheorem{lem}[thm]{\protect\lemmaname}
\theoremstyle{plain}
\newtheorem{cor}[thm]{\protect\corollaryname}

\usepackage{mathrsfs}
\usepackage{amsthm}\usepackage{bbm}\usepackage{fullpage}

\newcommand{\sech}{\text{sech}}

\newcommand{\dd}{\mathrm{d}}

\newcommand{\lucas}[1]{{\color{red}#1}}

\newrefformat{claim}{Claim \ref{#1}}

\usepackage{babel}

\makeatother

\usepackage{babel}
\providecommand{\conjecturename}{Conjecture}
\providecommand{\lemmaname}{Lemma}
\providecommand{\propositionname}{Proposition}
\providecommand{\remarkname}{Remark}
\providecommand{\theoremname}{Theorem}
\providecommand{\corollaryname}{Corollary}

\begin{document}

%\tableofcontents

\title{Plate theory for metric-constrained actuation of liquid crystal elastomer sheets}

\author{Lucas Bouck$^{1}$}
\address{$^1$ Department of Mathematical Sciences,  Carnegie Mellon University, Pittsburgh, PA, 15213, USA }
\author{David Padilla-Garza$^{2}$}
\address{$^2$ Einstein Institute of Mathematics, The Hebrew University of Jerusalem, Jerusalem, 9190401, Israel }
\author{Paul Plucinsky$^{3,\ast}$}
\address{$^3$ Aerospace and Mechanical Engineering, University of Southern California, Los Angeles,  CA, 90089, USA}
\address{$\ast$ Corresponding Author Email: plucinsk@usc.edu}

\begin{comment}\author{Lucas Bouck, David Padilla-Garza, Paul Plucinsky}
\end{comment}

\date{\today}

\maketitle 

\vspace*{-1cm}

\begin{abstract} 
Liquid crystal elastomers (LCEs) marry the large deformation response of a cross-linked polymer network with the nematic order of liquid crystals pendent to the network. Of particular interest is the actuation of LCE sheets where the nematic order, modeled by a unit vector called the director, is specified heterogeneously in the plane of the sheet. Heating such a sheet leads to a large spontaneous deformation, coupled to the director design through a metric constraint that is now well-established by the literature. Here we go beyond the metric constraint and identify the full plate theory that underlies this phenomenon. Starting from a widely used bulk model for LCEs, we derive a plate theory for the pure bending deformations of patterned LCE sheets in the limit that the sheet thickness tends to zero using the framework of $\Gamma$-convergence. Specifically, after dividing the bulk energy by the cube of the thickness to set a bending scale, we show that all limiting midplane deformations with bounded energy at this scale satisfy the aforementioned metric constraint. We then identify the energy of our plate theory as an ansatz-free lower bound of the limit of the scaled bulk energy, and construct a recovery sequence that achieves this plate energy for all smooth enough midplane deformations. We conclude by applying our plate theory to a variety of examples.
\end{abstract}

\section{Introduction}

Liquid crystal elastomers (LCEs) are rubbery solids composed of a lightly cross-linked polymer network with rod-like mesogen molecules ("liquid crystals") pendent to the polymer backbone. At low temperatures, the liquid crystals have a tendancy towards alignment that gives rise to an orientational order --- a nematic phase described macroscopically by a director (unit vector) and order parameter at each point in the solid. This orientational order strongly couples to the entropic elasticity of the polymer network, leading to a rich range of mechanical behaviors \cite{warner2007liquid}. On increasing the temperature, for instance, thermal fluctuations suppress the LCE's nematic order, resulting in a phase transition that renders the material effectively isotropic at high temperatures. Large spontaneous distortion accompanies the transition --- the solid contracts along its initial director (by roughly $20-80\%$ strain depending on the cross-linking density) and expands transversely in a manner that is nearly incompressible. 

\begin{figure}[h!]
\hspace*{-0.5cm}
\centering
\includegraphics[width=.9\textwidth]{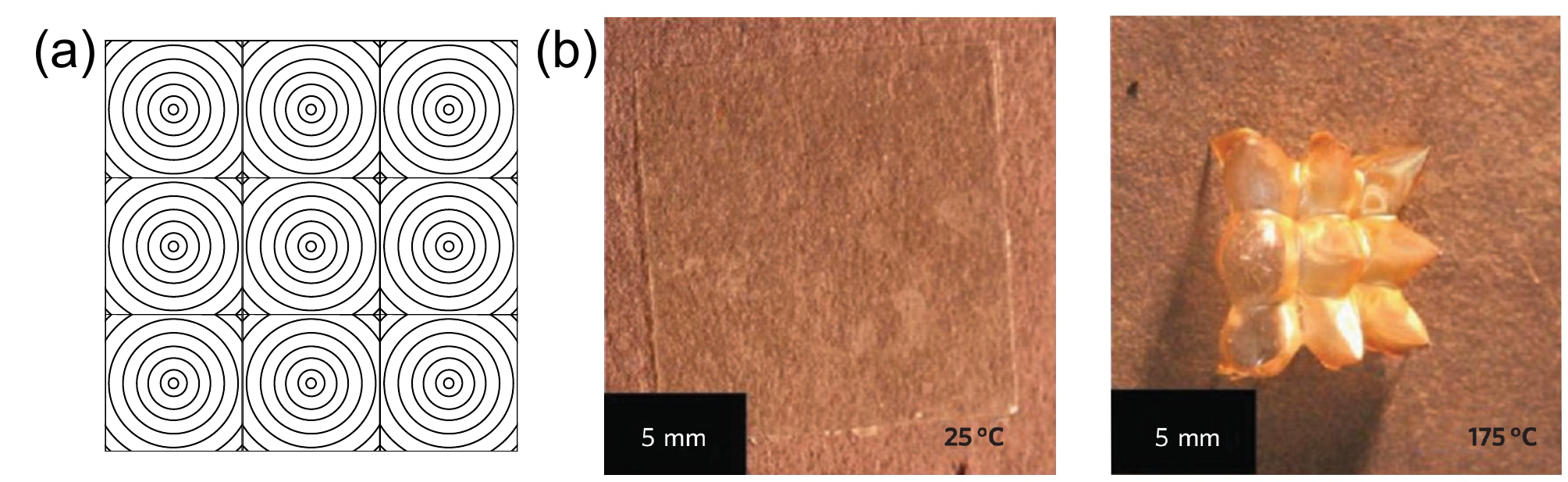} 
\caption{Example of hetergeneous director programming and actuation in LCE sheets; the experimental images in (b) are reproduced from \cite{ware2015voxelated,white2015programmable}.}
\label{fig:IntroFig}
\end{figure}

About a decade ago, materials scientists pioneered methods for synthesizing LCE sheets with heterogeneous control over the director profile in the plane of the sheet \cite{de2012engineering,ware2015voxelated}. This innovation enabled dramatic displays of shape-actuation and solidified LCEs as an important materials template for emerging applications in soft robotics and artificial muscles \cite{white2015programmable}. Fig.\;\ref{fig:IntroFig} highlights an exemplar from the literature --- the conical actuation of a voxelated LCE sheet in \cite{ware2015voxelated}. This LCE is formed in a thin layer between two glass plates whose surfaces are photo-patterned to trigger spatially varying alignment of the mesogens molecules during the polymerization process; the plates are then removed to produce a free-standing and heterogeneously patterned sheet. The resulting director profile, sketched in Fig.\;\ref{fig:IntroFig}(a), is oriented azimuthally (tangent to the lines shown) around nine +1-defects. Upon heating, the sheet undergoes a nematic-to-isotropic phase transition that causes azimuthal contraction and radial expansion about each defect, resulting in the global conical actuation shown Fig.\;\ref{fig:IntroFig}(b). Other experimental modalities have emerged in recent years that enrich the programming capability. LCEs can be 3D printed --- the mesogens align with the printing direction to produce heterogeneously patterned plates and shells \cite{ambulo2017four,kotikian20183d} with similar actuation principles as those on display in Fig.\;\ref{fig:IntroFig}. Refinements in chemistry have also enabled actuation by light, as well as heat, a more practical stimuli-response for engineering applications \cite{ gelebart2017making,yu2003directed}. All told, LCEs make for an exciting engineering system to explore the interplay between design and deformation in a sheet. 

 The first theoretical examination of director programming actually precedes the experimental work and came by way of Modes \textit{et al.} \cite{modes2011gaussian,modes2011blueprinting}. Starting from the classical trace formula for the entropic energy of an LCE \cite{bladon1993transitions}, they predicted that conical and saddle-like shapes should emerge on heating a plate programmed using "+1" and "-1"-defects, respectively. A flurry of theoretical activity followed. Building on universal and elegant concepts in the development of non-Euclidean plate theory \cite{efrati2009,klein2007shaping,sharon2010mechanics}, Aharoni \textit{et al.}\;\cite{aharoni2014geometry} proposed a metric constraint to govern the relationship between the director program and actuation in patterned LCE sheets. They also described a class of surfaces of revolution that solved the constraint exactly.
 Mostajeran and Warner continued this theme by showing how to blueprint a variety of Gaussian curvatures \cite{mostajeran2015curvature} and by developing a family of log-spiral designs and actuations \cite{warner2018nematic}. They also worked closely with experimentalists to validate these design motifs \cite{kowalski2018curvature,mostajeran2016encoding}. Plucinsky \textit{et al.}\;\cite{plucinsky2016programming} generalized the metric constraint to include jumps in the director field, enabling complex and easily realizable shape change through nonisometric origami \cite{plucinsky2018patterning,warner2020topographic}. They further complimented this literature by showing that the smooth metric constraint arises as a necessary condition for bending deformations in LCE sheets \cite{plucinsky2018actuation}. Finally, Aharoni \textit{et al.}\;\cite{aharoni2018universal} proposed and exemplified an inverse design strategy based on the metric constraint to program an LCE sheet to take a fairly arbitrary 3D shape. The theme of inverse design continues to be an active area of theoretical research \cite{duffy2024programming,griniasty2019curved,griniasty2021multivalued}. Other intriguing avenues of recent research concern curved creases and how they interact in a predictable and functional way with the metric constraint \cite{feng2020evolving,feng2024geometry,feng2022interfacial}. All of this literature highlghts the richness and importance of the metric constraint as a concept for designable actuation in LCE sheets.

These successes notwithstanding, there has been a thrust in recent years to go beyond a purely geometric metric constraint and develop plate theories for LCE sheets. A variety of researchers \cite{bouck2023convergent,cirak2014computational,duffy2020defective} have proposed and numerically implemented Koiter-type theories --- the metric constraint is relaxed and incorporated into a stretching energy term proportional to the plate thickness, while the bending term at thickness cubed is taken as an \textit{ad hoc} quadratic energy density of a bending strain. These theories are successful at simulating shape actuation in LCE sheets, including under loads and when defects introduce sources of frustration \cite{bouck:2023thesis, bouck2024reduced, duffy2023lifting}. However, their bending terms are not systematically derived and thus leave something to be desired from a theoretical perspective. Many have taken up the challenge of dimension reduction for LCE sheets, though not so much in asymptotic regimes where the aformentioned metric constraint plays a leading role. Mihai and Goriely \cite{mihai2020plate} derived a F\"{o}ppl-von K\'arm\'an-like set of plate equations for LCE sheets and used it to study wrinkling of a nematic-isotropic bilayer \cite{goriely2021liquid}. Bartels \textit{et al.}\;\cite{bartels2023nonlinear} also studied a bilayers problem in which the isotropic layer mutes the shape actuation of the nematic layer in a way that enabled them to derive a Kirchhoff-like nonlinear plate theory whose bending deformations are euclidean isometries. Liu \text{et al.}\;\cite{liu2020consistent} took a non-asymptotic approach to derive a plate theory for LCEs by expanding each field in a power series in the thickness coordinate and by constraining the expansion to solve equilibrium-type equations that arise at successive powers of the thickness. Finally, Virga and colleagues \cite{ozenda2020blend} derived a blended plate theory for LCE sheets that accounts for both stretching and bending and used it to study the elastic behavior of ridges that naturally arise when the programmed director only possesses non-smooth isometric immersions \cite{pedrini2021ridge}. 

We seek instead a dimensionally reduced plate theory for LCE sheets where the metric constraint --- the key guiding principle for large shape actuation in these sheets --- is front and center. Our starting point is a 3D energy for LCEs that is function of the reference (or programmed) director, current (or deformed) director, and the deformation. It consists of four terms with well-established physical origins: (1) an entropic energy that follows the neo-classical theory of Bladon \textit{et al.}\;\cite{bladon1993transitions}; (2) a standard incompressible energy penalty expressing the basic fact LCEs are polymers and so nearly incompressible; (3) a non-ideal energy \cite{biggins2012elasticity} that promotes the director to convect with the deformation and is present because shape actuating LCE sheets are typically cross-linked in the nematic phase; (4) a Frank elastic energy \cite{frank1958liquid} that recognizes that the mesogens prefer to be uniformly aligned and thus penalizes spatial variations in the director field. We perform an asymptotic analysis of this energy as the thickness tends to zero by choosing the moduli of these four terms to scale with the thickness as 
\begin{align*}
\underbrace{\text{incompressiblity penalty}}_{\text{and }\rightarrow \;\infty} \gg \underbrace{\text{entropic energy}}_{\text{and is } \sim \;1} \gg \underbrace{\text{non-ideal energy} \gg \text{Frank elasticity}}_{\text{and } \rightarrow \;0},
\end{align*}
which is well-motivated by the physics (see Sections \ref{ssec:Entropic}-\ref{ssec:Frank}). We follow the framework of $\Gamma$-convergence \cite{braides2002gamma} for the dimension reduction, building on the seminal work by Friesecke, James and M\"{u}ller on classical plates \cite{friesecke2002theorem} along with its generalization to non-Euclidean elasticity \cite{lewicka2023calculus} and incompressibility \cite{conti2009gamma,li2013kirchhoff}. A key difference in our work, as compared to these others, is that the current director is a field to be minimized just like deformation. This additional freedom introduces some nuance in the derivation, especially concerning the optimal thru-thickness behavior of the director field and its influence on the overall plate energy density. Another difference, as compared with the above-mentioned references, is that the limiting metric constraint in our setting need not be constant or smooth. 

In the end, we derive a plate theory (Eq.\;(\ref{eq:plateTheory}) below) that consists of the metric constraint of Aharoni \textit{et al.}\;\cite{aharoni2014geometry} along with a plate energy that depends on the programmed director field, the second fundamental form of the deformation, and the gradient of the deformed director field. We prove that this theory is the $\Gamma$-limit of our 3D energy for a wide class of physically relevant smooth director programs. We believe it is also the $\Gamma$-limit for all sufficiently smooth director programs, but our analysis falls just short of this type of result. 
 The gap is due to well-known mathematical hurdles (outlined in Section \ref{ssec:IncompressTech}) for constructing nearly incompressible 3D deformations that limit to metric-constrained midplane deformations with low regularity. Leaving aside this point, our work uses geometric rigidity \cite{friesecke2002theorem} and a 3D elastic model with well-established physical origins to derive a plate theory for the large actuation response of patterned LCE sheets.

The rest of the paper is organized as follows. Section \ref{sec:ModelMain} introduces the 3D model and states the main dimension reduction results, namely, Theorem \ref{MainTheorem} and \ref{RecoveryTheorem} and Corollary \ref{keyCor}. The corresponding proofs follow in Section \ref{sec:Compactness}-\ref{sec:Recovery}. We then conclude with Section \ref{sec:Examples} by highlighting examples of the plate theory. 

\section{Model and main results}\label{sec:ModelMain}

Let $\Omega := \omega \times (-1/2,1,2)$ for $\omega \subset \mathbb{R}^2$ a bounded and Lipschitz domain. We consider the variational problem of minimizing a suitable description of the elastic energy of a patterned LCE plate at the bending scale. For a given planar director field $\mathbf{n}_0 \in H^1(\omega, \mathbb{S}^1)$, we assume that the elastic energy to deform the body and director field by $\mathbf{y}_h \colon \Omega \rightarrow \mathbb{R}^3$ and $\mathbf{n}_h \colon \Omega \rightarrow \mathbb{S}^2$, respectively, is of the form 
\begin{equation}
\begin{aligned}\label{eq:overallEnergy}
E^h_{\mathbf{n}_0}(\mathbf{y}_{h}, \mathbf{n}_h) := \frac{1}{h^2} \int_{\Omega}\Big\{ W\big( (\boldsymbol{\ell}^{\text{f}}_{\mathbf{n}_h})^{-1/2}& \nabla_h \mathbf{y}_h (\boldsymbol{\ell}^{0}_{\mathbf{n}_0})^{1/2}\big) + \kappa_h \big( \det \nabla_h \mathbf{y}_h -1)^2 \\
&+ \mu_h \big| \mathbf{P}_{\mathbf{n}_0}( \nabla_h \mathbf{y}_h)^T \mathbf{n}_h\big|^2 + \gamma_h |\nabla_h \mathbf{n}_h|^2 \Big\} \dd V.
\end{aligned}
\end{equation}
In this formula, the energy has already been rescaled from the physical plate domain $\omega \times (-h/2,h/2)$ to the thickness independent domain $\Omega$, as is standard for deriving a plate theory using $\Gamma$-convergence. This rescaling replaces gradients of the deformation and director field on the physical domain with $\nabla_h \mathbf{y}_h := (\nabla \mathbf{y}_h , h^{-1}\partial_3 \mathbf{y}_h)$ and $\nabla_h \mathbf{n}_h := (\nabla \mathbf{n}_h, h^{-1} \partial_3 \mathbf{n}_h)$ on $\Omega$, respectively. Here and throughout, we consistently refer to $\nabla$ as the planar gradient, i.e., the gradient in $\mathbf{x} = (x_1,x_2) \in \omega$, in anticipation that the limiting energy will be defined on the midsurface $\omega$ (see Fig.\;\ref{fig:setup}). 
There are four distinct terms in this energy, detailed below.

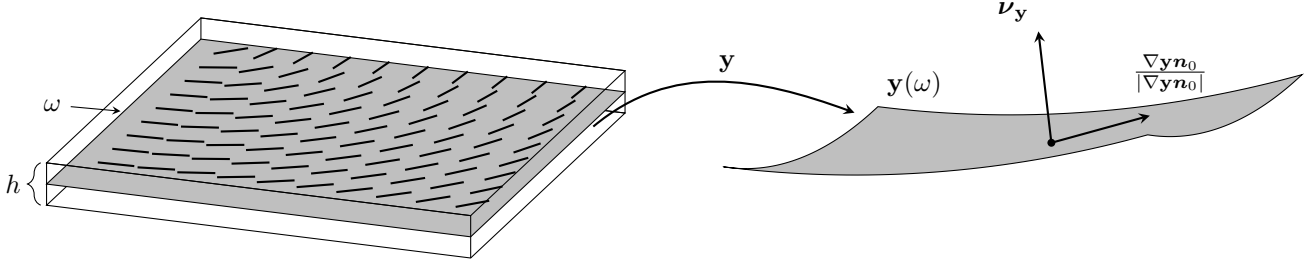
\begin{figure}[h!]
\tdplotsetmaincoords{70}{110}
\hspace*{-.6cm}
% scale to adjust size
\begin{tikzpicture}[tdplot_main_coords, scale=1.5]
    % Left figure
    \begin{scope}[xshift=-3cm]
        % Define variables
        \def\l{4} % length
        \def\w{4} % width
        \def\h{0.4} % thickness
        % Draw the plate
        \draw (0,0,0) -- (\l,0,0) -- (\l,\w,0) -- (0,\w,0) -- cycle;
        \draw (0,0,0) -- (0,\w,0) -- (0,\w,\h) -- (0,0,\h) -- cycle;
        % Draw the midplane rectangle
        \draw[solid, fill=gray!50] (0,0,\h/2) -- (\l,0,\h/2) -- (\l,\w,\h/2) -- (0,\w,\h/2) -- cycle;
        % Label the midplane with arrow
        \node[anchor=east] (omega) at (\l/2,-.5,\h/2) {$\omega$};
        \draw[-stealth] (omega) -- (\l/2,-.03,\h/2);
        
        \draw (0,0,\h) -- (\l,0,\h) -- (\l,\w,\h) -- (0,\w,\h) -- cycle;
        \draw (\l,0,0) -- (\l,0,\h) -- (\l,\w,\h) -- (\l,\w,0);
       
        % Add label with brace facing right
        \draw[decorate,decoration={brace,mirror,amplitude=5pt,raise=2pt}]
        (\l,-0.005,\h) -- (\l,-0.005,0) node [midway,left=6pt] {$h$};
        
        % Add line field (x,y,0)/sqrt(x^2+y^2) on midplane
        \foreach \x in {0.25,.75,...,4.5} {
            \foreach \y in {0.25,.75,...,5} {
                \pgfmathsetmacro{\vx}{-.45*\y/sqrt(\x*\x+\y*\y)}
                \pgfmathsetmacro{\vy}{.45*\x/sqrt(\x*\x+\y*\y)}
                \draw[black,thick,scale=.75] (\x,\y,\h/2) -- ++(\vx,\vy,0);
            }
        }
    \end{scope}
    
    % Arrow between figures
    \draw[-stealth, thick] (3,1.6,.6) to[out=40,in=160] (2,3.7,.65);
    \node[anchor=south] at (2,2.5,.8) {$\mathbf{y}$};
    % Right figure
    \begin{scope}[xshift=3cm]
        % Define variables
        \def\l{6} % length
        \def\w{6} % width
        \def\h{0.5} % thickness
        
        % Function to deform points
        \pgfmathdeclarefunction{deformX}{3}{%
            \pgfmathparse{#1}%
        }
        \pgfmathdeclarefunction{deformY}{3}{%
            \pgfmathparse{#2}%
        }
        \pgfmathdeclarefunction{deformZ}{3}{%
            \pgfmathparse{ 0.05 * (#1*#1) + .05*(#2*#2) - 1.75*(#3)}%
        }
        
        % Draw the deformed plate
        \draw[solid, fill=gray!50] 
            plot[domain=0:4,samples=20,smooth] (\x,0,{deformZ(\x,0,\h/2)}) --
            plot[domain=0:4,samples=20,smooth] ({deformX(4,\x,\h/2)},{deformY(4,\x,\h/2)},{deformZ(4,\x,\h/2)}) --
            plot[domain=4:0,samples=20,smooth] (\x,4,{deformZ(\x,4,\h/2)}) --
            plot[domain=4:0,samples=20,smooth] ({deformX(0,\x,\h/2)},{deformY(0,\x,\h/2)},{deformZ(0,\x,\h/2)}) -- cycle;
        % Label the midplane
        \node[anchor=south west] at (0,0,{deformZ(0,0,\h/2)}) {$\mathbf{y}(\omega)$};
        % Draw and label a unit normal vector
        \def\nx{-0.2}  % x-component of normal vector
        \def\ny{-0.2}  % y-component of normal vector
        \def\nz{1}     % z-component of normal vector
        \pgfmathsetmacro{\nvlen}{sqrt(\nx*\nx + \ny*\ny + \nz*\nz)}
        \pgfmathsetmacro{\nnx}{\nx/\nvlen}
        \pgfmathsetmacro{\nny}{\ny/\nvlen}
        \pgfmathsetmacro{\nnz}{\nz/\nvlen}
        
        % Add a dot at the base of the normal vector
        \fill (1,2,{deformZ(1,2,\h/2)}) circle (1pt);
        \draw[-stealth,black,thick] (1,2,{deformZ(1,2,\h/2)}) -- ++(\nnx,\nny,\nnz);
        \node[anchor=south east] at (1+\nnx,2+\nny,{deformZ(1,2,\h/2)+\nnz}) {$\boldsymbol{\nu}_\mathbf{y}$};
        
        % Draw and label a unit tangent vector
        \def\mx{0}  % x-component of normal vector
        \def\my{0.2}  % y-component of normal vector
        \def\mz{.08}     % z-component of normal vector
        \pgfmathsetmacro{\mvlen}{sqrt(\mx*\mx + \my*\my + \mz*\mz)}
        \pgfmathsetmacro{\mmx}{\mx/\mvlen}
        \pgfmathsetmacro{\mmy}{\my/\mvlen}
        \pgfmathsetmacro{\mmz}{\mz/\mvlen}
        
        \draw[-stealth,black,thick] (1,2,{deformZ(1,2,\h/2)}) -- ++(\mmx,\mmy,\mmz);
        \node[anchor=south west] at (1+\mmx,2+\mmy-.25,{deformZ(1,2,\h/2)+\mmz+.1}) {$\frac{\nabla\mathbf{y} \boldsymbol{n}_0}{|\nabla\mathbf{y} \boldsymbol{n}_0|}$};
    \end{scope}
\end{tikzpicture}
\caption{(Left): LCE plate with thickness $h$ with the midplane $\omega$ highlighted in grey. Director field in black lines is the blueprinted director $\mathbf{n}_0$. (Right): Deformation of midplane $\mathbf{y}(\omega)$ with normal vector to surface $\boldsymbol{\nu}_{\mathbf{y}}$. In the asymptotics, the director will be shown to convect with the deformation as indicated by the formula $\mathbf{n} = \frac{\nabla \mathbf{y} \mathbf{n}_0}{|\nabla \mathbf{y} \mathbf{n}_0|}$ shown above.}
\label{fig:setup}
\end{figure}
\subsection{Entropic elasticity}\label{ssec:Entropic} The first term is the purely entropic Warner-Terentjev \cite{bladon1994deformation} type energy, which is modeled by a smooth hyperelastic energy density $W(\mathbf{F})$ that satisfies 
\begin{equation}
\begin{aligned}\label{eq:hyperelastic}
(\text{frame indifference:})& \quad W(\mathbf{Q}\mathbf{F}) = W(\mathbf{F}) \quad \text{ for all } \mathbf{F} \in \mathbb{R}^{3\times3}, \mathbf{Q} \in SO(3), \\
(\text{isotropy:})& \quad W(\mathbf{F}\mathbf{Q}) = W(\mathbf{F}) \quad \text{ for all } \mathbf{F} \in \mathbb{R}^{3\times3}, \mathbf{Q} \in SO(3), \\
(\text{quadratic near $SO(3)$:})& \quad W(\mathbf{F}) \geq c_{\text{lb}} \text{dist}^2(\mathbf{F}, SO(3)) \quad \text{ for all } \mathbf{F} \in \mathbb{R}^{3\times3},
\end{aligned}
\end{equation}
for $\text{dist}(\mathbf{F}, SO(3)) := \inf_{\mathbf{R} \in SO(3)} |\mathbf{F}- \mathbf{R}|$. Being isotropic and smooth, the energy density expands in the classical way as 
\begin{equation}
\begin{aligned}\label{eq:TaylorW}
W(\mathbf{I} + \mathbf{A}) = \Big( \mu \big| \text{sym} \mathbf{A} \big|^2 + \frac{\lambda}{2}\big( \text{Tr}( \text{sym} \mathbf{A}) \big)^2\Big) + o(|\mathbf{A}|^2)
\end{aligned}
\end{equation}
for $\text{sym}\mathbf{A} = \frac{1}{2}( \mathbf{A} + \mathbf{A}^T)$ and for some Lam\'{e} parameters $\mu, \lambda >0$. ($\mu$ is also the shear modulus of the network.) The leading order term is, of course, the quadratic form of 3D linear elasticity, labeled herein as 
\begin{equation}
\begin{aligned}\label{eq:Q3A}
Q_3(\mathbf{A}) := 2\mu \big| \text{sym} \mathbf{A} \big|^2 + \lambda \big( \text{Tr}( \text{sym} \mathbf{A}) \big)^2, \quad \mathbf{A} \in \mathbb{R}^{3\times3}. 
\end{aligned}
\end{equation}

Nematic orientation is encoded into the argument of $W$ through the so-called step length tensors 
\begin{equation}
\begin{aligned}\label{eq:stepLengthsDef}
&\boldsymbol{\ell}_{\mathbf{v}}^{\text{f}} := \lambda_\text{f} \mathbf{v} \otimes \mathbf{v} + \lambda_\text{f}^{-1/2} \big( \mathbf{I} - \mathbf{v} \otimes \mathbf{v} \big) \in \mathbb{R}^{3\times3}_{\text{sym}}, \quad \mathbf{v} \in \mathbb{S}^2, \\
&\boldsymbol{\ell}_{\mathbf{v}_0}^{\text{0}} := \lambda_0 \begin{pmatrix} \mathbf{v}_0 \\ \mathbf{0} \end{pmatrix} \otimes \begin{pmatrix} \mathbf{v}_0 \\ \mathbf{0} \end{pmatrix} + \lambda_\text{0}^{-1/2} \big( \mathbf{I} - \begin{pmatrix} \mathbf{v}_0 \\ \mathbf{0} \end{pmatrix} \otimes \begin{pmatrix} \mathbf{v}_0 \\ \mathbf{0} \end{pmatrix} \big) \in \mathbb{R}^{3\times3}_{\text{sym}} , \quad \mathbf{v}_0 \in \mathbb{S}^1,
\end{aligned}
\end{equation}
for eigenvalues $\lambda_{\text{0}}, \lambda_{\text{f}} \geq 1$ that quantify the degree of nematic anisotropy in the reference and current (deformed) states of the LCE. Typically, these eigenvalues depend on the the temperature via $\lambda_{\text{f}} := \lambda(T_{\text{f}})$ and $\lambda_{\text{0}} := \lambda(T_0)$ for some function $\lambda(T) \geq 1$ that monotonically decreases in temperature $T$ and limits to $1$ as $T \rightarrow \infty$. They can also depend on other modes of actuation, like illumination or a magnetic field. Note that the eigenvectors associated to $\lambda_{\text{0}}, \lambda_{\text{f}}$ in these tensors form a line spanned by the director $\mathbf{v}$, while the transverse plane has eigenvalues $\lambda_0^{-1/2}, \lambda_{\text{f}}^{-1/2}$, respectively. Thus, $\det (\boldsymbol{\ell}_{\mathbf{v}}^{\text{f}}) = \det (\boldsymbol{\ell}_{\mathbf{v}}^{\text{0}}) = 1$. This completes our description of the the purely entropic part of the energy in (\ref{eq:overallEnergy}); we refer the interested reader to \cite{bladon1994deformation, warner2007liquid} for more on the physical origins of this term and some of its consequences.

\subsection{Approximate incompressibility} The second term in the energy in (\ref{eq:overallEnergy}) is an elastic penalty to deformations of the physical plate domain $\omega \times (-h/2,h/2)$ that are not incompressible, which becomes a penalty on $\det \nabla_h \mathbf{y}_h$ after rescaling (see the discussion below (\ref{eq:overallEnergy})). Since nematic elastomers are a soft polymer network, they are nearly incompressible. Thus, the the moduli of this energy term should be $\gg 1$ on physical grounds. We therefore consider the asymptotic regime for $\kappa_h$ in our analysis
\begin{equation}
\begin{aligned}\label{eq:kappaH}
\kappa_h > 0 \quad \text{ such that } \quad \kappa_h \rightarrow \infty, \quad h^{2s} \kappa_h \rightarrow 0 \quad \text{ as } h \rightarrow 0
\end{aligned}
\end{equation} 
for a suitably chosen $s \in (0,1]$. The first limit $\kappa_h \rightarrow \infty$ enforces an approximate incompressibility constraint. The second $h^{2s} \kappa_h \rightarrow 0$ is a scaling assumption that will lead to this energy term rigorously vanishing as $h \rightarrow 0$. 

That $\kappa_h$ should blow up slower than $h^{-2s}$, $s \in (0,1]$, is quite subtle, so we briefly offer some heuristics on these scalings. Recall from the characteristic polynomial on $\mathbb{R}^{3\times3}$ the useful identity
\begin{equation}
\begin{aligned} \label{eq:detIdent}
(\det ( \mathbf{I} + \mathbf{A}) -1 )^2 &= \left( \text{Tr}(\mathbf{A}) + ( \text{Tr}(\mathbf{A})^2 - \text{Tr}(\mathbf{A}^2)) + \det \mathbf{A}\right)^2 \\ 
&= \text{Tr}(\mathbf{A})^2 + o(|\mathbf{A}|^2) 
\end{aligned}
\end{equation}
Note also that $\text{Tr}(\mathbf{A})$, $\text{Tr}(\mathbf{A})^2 - \text{Tr}(\mathbf{A}^2)$, and $\det \mathbf{A}$ are homogeneous polynomials in the components of $\mathbf{A}$ of degree $p = 1,2$ and $3$, respectively, which means they can be bounded from above by $|\mathbf{A}|^p$ up to an constant. For instance, 
\begin{align}\label{eq:detEsts}
|\text{Tr}( \mathbf{A} )| \leq \sqrt{3} |\mathbf{A}|, \quad | \text{Tr}(\mathbf{A})^2 - \text{Tr}(\mathbf{A}^2)| \leq | \mathbf{A}|^2, \quad |\det \mathbf{A}| \leq \frac{1}{\sqrt{6}} |\mathbf{A}|^3
\end{align}
for any $\mathbf{A} \in \mathbb{R}^{3\times3}$. In constructing bending ansatz of the sheet, we will find that $\det \nabla_h \mathbf{y}_h = \det (\mathbf{I} + h \mathbf{A} + \boldsymbol{\delta} \mathbf{A}_h)$ for tensor fields $\mathbf{A}, \boldsymbol{\delta} \mathbf{A}_h \colon \Omega \rightarrow \mathbb{R}^{3\times3}$ that generically satisfy $\| h \mathbf{A}\|_{L^{\infty}(\Omega)} \sim \|\boldsymbol{\delta}\mathbf{A}_h\|_{L^{\infty}(\Omega)} \sim h$ and $ \int_{\Omega}|\boldsymbol{\delta} \mathbf{A}_h|^2 \dd V \ll \int_{\Omega}|h \mathbf{A}|^2 \dd V$. This bending ansatz also contains some DOFs that couple $\mathbf{A}$ and $\boldsymbol{\delta} \mathbf{A}_h$ in a delicate way, which motivate the scaling in (\ref{eq:kappaH}).
At leading order, the identity in (\ref{eq:detIdent}) furnishes $\int_{\Omega} h^{-2} \kappa_h (\det \nabla_h \mathbf{y}_h - 1)^2 \dd V = \kappa_h \int_{\Omega} \text{Tr}( \mathbf{A})^2 \dd V + \text{H.O.T.}$ Thus, our first assumption $\kappa_h \rightarrow \infty$ in (\ref{eq:kappaH}) leads to the constraint
$\text{Tr}(\mathbf{A})= 0$ for a finite energy limit. This constraint can be satisfied using the aforementioned DOFs. However, doing so leads to a remainder term of the form $\int_{\Omega} |\boldsymbol{\delta} \mathbf{A}_h|^2 \dd V = O(h^{2(s+1)})$ for some $s \in (0,1]$. The estimates in (\ref{eq:detEsts}) then allow us to conclude that $\int_{\Omega} h^{-2} \kappa_h (\det \nabla_h \mathbf{y}_h - 1)^2 \dd V = O(\kappa_h h^{-2} \int_{\Omega} |\boldsymbol{\delta} \mathbf{A}_h|^2 \dd V) = O( h^{2s}\kappa_h)$, which vanishes per the second assumption in (\ref{eq:kappaH}). The $s$-dependence in the scaling of the remainder is related to the smoothness or lackthereoff of the space of bending deformations of a plate; see Section \ref{ssec:FunctionSpaces}-\ref{ssec:IncompressTech} for further discussion. The detailed proof of this result found in Section \ref{ssec:SobolevSetting}.

\subsection{Director anchoring} The third term in (\ref{eq:overallEnergy}) is the so-called non-ideal energy \cite{biggins2012elasticity} that restricts the LCE from freely forming microstructure. The energy density employs the tensor 
\begin{align*}
\mathbf{P}_{\mathbf{v}_0} := \mathbf{I} - \begin{pmatrix} \mathbf{v}_0 \\ 0 \end{pmatrix} \otimes \begin{pmatrix} \mathbf{v}_0 \\ 0 \end{pmatrix} \in \mathbb{R}^{3\times3}_{\text{sym}} , \quad \mathbf{v}_0 \in \mathbb{S}^1,
\end{align*}
which projects vectors onto the tangent plane normal to $\mathbf{v}_0$. This energy seeks to anchor the director field according to the conventional transformation rule of reference to deformed normals in continuum mechanics. 
 
 In the physics literature, the terminology "nematic elastomer" versus "nematic glass" is often associated to the strength of the non-ideal energy relative to the purely entropic energy density (the first term in (\ref{eq:overallEnergy})). Nematic elastomers, being only lightly cross-linked, are capable of forming microstructure under a wide class of deformation (see \cite{cesana2015effective,desimone2002macroscopic} for effective theories on this microstructure and \cite{conti2002soft,plucinsky2017microstructure,verwey1996elastic} for a illustrative examples in the context of stretched sheets). In other words, the nematic director faces little obstruction to rotating through the polymer network. Nematic glasses however are more heavily cross-linked, which suppresses the free rotation of the director field and leads to a strong bias towards it being convected with the deformation. We are primarily interested in the nematic elastomers, due to their large actuation strains as compared to the much stiffer nematic glasses. We therefore assume that the moduli $\mu_h$ of this non-ideal term is much smaller than the shear modulus $\mu$ in (\ref{eq:Q3A}) through the scaling
 \begin{equation}
 \begin{aligned}\label{eq:muH}
\mu_h > 0 \quad \text{ such that } \quad \mu_h \rightarrow 0, \quad h^{-2} \mu_h \rightarrow \infty \quad \text{ as } h \rightarrow 0. 
 \end{aligned}
 \end{equation}
 The first limit above encodes that the purely entropic term dominates the non-ideal term. The second is needed in our analysis to ensure that the material cannot form microstructure at the bending scale. It will allow us to derive a metric constraint, relating the programmed director field $\mathbf{n}_0$ to the limiting midplane deformation. 
 
\subsection{Frank Elasticity}\label{ssec:Frank} The final term in (\ref{eq:overallEnergy}) is a source of elasticity penalizing deviations from a uniform director field. The most well-known such energy source is called Frank elasticity \cite{frank1958liquid}, a quadratic energy density penalizing splay, twist, bend and saddle-splay of the director field. The term employed here is called the "one constant" approximation of Frank elasticity, obtained by setting all the moduli of the four terms of Frank elasticity to be equal (and given $\gamma_h$ here). Frank elasticity will serve as a a regularizing term for our purposes in deriving a bending theory. Specifically, we assume that $\gamma_h$ is small in the sense that
 \begin{equation}
 \begin{aligned}\label{eq:gammaH}
h^{-2} \gamma_h\geq \gamma_{\text{lb}} > 0 \quad \text{ such that } \quad h^{-2} \gamma_h \rightarrow \gamma \quad \text{ as } h \rightarrow 0.
 \end{aligned}
 \end{equation}
While this $\gamma_h \sim h^2$ scaling may seem restrictive, it appears to be necessary to obtain compactness of sequences of deformations and directors that lead to a bending type theory at the bending scale. If Frank elasticity is too weak ($\gamma_h \ll h^2$), then the director can form microstructure and the standard tools of geometric rigidity fail to supply us with the needed compactness. If, alternatively, Frank elasticity is too strong ($\gamma_h \gg h^2$), then this term dominates the energy at the bending scale, forcing the director to be essentially uniform in the deformed configuration, even when the prescribed director $\mathbf{n}_0$ is heterogeneous. As the latter is inconsistent with what is typically observed experimentally, we can safely assume that strong Frank elasticity is not a realistic asymptotic regime. In summary, $\gamma_h \sim h^2$ is an interesting mathematical regime and physically reasonable.

\subsection{Function spaces}\label{ssec:FunctionSpaces} Having outlined the different terms in the energy and the $h$-dependent scalings of their moduli, we turn to address the appropriate function spaces for studying minimizers of this energy. Recall that the fields being minimized over are the deformation $\mathbf{y}_h$ and the current director $\mathbf{n}_h$; the reference director $\mathbf{n}_0$ is prescribed and assumed to be in $H^1(\omega, \mathbb{S}^1)$. We claim that
\begin{equation}
\begin{aligned}\label{eq:lbClaim}
\int_{\Omega} \Big\{ h^{-2} c_{\text{lb}}\big( c_1(\lambda_\text{f}, \lambda_0) \big|(\nabla \mathbf{y}_h, \partial_3 \mathbf{y}_h)\big|^2 -c_2\big) + \gamma_{\text{lb}} \big|(\nabla \mathbf{n}_h, \partial_3 \mathbf{n}_h)\big|^2 \Big\} \dd V \leq E^h_{\mathbf{n}_0}(\mathbf{y}_{h}, \mathbf{n}_h) \quad \text{ for all } h \in (0,1).
\end{aligned}
\end{equation}
It follows immediately that the energy is well defined as $h \rightarrow 0$ only if 
\begin{align*}
(\mathbf{y}_h, \mathbf{n}_h) \in H^1(\Omega, \mathbb{R}^3) \times H^1(\Omega,\mathbb{S}^2). 
\end{align*}
This observation makes $H^1(\Omega, \mathbb{R}^3) \times H^1(\Omega, \mathbb{S}^2)$ is the natural space of functions to investigate the limiting behavior of the energy $E_{\mathbf{n}_0}^h(\mathbf{y}_h, \mathbf{n}_h) $. 

To briefly address the claim in (\ref{eq:lbClaim}), observe that $\int_{\Omega} \gamma_{\text{lb}} \big|(\nabla \mathbf{n}_h, \partial_3 \mathbf{n}_h)\big|^2 \dd V\leq \int_{\Omega} h^{-2} \gamma_h |\nabla_h \mathbf{n}_h|^2 \dd V$ for all $h \in (0,1)$ due to the definition of $\nabla_h$ and the inequality in (\ref{eq:gammaH}). Also, the estimates in Lemma \ref{AppendLemma22} give that 
\begin{equation}
\begin{aligned}\label{eq:lbLemma} 
 \text{dist}^2\big((\boldsymbol{\ell}_{\mathbf{v}}^{\text{f}})^{-1/2} \mathbf{F} (\boldsymbol{\ell}_{\mathbf{v}_0}^0)^{1/2} , SO(3) \big) \geq c_1(\lambda_{\text{f}},\lambda_0) |\mathbf{F}|^2 - c_2 \quad \text{ for all } \mathbf{F} \in \mathbb{R}^{3\times 3}, \mathbf{v} \in \mathbb{S}^2, \mathbf{v}_0 \in \mathbb{S}^1
\end{aligned}
\end{equation}
for some $c_1(\lambda_{\text{f}},\lambda_0)> 0$ that depends only on $\lambda_{\text{f}}$ and $\lambda_0$. Thus, $\int_{\Omega} c_{\text{lb}} c(\lambda_\text{f}, \lambda_0) \big( \big|(\nabla \mathbf{y}_h, \partial_3 \mathbf{y}_h)\big|^2 -1\big) \dd V \leq \int_{\Omega} W\big( (\boldsymbol{\ell}^{\text{f}}_{\mathbf{n}_h})^{-1/2} \nabla_h \mathbf{y}_h (\boldsymbol{\ell}^{0}_{\mathbf{n}_0})^{1/2}\big) \dd V$ for all $h \in (0,1)$ using the inequalities in (\ref{eq:hyperelastic}) and (\ref{eq:lbLemma}). The claim in (\ref{eq:lbClaim}) follows since $h^{-2} \int_{\Omega} \big\{ W\big( (\boldsymbol{\ell}^{\text{f}}_{\mathbf{n}_h})^{-1/2} \nabla_h \mathbf{y}_h (\boldsymbol{\ell}^{0}_{\mathbf{n}_0})^{1/2}\big) + \gamma_h|\nabla_h \mathbf{n}_h|^2 \big\} \dd V \leq E_{\mathbf{n}_0}^h(\mathbf{y}_h, \mathbf{n}_h)$.

The class of deformations and director fields that have finite energy $E^h_{\mathbf{n}_0}(\mathbf{y}_{h}, \mathbf{n}_h)$ as $h \rightarrow 0$ is highly restrictive. In a prior work \cite{plucinsky2018actuation}, the last author and colleagues proved a compactness result that identified this class as
\begin{equation}
\begin{aligned}\label{eq:admissibleSet}
\mathcal{A}_{\mathbf{n}_0} := \Big\{ &(\mathbf{y}, \mathbf{n}) \in H^1(\Omega,\mathbb{R}^3) \times H^1(\Omega, \mathbb{S}^2) \text{ subject to } \\
&\qquad \mathbf{y} \text{ is independent of $x_3$, belongs to $ H^2(\omega, \mathbb{R}^3)$, and satisfies $(\nabla \mathbf{y})^T \nabla \mathbf{y} = \mathbf{g}_{\mathbf{n}_0}$ a.e.} , \\
&\qquad \text{$\mathbf{n}$ is constrained as } \mathbf{n} = \sigma \frac{\nabla \mathbf{y}\mathbf{n}_0}{|\nabla \mathbf{y}\mathbf{n}_0|} \text{ a.e. for a fixed constant $\sigma$ in $\{ -1, 1\}$} \Big\}. 
\end{aligned}
\end{equation}
for all $\mathbf{n}_0 \in H^1(\omega, \mathbb{S}^1)$. The term $(\nabla \mathbf{y})^T \nabla \mathbf{y} = \mathbf{g}_{\mathbf{n}_0}$ in this set defines a metric constraint concretely linking the programmed director $\mathbf{n}_0$ to the (midplane) deformation $\mathbf{y}$ through the tensor
\begin{equation}
\begin{aligned}\label{eq:metricDef}
\mathbf{g}_{\mathbf{v}_0} := \lambda_0^{-1} \lambda_{\text{f}} \mathbf{v}_0 \otimes \mathbf{v}_0 + \lambda_{\text{f}}^{-1/2}\lambda_0^{1/2} \mathbf{v}_0^{\perp} \otimes \mathbf{v}_0^{\perp} \in \mathbb{R}^{2\times 2}_{\text{sym}}, \quad \mathbf{v}_0 \in \mathbb{S}^1
\end{aligned}
\end{equation}
for $\mathbf{v}_0^{\perp} := \mathbf{R}(\pi/2) \mathbf{v}_0$, where $\mathbf{R}(\pi/2) \in SO(2)$ denotes a right-hand rotation by $\pi/2$. In addition, the constraint $\mathbf{n} = \sigma \frac{\nabla \mathbf{y}\mathbf{n}_0}{|\nabla \mathbf{y}\mathbf{n}_0|}$ in this set directly implies that that $\mathbf{n}$ belongs to $H^1(\Omega, \mathbb{S}^2)$ and is independent of $x_3$ (because $\mathbf{y}$ and $\mathbf{n}_0$ are both independent of $x_3$ with $\mathbf{y}$ in $H^2$ and $\mathbf{n}_0$ in $H^1$). 

As a final point on function spaces, we will have need to introduce fractional Sobolev spaces in the course of constructing recovery sequences of the energy $E_{\mathbf{n}_0}^{h}$ for large classes of limiting fields that belong to $\mathcal{A}_{\mathbf{n}_0}$. Let $V$ denote a normed vector space. 
For $s\in (0,1)$,  the fractional Sobolev space $H^s(\omega,V)$ is defined by  
\begin{align}\label{eq:HsDef}
H^s(\omega, V) := \{ f \in L^2(\omega, V) \colon \| f \|_{H^{s}(\omega,V)}  < \infty \},
\end{align}
where 
\begin{align*}
    \left\| f \right\|_{ H^{s}(\omega,V)} := \sqrt{\iint_{\omega \times \omega} \frac{\left| f(\mathbf{x}) - f(\tilde{\mathbf{x}}) \right|^{2}}{\left| \mathbf{x} - \tilde{\mathbf{x}}  \right|^{2+2s}}  \dd  A \dd \tilde{A} }.
\end{align*}
We will also use $H^{s}(\Omega, V)$. This space is defined by replacing $\omega$ with $\Omega$ in (\ref{eq:HsDef}) and modifying the norm to reflect that $\Omega$ is in $\mathbb{R}^3$ rather than $\mathbb{R}^2$, i.e., via $\| f\|_{H^{s}(\Omega,V)} := \sqrt{\iint_{\Omega \times \Omega} \frac{\left| f(\mathbf{x},x_3) - f(\tilde{\mathbf{x}},\tilde{x}_3) \right|^{2}}{\left| (\mathbf{x},x_3) - (\tilde{\mathbf{x}},\tilde{x}_3)  \right|^{3+2s}}  \dd  V \dd \tilde{V} }$. For completeness, in the case $s = 1$, we define $H^s(\omega, V)$ to be the standard Hilbert space of square-integrable functions with a square-integrable gradient.

\subsection{Main result} In this paper, we establish the following plate theory, 
\begin{align}\label{eq:plateTheory}
E_{\mathbf{n}_0}(\mathbf{y}, \mathbf{n}) := \begin{cases}
\int_{\omega}\big\{ \mathbf{II}_{\mathbf{y}} \colon \mathbb{B}(\mathbf{n}_0) \colon \mathbf{II}_{\mathbf{y}} + \gamma \lambda_{\text{f}}^{-1} \lambda_0 \big|\nabla\big( \nabla \mathbf{y} \mathbf{n}_0\big) \big|^2 \big\} \dd A& \text{ if } (\mathbf{y}, \mathbf{n}) \in \mathcal{A}_{\mathbf{n}_0} \\ 
+ \infty & \text{ otherwise},
\end{cases} 
\end{align} 
as the limit of the energy $E_{\mathbf{n}_0}^h$ as $h \rightarrow 0$
under any $h$-dependent asymptotic scaling of the moduli in (\ref{eq:kappaH}), (\ref{eq:muH}) and (\ref{eq:gammaH}) and assuming sufficient smoothness of the limiting fields. Much like Kirchhoff's plate theory, the theory in (\ref{eq:plateTheory}) consists of a metric constraint --- the one stated above --- and an energy density that is quadratic in the second fundamental form of the deformed midplane, a nonlinear strain measure on midplane deformations $\mathbf{y} \colon \omega \rightarrow \mathbb{R}^3$ defined by 
\begin{equation}
\begin{aligned}\label{eq:secFund}
\mathbf{II}_{\mathbf{y}} := (\nabla \mathbf{y})^T \nabla \boldsymbol{\nu}_{\mathbf{y}} \qquad \text{ for } \quad \boldsymbol{\nu}_{\mathbf{y}} := \frac{\partial_1 \mathbf{y} \times \partial_2 \mathbf{y}}{|\partial_1 \mathbf{y} \times \partial_2 \mathbf{y}|}.
\end{aligned}
\end{equation}
 Unlike Kirchhoff's theory, the moduli of the quadratic form is highly anisotropic; it depends explicitly on the programmed director field via the fourth-order tensor $\mathbb{B}(\mathbf{v}_0) \in \mathbb{R}^{2\times2 \times 2 \times 2}$ defined by 
\begin{equation}
\begin{aligned}\label{eq:Bv0Moduli}
\mathbb{B}(\mathbf{v}_0) := & \mu_1 \big( \mathbf{v}_0 \otimes \mathbf{v}_0 \otimes \mathbf{v}_0 \otimes \mathbf{v}_0 \big)+ \mu_2 \big( \mathbf{v}^{\perp}_0 \otimes \mathbf{v}^{\perp} _0 \otimes \mathbf{v}^{\perp}_0 \otimes \mathbf{v}^{\perp}_0 \big) \\
&\quad + \big( \sqrt{\mu}_1 \mathbf{v}_0 \otimes \mathbf{v}_0 + \sqrt{\mu_2} \mathbf{v}_0^{\perp} \otimes \mathbf{v}_0^{\perp} \big) \otimes \big( \sqrt{\mu}_1 \mathbf{v}_0 \otimes \mathbf{v}_0 + \sqrt{\mu_2} \mathbf{v}_0^{\perp} \otimes \mathbf{v}_0^{\perp} \big) \\
&\qquad + \mu_3 \Big( \text{sym}( \mathbf{v}_0 \otimes \mathbf{v}_0^{\perp}) \otimes \text{sym}( \mathbf{v}_0 \otimes \mathbf{v}_0^{\perp}) \Big) 
\end{aligned}
\end{equation}
for any $\mathbf{v}_0 \in \mathbb{S}^1$ and for the moduli
\begin{align*}
&\mu_1 := \frac{\mu}{12} \lambda_{\text{f}}^{-5/2} \lambda_0^{5/2} , \quad \mu_2 := \frac{\mu}{12} \lambda_{\text{f}}^{1/2} \lambda_{0}^{-1/2}, \quad \mu_3 := \frac{\mu}{2} \big( \alpha^{-2} - 2\alpha^{-3} \tanh(\tfrac{\alpha}{2})\big) \lambda_0 (\lambda_{\text{f}}^{-5/4} + \lambda_{\text{f}}^{1/4} )^2.
\end{align*}
The $\alpha$ dependence in the last moduli supplies a nontrivial coupling of Frank and entropic elasticity in the limit via 
\begin{align*}
\alpha = \Big( \frac{\mu}{2\gamma}\Big)^{1/2} (\lambda_{\text{f}}^{3/4} - \lambda_{\text{f}}^{-3/4}).
\end{align*}

We justify the limiting plate theory in (\ref{eq:plateTheory}) by two theorems, which together amount to a $\Gamma$-convergence result up to a well-known technical issue in the mathematical literature concerning incompressible or nearly incompressible plates (see Section \ref{ssec:IncompressTech}). The first theorem follows the $\Gamma$-convergence formalism by establishing a compactness result for sequences of deformations and director fields with bounded energy $E_{\mathbf{n}_0}^h$ as $h \rightarrow 0$ and by showing that $E_{\mathbf{n}_0}$ is an ansatz free lowerbound of the energy $E_{\mathbf{n}_0}^h$.
\begin{thm}\label{MainTheorem}
Let $\mathbf{n}_0 \in H^1(\omega, \mathbb{S}^1)$ and assume the least restrictive hypothesis on the scaling of $\kappa_h$ in (\ref{eq:kappaH}) of $s =1$. The energies $E_{\mathbf{n}_0}^h(\mathbf{y}_h,\mathbf{n}_h)$ and $E_{\mathbf{n}_0}(\mathbf{y},\mathbf{n})$ have the following properties:
\begin{itemize}[leftmargin=*]
\item \textbf{\emph{Compactness}.}\;For every sequence $\{ (\mathbf{y}_h, \mathbf{n}_h)\} \subset H^1(\Omega, \mathbb{R}^3) \times H^1(\Omega, \mathbb{S}^2)$ that satisfies $\liminf_{h \rightarrow 0} E_{\mathbf{n}_0}^h(\mathbf{y}_h, \mathbf{n}_h) < \infty$, there is a subsequence (not relabeled) such that 
\begin{align*}
\mathbf{y}_h - \frac{1}{|\Omega|} \int_{\Omega} \mathbf{y}_h \dd V \rightarrow \mathbf{y} \quad \text{ in } H^1(\Omega, \mathbb{R}^3) \quad \text{ and } \quad \mathbf{n}_h \rightharpoonup \mathbf{n} \quad \text{ in } H^1(\Omega, \mathbb{S}^2)
\end{align*}
for limiting fields that satisfy $(\mathbf{y}, \mathbf{n}) \in \mathcal{A}_{\mathbf{n}_0}$. 
\item \textbf{\emph{Lowerbound}}.\;For every sequence $\{ (\mathbf{y}_h, \mathbf{n}_h)\} \subset H^1(\Omega, \mathbb{R}^3) \times H^1(\Omega, \mathbb{S}^2)$ such that $(\mathbf{y}_h , \mathbf{n}_h) \rightharpoonup (\mathbf{y}, \mathbf{n})$ in $H^1(\Omega, \mathbb{R}^3) \times H^1(\Omega, \mathbb{S}^2)$, 
\begin{align*}
\liminf_{h \rightarrow 0} E_{\mathbf{n}_0}^h(\mathbf{y}_h, \mathbf{n}_h) \geq E_{\mathbf{n}_0}(\mathbf{y}, \mathbf{n}) .
\end{align*}
\end{itemize}
\end{thm}
\noindent The second theorem provides a recovery sequence that limits to the plate energy $E_{\mathbf{n}_0}$ under two assumptions: 1) the second fundamental form the limiting deformation $\mathbf{II}_{\mathbf{y}}$ enjoys more regularity than that required by the compactness; 2) the incompressibility penalty $\kappa_h$ tends to infinity at a slow enough rate, tied in a precise way to the regularity of $\mathbf{II}_{\mathbf{y}}$. 
\begin{thm}[{\textbf{Recovery Sequence}}]\label{RecoveryTheorem}
Let $\mathbf{n}_0 \in H^1(\omega, \mathbb{S}^1)$ and assume for some $s \in (0,1]$ that $(\mathbf{y}, \mathbf{n}) \in \mathcal{A}_{\mathbf{n}_0}$ is such that $ \mathbf{II}_{\mathbf{y}} \in L^{\infty}(\Omega, \mathbb{R}_{\emph{sym}}^{2 \times 2}) \cap H^{s}(\Omega, \mathbb{R}_{\emph{sym}}^{2 \times 2})$ and $\kappa_h$ satisfies (\ref{eq:kappaH}) for this value of $s$. Then there exists a sequence $\{ (\mathbf{y}_h, \mathbf{n}_h)\} \subset H^1(\Omega, \mathbb{R}^3) \times H^1(\Omega, \mathbb{S}^2)$ such that $(\mathbf{y}_h , \mathbf{n}_h) \rightarrow (\mathbf{y}, \mathbf{n})$ in $H^1(\Omega, \mathbb{R}^3) \times H^1(\Omega, \mathbb{S}^2)$ and 
\begin{align*}
\lim_{h \rightarrow 0} E_{\mathbf{n}_0}^h(\mathbf{y}_h, \mathbf{n}_h) = E_{\mathbf{n}_0}(\mathbf{y}, \mathbf{n}). 
\end{align*}
\end{thm}

 We close this section by stating and discussing a corollary of these two theorems that precisely illustrates the gap between our result and a full $\Gamma$-convergence. Let
\begin{align*}
\mathcal{A}_{\mathbf{n}_0}^s := \{ (\mathbf{y}, \mathbf{n}) \in \mathcal{A}_{\mathbf{n}_0} \colon \mathbf{II}_{\mathbf{y}} \in L^{\infty}(\Omega, \mathbb{R}_{\text{sym}}^{2 \times 2}) \cap H^{s}(\Omega, \mathbb{R}_{\text{sym}}^{2 \times 2}) \}, \quad s \in (0, 1]. 
\end{align*}
denote the space of admissible bending deformations of the LCE sheet that posses $L^{\infty} \cap H^s$ regularity in their second fundamental form. 
\begin{cor}\label{keyCor}
 Suppose $\mathbf{n}_0 \in H^1(\omega, \mathbb{S}^2)$ is such that $\mathcal{A}_{\mathbf{n}_0}^s = \mathcal{A}_{\mathbf{n}_0}$ for some $s \in (0,1]$. Assume also 
 that $\kappa_h$ satisfies (\ref{eq:kappaH}) for this value of $s$. Then, in the weak-$H^1(\Omega, \mathbb{R}^3) \times H^1(\Omega, \mathbb{S}^2)$ topology, $E_{\mathbf{n}_0}^h$ is equicoercive and $\Gamma$-converges as $h \rightarrow 0$ to $E_{\mathbf{n}_0}$. 
\end{cor}
\noindent The assumption $\mathcal{A}_{\mathbf{n}_0}^s = \mathcal{A}_{\mathbf{n}_0}$ may appear restrictive, but metric constraints are often quite rigid. Hornung and Vil$\check{\text{c}}$i\'c \cite{hornung2018regularity}, for instance, proved a regularity result on the space $H^2_{\mathbf{g}}(\omega, \mathbb{R}^3) := \{ \mathbf{y} \in H^2(\omega, \mathbb{R}^3) \colon (\nabla \mathbf{y})^T \nabla \mathbf{y} = \mathbf{g} \text{ a.e.\;on } \omega \}$ that says the following: if $\mathbf{g}$ belongs to $ C^{\infty}(\omega, \mathbb{R}^{2\times2}_{\text{sym}})$ and has positive Gauss curvature everywhere, then any $\mathbf{y} \in H^2_{\mathbf{g}}(\omega, \mathbb{R}^3)$ is actually smooth. In Section \ref{sec:Examples}, we illustrate a canonical example from the literature \cite{aharoni2014geometry,mostajeran2015curvature} of a smooth director profile $\mathbf{n}_0$ for which $\mathbf{g}_{\mathbf{n}_0}$ has constant positive Gauss curvature and yields a spherical cap on actuation. The corollary evidently applies for this particular director profile and actuation thanks to the results of Hornung and Vil$\check{\text{c}}$i\'c. More generally: 
\begin{rem}
 If $\mathbf{\mathbf{n}}_0 \in C^{\infty}(\overline{\omega}, \mathbb{S}^2)$ is such that $\mathbf{g}_{\mathbf{n}_0}$ has positive Gauss curvature everywhere on $\omega$, then $\mathcal{A}_{\mathbf{n}_0} = \mathcal{A}_{\mathbf{n}_0}^1$. Consequently, $E_{\mathbf{n}_0}$ is the $\Gamma$-limit of $E_{\mathbf{n}_0}^h$ under the least restrictive hypothesis on $\kappa_h$ in (\ref{eq:kappaH}) of $s =1$.
\end{rem}
\noindent This remark is broadly applicable; indeed, several families of director profiles that encode positive Gauss curvature beyond the spherical cap have been illustrated in the LCE literature \cite{aharoni2014geometry,aharoni2018universal,modes2011gaussian, modes2011blueprinting, mostajeran2015curvature, mostajeran2016encoding,warner2018nematic}. However, that same literature also highlights examples of metric-constrained actuation that involve negative, zero, and/or mixed Gauss curvature, none of which is amenable to any regularity result that we are aware of. Thus more is needed to complete picture of the plate theories for LCE sheets than the results we can prove to this point.

\subsection{Technical issues of incompressibility.}\label{ssec:IncompressTech} The crux of the matter is incompressiblity. Here we explain the issue in the simple setting of a classical isoptropic and nearly incompressible plate of energy $E_{\text{class}}^h(\mathbf{y}_h) := \frac{1}{h^2}\int_{\Omega}\{ W(\nabla_h \mathbf{y}_h) + \kappa_h ( \det \nabla_h \mathbf{y}_h - 1)^2 \}\dd V$. (There are additional extraneous details to track in the case of LCE sheets, but we focus on the main technical hurdle here.) The basic idea for a recovery sequence in this setting is to set 
\begin{align}\label{eq:classicalAnsatz}
 \mathbf{y}_h(\mathbf{x}, x_3) = \mathbf{y}(\mathbf{x}) + \big(h x_3 + h^2 \frac{x_3^2}{2} \alpha_h(\mathbf{x}) \big) \boldsymbol{\nu}_{\mathbf{y}}(\mathbf{x})
\end{align}
where $\mathbf{y}$ is an isometry (satisfies $(\nabla \mathbf{y})^T \nabla \mathbf{y} = \mathbf{I}$ a.e.), $\boldsymbol{\nu}_{\mathbf{y}}$ is the surface normal in (\ref{eq:secFund}), and $\alpha_h$ is an auxiliary field carefully chosen to optimize the limit of $ \frac{1}{h^2} \int_{\Omega} W(\nabla_h \mathbf{y}_h) \dd V$ while ensuring that the incompressibility penalty vanishes. (Equation (\ref{eq:classicalAnsatz}) is a variant on the modified Kirchhoff-Love ansatz discussed in detail in \cite{ozenda2021kirchhoff}.) Note that $(\nabla \mathbf{y}, \boldsymbol{\nu}_{\mathbf{y}})$ is a rotation field, so this ansatz satisfies 
\begin{align*}
(\nabla \mathbf{y}, \boldsymbol{\nu}_{\mathbf{y}})^T \nabla_h \mathbf{y}_h = \mathbf{I} + h x_3 \begin{pmatrix} (1+ \tfrac{1}{2} hx_3 \alpha_h )\mathbf{II}_{\mathbf{y}} & \mathbf{0} \\ \tfrac{1}{2} h x_3 (\nabla \alpha_h)^T & \alpha_h \end{pmatrix}.
\end{align*}

We now use frame-indifference and the approximations $W(\mathbf{I} + \mathbf{A}) \approx \tfrac{1}{2} Q_3(\mathbf{A})$ from (\ref{eq:TaylorW}) and $(\det (\mathbf{I} + \mathbf{A}) - 1)^2 \approx \text{Tr}(\mathbf{A})^2$ from (\ref{eq:detIdent}) to formally expand the total energy. Specifically,
\begin{equation}
\begin{aligned}\label{eq:simpleExpansion}
E^h_{\text{class}}(\mathbf{y}_h) &= \frac{1}{h^2} \int_{\Omega} \big\{ W\big(( \nabla \mathbf{y}, \boldsymbol{\nu}_{\mathbf{y}})^T \nabla_h \mathbf{y}_h\big) + \kappa_h ( \det\big( (\nabla \mathbf{y}, \boldsymbol{\nu}_{\mathbf{y}})^T \nabla_h \mathbf{y}_h \big) - 1)^2 \big\} \dd V \\
 & \approx \int_{\Omega} x_3^2 \Big\{\frac{1}{2} Q_3 \Big( \begin{pmatrix} (1+ \tfrac{1}{2} hx_3 \alpha_h )\mathbf{II}_{\mathbf{y}} & \mathbf{0} \\ \tfrac{1}{2} h x_3 (\nabla \alpha_h)^T & \alpha_h \end{pmatrix} \Big) + \kappa_h \text{Tr}\Big(\begin{pmatrix} (1+ \tfrac{1}{2} hx_3 \alpha_h )\mathbf{II}_{\mathbf{y}} & \mathbf{0} \\ \mathbf{0}^T & \alpha_h \end{pmatrix} \Big)^2 \Big\} \dd V \\
 &= \int_{\omega}\frac{1}{12} \Big\{\frac{1}{2} Q_3 \Big(\begin{pmatrix} \mathbf{II}_{\mathbf{y}} & \mathbf{0} \\ \mathbf{0}^T & \alpha_h \end{pmatrix} \Big) + \kappa_h \big(\alpha_h + \text{Tr}(\mathbf{II}_{\mathbf{y}}) \big)^2 \Big\} \dd V \\
 &\qquad \qquad + O(h^2 \| \nabla \alpha_h \|_{L^2(\omega)}^2) + O(h^2(1 + \kappa_h) \| \alpha_h \mathbf{II}_{\mathbf{y}} \|^2_{L^2(\omega)}) .
 \end{aligned}
 \end{equation}
 The main point is that the incompressibility penalty demands that $\alpha_h$ approximate $-\text{Tr}(\mathbf{II}_{\mathbf{y}})$, a term that need not be very regular. At the same time, $\nabla \alpha_h$ cannot blow up too fast. These two features appear to be difficult to ensure in general without employing some non-generic regularity result about metric-constrained deformations. We discuss this in more detail. 
 
First note that the lack of regularity in $\mathbf{II}_{\mathbf{y}}$ is by itself not really the problem provided there is no incompressibiltiy penalty. Indeed set $\kappa_h =0$ above and suppose $\mathbf{II}_{\mathbf{y}}$ is only in $L^2(\omega)$. The idea for a recovery sequence then, as applied in Friesecke, James and M\"{u}ller \cite{friesecke2002theorem}, is to freeze $\alpha_h = \alpha$ to be a smooth and compactly supported function, pass to the limit in $h$, and then argue by density that $\alpha = \alpha_j$ can chosen such that the limit energy is $1/j$ close to the optimal one. A suitable diagonal sequence $\{(h_j, \alpha_{h_j})\}$, $j \rightarrow \infty$, produces the desired recovery sequence. 

However, the presence of a $\kappa_h >0$ that $\rightarrow \infty$ means that we can only freeze $\alpha_h = \alpha$ and pass to the limit if $\alpha = -\text{Tr}(\mathbf{II}_{\mathbf{y}})$, which necessitates some regularity in the deformation. This idea works, for instance, when $\mathbf{y} \in C^{3}(\overline{\omega}, \mathbb{R}^3)$. In this case, the choice $\alpha_h = -\text{Tr}(\mathbf{II}_{\mathbf{y}})$ for all $h$ allows us to pass from $Q_3$ to the optimal energy density for a nearly incompressible plate, while eliminating the leading order term in $\kappa_h$. The remainder terms then scale as $O(\kappa_h h^2)$ and vanish as $h \rightarrow 0$ by (\ref{eq:kappaH}) since $\nabla \alpha_h = - \nabla \text{Tr}(\mathbf{II}_{\mathbf{y}})$ and $\alpha_h \mathbf{II}_{\mathbf{y}} = - \text{Tr}(\mathbf{II}_{\mathbf{y}}) \mathbf{II}_{\mathbf{y}}$ are square integrable. Interestingly, in this $\mathbf{y} \in C^3(\overline{\omega})$ setting, it is even possible to solve the incompressibility constraint exactly. The idea, which comes from Conti and Dolzmann \cite{conti2006derivation,conti2009gamma}, is to replace the expression $hx_3 + \tfrac{1}{2}h^2 \alpha_h(\mathbf{x})$ in (\ref{eq:classicalAnsatz}) with a general function $\varphi_h(\mathbf{x}, x_3)$. Then, the constraint $\det \nabla_h \mathbf{y}_h = 1$ is equivalent to an ODE for $\varphi_h$ in the $x_3$-coordinate. By the smoothness of $\mathbf{y}$, a solution to the ODE is guaranteed to exist satisfy the plate ansatz $\varphi_h(\mathbf{x},x_3) \approx h x_3$ for all sufficiently small $h$. While this construction is far from generic on account of its smoothness, it turns out (see \cite{hornung2011approximation,pakzad2004sobolev}) that any deformation in the space of classical bending isometries $H^2_{\text{iso}}(\omega, \mathbb{R}^3) := \{ \mathbf{y} \in H^2(\omega, \mathbb{R}^3) \colon (\nabla \mathbf{y})^T \nabla \mathbf{y} = \mathbf{I} \text{ a.e.}\}$ is suitably approximated by 
% $C^3(\overline{\omega})$-isometries.
smooth isometries.
Thus, Conti and Dolzmann \cite{conti2009gamma} used it to prove a full $\Gamma$-convergence result for classical incompressible plates. It can also be adapted to incompressible plate theories in a more general non-Euclidean setting, provided an analogous density result exists for that setting \cite{li2013kirchhoff}.

Here, unfortunately, we think it unlikely that such density results hold for all physically relevant metrics of LCE sheets --- there is simply too much freedom in the design of a the director profile, leading to a rich variety of actuated surfaces. Instead, we must confront energy expansions, like in (\ref{eq:simpleExpansion}), without invoking any specialized regularity results on the metric-constrained actuation. Our workaround is to assume some mild additional regularity on the second fundamental form, namely, that $\mathbf{II}_{\mathbf{y}} \in L^{\infty}(\omega, \mathbb{R}^{3\times3}_{\text{sym}}) \times H^s(\omega, \mathbb{R}^{3\times3}_{\text{sym}})$ for some $s \in (0,1].$ In this setting, it is possible to pass to the desired limit in (\ref{eq:simpleExpansion}) by choosing $\alpha_h = -\text{Tr}(\mathbf{II}_{\mathbf{y}}) \ast m_h$ for a standard mollifier $m_h(\mathbf{x}) = h^{-2} m(h^{-1} \mathbf{x})$. This choice leads to estimates of the form (see Exercise 6.64 \cite{leoni2023first}) 
\begin{align*}
\| \alpha_h + \text{Tr}(\mathbf{II}_{\mathbf{y}})\|_{L^2(\omega)}^2 = O(h^{2s} \| \mathbf{II}_{\mathbf{y}}\|^2_{H^s(\omega)}), \quad \| \nabla \alpha_h\|_{L^2(\omega)}^2 = O(h^{2(s-1)} \| \mathbf{II}_{\mathbf{y}}\|^2_{H^s(\omega)}).
\end{align*}
One can then check that the leading order error term in (\ref{eq:simpleExpansion}) scales as $O(\kappa_h h^{2s})$ and thus vanishes as $h \rightarrow 0$ under our assumptions on $\kappa_h$ in (\ref{eq:kappaH}). Essentially the same ideas apply in the LCE setting, which is how we are lead to the assumptions and the proof of Theorem \ref{RecoveryTheorem}. Notice that the remainder term $O(\kappa_h h^{2s})$ is sharp, which suggests that the results cannot easily be improved upon, except perhaps on a case-by-case basis where one can establish some specific regularity properties on the metric-constrained deformations of interest.

\section{Compactness}\label{sec:Compactness} 

This section supplies a proof of the compactness statement in Theorem \ref{MainTheorem}, which is based on geometric rigidity \cite{friesecke2002theorem} and its generalizations to non-Euclidean plate theory \cite{lewicka2011scaling}. 

\subsection{Geometric rigidity for LCE sheets.} We begin with a key inequality that allows us to extract a limiting rotation field $\mathbf{R} \in H^1(\omega, SO(3))$ for sequences with bounded energy at the bending scale. 
\begin{lem}\label{keyCompactLemma}
There is a constant $C = C(\omega, \lambda_\text{f}, \lambda_0) >0$ with the following properties: For every $h >0$, $\mathbf{y}_h \in H^1(\Omega, \mathbb{R}^3)$, $\mathbf{n}_h \in H^1(\Omega, \mathbb{S}^2)$, and $\mathbf{n}_0 \in H^1(\omega, \mathbb{S}^1)$, there is a matrix field $\mathbf{G}_h \colon \omega \rightarrow \mathbb{R}^{3\times3}$ satisfying the estimates
\begin{align*}
&\int_{\Omega} \big| \mathbf{G}_h - (\boldsymbol{\ell}^{\text{f}}_{\mathbf{n}_h})^{-1/2} \nabla_h \mathbf{y}_h (\boldsymbol{\ell}^{0}_{\mathbf{n}_0})^{1/2}\big|^2 \dd V + h^2 \int_{\omega} \big|\nabla \mathbf{G}_h \big|^2 \dd A \\
&\qquad \leq C \int_{\Omega}\big\{ \emph{dist}^2\big( (\boldsymbol{\ell}^{\text{f}}_{\mathbf{n}_h})^{-1/2} \nabla_h \mathbf{y}_h (\boldsymbol{\ell}^{0}_{\mathbf{n}_0})^{1/2} ,SO(3)\big) + h^2( \big| \nabla_h \mathbf{n}_h\big|^2 + |\nabla \mathbf{n}_0|^2 \big) \big\} \dd V.
\end{align*} 
\end{lem}
\noindent An estimate of this type is now more-or-less standard fare in the literature on plate theories derived by $\Gamma$-convergence \cite{friesecke2002theorem}. The one stated here is essentially a repackaging of the statement in \cite[Proposition C.1.]{plucinsky2018actuation} and can be proved using the same argument (in Appendix C there). We do not repeat the proof.

\subsection{Proof of compactness} 
Let $\{ (\mathbf{y}_h, \mathbf{n}_h)\} \subset H^1(\Omega, \mathbb{R}^3) \times H^1(\Omega, \mathbb{S}^2)$ be a sequence with bounded energy in the sense that $\liminf_{h \rightarrow 0} E_{\mathbf{n}_0}^h(\mathbf{y}_h, \mathbf{n}_h) <\infty$. Then there is a subsequence (not relabeled) such that 
\begin{equation}
\begin{aligned}\label{eq:boundedEnergy}
E_{\mathbf{n}_0}^h(\mathbf{y}_h, \mathbf{n}_h) \leq M 
\end{aligned}
\end{equation}
for some $M > 0$ independent of $h$. We now establish a series of convergence properties based on this boundedness, all of which culminate in a proof of the compactness statements in Theorem \ref{MainTheorem}. All convergences stated in the lemmas below are for a suitably chosen subsequence as $h \rightarrow 0$. 

\subsubsection{Weak convergence of $(\mathbf{y}_h,\mathbf{n}_h)$}
Our first result concerns convergence properties solely of the sequence of director fields in this energy.

\begin{lem}\label{weakConLemmaN} 
$\mathbf{n}_h \rightharpoonup \mathbf{n}$ in $H^1(\Omega, \mathbb{S}^2)$ for some midplane director $\mathbf{n}$ independent of $x_3$; $h^{-1} \partial_3 \mathbf{n}_h \rightharpoonup \boldsymbol{\tau}$ in $L^2(\Omega, \mathbb{R}^3)$ for some vector field $\boldsymbol{\tau}$ that satisfies $\boldsymbol{\tau} \cdot \mathbf{n} = 0$ a.e.; $(\boldsymbol{\ell}_{\mathbf{n}_h}^{\emph{f}})^{\pm1/2} \rightarrow (\boldsymbol{\ell}_{\mathbf{n}}^{\emph{f}})^{\pm1/2}$ in $L^2(\Omega, \mathbb{R}^{3\times3})$.
\end{lem}
\begin{proof}
Due to the Frank elastic term in the energy, we have for all $h \in (0,1)$ that 
\begin{align*}
\int_{\Omega}\big\{ |\nabla \mathbf{n}_h|^2 + |\partial_3 \mathbf{n}_h|^2 \big\} \dd V \leq \int_{\Omega}\big\{ |\nabla \mathbf{n}_h|^2 + |h^{-1} \partial_3 \mathbf{n}_h|^2 \big\} \dd V \leq \gamma_{\text{lb}}^{-1} E_{\mathbf{n}_0}^h(\mathbf{y}_h, \mathbf{n}_h) \leq \gamma_{\text{lb}}^{-1} M.
\end{align*}
Since $|\mathbf{n}_h| =1$ a.e.\;and is thus bounded in $L^2$, we conclude that $\mathbf{n}_h \rightharpoonup \mathbf{n}$ in $H^1$. As the estimate $\int_{\Omega} | \partial_3 \mathbf{n}_h|^2 \leq \gamma_{\text{lb}}^{-1} M h^2$ also holds, we conclude that $\partial_3 \mathbf{n}_h \rightarrow \partial_3 \mathbf{n} = 0$ in $L^2$ and thus $\mathbf{n}$ is independent of $x_3$. Lastly concerning the first statement, $\mathbf{n}_h \rightarrow \mathbf{n}$ in $L^2$ by Rellich's theorem. So $0 = \int_{\Omega} (|\mathbf{n}_h| -1)^2 \dd A\rightarrow \int_{\Omega} ( |\mathbf{n}| - 1)^2 \dd x$ and thus $|\mathbf{n}| =1$ a.e. For second statement, it's clear that $h^{-1} \partial_3 \mathbf{n}_h \rightharpoonup \boldsymbol{\tau}$ in $L^2$ since $\int_{\Omega} |h^{-1} \partial_3 \mathbf{n}_h|^2 \dd V \leq \gamma_{\text{lb}}^{-1} M$. Next observe that $ \mathbf{n}_h \cdot h^{-1} \partial_3 \mathbf{n}_h \rightharpoonup \mathbf{n} \cdot \boldsymbol{\tau}$ in $L^2$ using the strong and weak $L^2$ convergence of $\mathbf{n}_h$ and $h^{-1} \partial_3 \mathbf{n}_h$, respectively. Since $\mathbf{n}_h$ is also unit vector, $\mathbf{n}_h \cdot \partial_3 \mathbf{n}_h = 0$ a.e. and thus $\mathbf{n} \cdot \boldsymbol{\tau} = 0$ a.e. The final convergences follows from the Lipschitz continuity of the step length tensors in (\ref{eq:stepLengthsDef}), i.e., the fact that $|(\boldsymbol{\ell}^{\text{f}}_{\mathbf{v}_1})^{\pm1/2} - (\boldsymbol{\ell}^{\text{f}}_{\mathbf{v}_2})^{\pm1/2}| \leq C(\lambda_\text{f}) |\mathbf{v}_1 - \mathbf{v}_2|$ for all $\mathbf{v}_{1,2} \in \mathbb{S}^2$, since $\mathbf{n}_h$ strongly converges to $\mathbf{n}$ in $L^2$. 
\end{proof}

Our second result concerns convergence properties solely of the sequence of deformations in this energy. 
\begin{lem}\label{weakConLemmaY}
$\mathbf{y}_h - \frac{1}{|\Omega|} \int_{\Omega} \mathbf{y}_h \dd V \rightharpoonup \mathbf{y}$ in $H^1(\Omega, \mathbb{R}^3)$ for some midplane deformation $\mathbf{y}$ independent of $x_3$ and $h^{-1} \partial_3 \mathbf{y}_h \rightharpoonup \mathbf{b}$ in $L^2(\Omega, \mathbb{R}^3)$ for some vector field $\mathbf{b}$.
\end{lem}
\begin{proof}
Using the inequalities in (\ref{eq:hyperelastic}) and (\ref{eq:lbLemma}), we have for all $h \in (0,1)$ that 
\begin{align*}
\int_{\Omega}\big\{ |\nabla \mathbf{y}_h|^2 + |\partial_3 \mathbf{y}_h|^2 \big\} \dd V& \leq \int_{\Omega}\big\{ |\nabla \mathbf{y}_h|^2 + |h^{-1}\partial_3 \mathbf{y}_h|^2 \big\} \dd V \\
&\leq c_1(\lambda_{\text{f}}, \lambda_0)^{-1}\Big( \int_{\Omega} \text{dist}^2\big( (\boldsymbol{\ell}^{\text{f}}_{\mathbf{n}_h})^{-1/2} \nabla_h \mathbf{y}_h (\boldsymbol{\ell}^{0}_{\mathbf{n}_0})^{1/2}\big) \dd V + c_2|\Omega| \Big) \\
&\leq c_1(\lambda_{\text{f}}, \lambda_0)^{-1}\Big( c_{\text{lb}}^{-1} E_{\mathbf{n}_0}^h(\mathbf{y}_h, \mathbf{n}_h) + c_2|\Omega| \Big) \leq c_1(\lambda_{\text{f}}, \lambda_0)^{-1}\Big( c_{\text{lb}}^{-1} M + c_2|\Omega| \Big) .
\end{align*}
Combining this estimate with the Poincar\'{e} inequality yields $\mathbf{y}_h - \frac{1}{|\Omega|} \int_{\Omega} \mathbf{y}_h \dd V \rightharpoonup \mathbf{y}$ in $H^1(\Omega, \mathbb{R}^3)$. $\mathbf{y}_h$ is independent of $x_3$ and $h^{-1} \partial_3\mathbf{y}_h$ weakly converges in $L^2$ by a similar argument to that of the director field. 
\end{proof}

\subsubsection{Geometric rigidity and strong convergence of $(\nabla \mathbf{y}_h, \mathbf{b}_h)$}
We now  use the general estimate in Lemma \ref{keyCompactLemma} to establish compactness properties for two sequences of tensor fields built from $\{ (\mathbf{y}_h, \mathbf{n}_h) \}$ and $\mathbf{n}_0$. 
\begin{lem}\label{tensorEstsLemma}
There is a constant $C = C(\omega, \lambda_{\emph{f}}, \lambda_0, c_{\emph{lb}}, \gamma_{\emph{lb}}, \|\nabla \mathbf{n}_0\|_{L^2(\omega)}, M, \epsilon) >0$, a sequence of tensor fields $\{ \mathbf{G}_h \colon \omega \rightarrow \mathbb{R}^{3\times3}\}$, and one of rotation fields $\{ \mathbf{R}_h \colon \omega \rightarrow SO(3) \}$ such that 
\begin{equation}
\begin{aligned}\label{eq:tenRotInequlities}
\int_{\Omega} | \mathbf{G}_h - (\boldsymbol{\ell}^{\emph{f}}_{\mathbf{n}_h})^{-1/2} \nabla_h \mathbf{y}_h (\boldsymbol{\ell}^{0}_{\mathbf{n}_0})^{1/2}\big|^2\dd V + \int_{\omega} | \mathbf{R}_h - \mathbf{G}_h|^2 \dd A \leq C h^2 \quad \text{ and } \quad 
\int_{\Omega} |\nabla \mathbf{G}_h|^2 \dd A \leq C.
\end{aligned}
\end{equation}
\end{lem}
\begin{rem}
The parameter $\epsilon > 0$ above is introduced in Appendix \ref{sec:linAlgebra}. It is used to define the sequence of rotation fields in the lemma.
\end{rem}
\begin{proof}
Lemma \ref{keyCompactLemma} and the boundedness assumptions (\ref{eq:boundedEnergy}) furnishes a $C= C(\omega, \lambda_{\text{f}}, \lambda_0) > 0$ and a sequence of tensor fields such $\{ \mathbf{G}_h \colon \omega \rightarrow \mathbb{R}^{3\times3}\}$ such that 
\begin{equation}
\begin{aligned}\label{eq:tenInequality}
& \int_{\Omega} \frac{1}{h^2} \big| \mathbf{G}_h - (\boldsymbol{\ell}^{\text{f}}_{\mathbf{n}_h})^{-1/2} \nabla_h \mathbf{y}_h (\boldsymbol{\ell}^{0}_{\mathbf{n}_0})^{1/2}\big|^2 \dd V + \int_{\omega} \big|\nabla \mathbf{G}_h \big|^2 \dd A \\
&\qquad \leq C \int_{\Omega}\big\{ c_{\text{lb}}^{-1} h^{-2} W \big( (\boldsymbol{\ell}^{\text{f}}_{\mathbf{n}_h})^{-1/2} \nabla_h \mathbf{y}_h (\boldsymbol{\ell}^{0}_{\mathbf{n}_0})^{1/2} \big) + ( \big| \nabla_h \mathbf{n}_h\big|^2 + |\nabla \mathbf{n}_0|^2 \big) \big\} \dd V \\
&\qquad \leq C\Big( \frac{\max \{ c_{\text{lb}}^{-1} , \gamma_{\text{lb}}^{-1} \}}{h^2} \int_{\Omega} \big\{ W \big( (\boldsymbol{\ell}^{\text{f}}_{\mathbf{n}_h})^{-1/2} \nabla_h \mathbf{y}_h (\boldsymbol{\ell}^{0}_{\mathbf{n}_0})^{1/2} \big) + \gamma_h \big| \nabla_h \mathbf{n}_h\big|^2 \big\} \dd V + \| \nabla \mathbf{n}_0 \|^2_{L^{2}(\omega)} \Big) \\
&\qquad \leq C \big( \max \{ c_{\text{lb}}^{-1} , \gamma_{\text{lb}}^{-1} \} E_{\mathbf{n}_0}^h(\mathbf{y}_h, \mathbf{n}_h) + \| \nabla \mathbf{n}_0 \|^2_{L^{2}(\omega)}\big) \leq C \big( \max \{ c_{\text{lb}}^{-1} , \gamma_{\text{lb}}^{-1} \} M + \| \nabla \mathbf{n}_0 \|^2_{L^{2}(\omega)} \big)
\end{aligned}
\end{equation}
using the inequalities in (\ref{eq:hyperelastic}) and (\ref{eq:gammaH}). Now define a sequence of rotation fields $\{ \mathbf{R}_h \colon \omega \rightarrow SO(3)\}$ by projecting $\mathbf{G}_h$ pointwise to the space of rotations via $\mathbf{R}_h := \boldsymbol{\pi}_{\text{ext}} (\mathbf{G}_h)$ for $\boldsymbol{\pi}_{\text{ext}}$ as defined in Appendix \ref{sec:linAlgebra}. Lemma \ref{projectionLemma} in that appendix gives
\begin{equation}
\begin{aligned}\label{eq:rotInequality}
\int_{\omega} \frac{1}{h^2} |\mathbf{G}_h - \mathbf{R}_h|^2 \dd A\leq C(\epsilon) h^{-2} \int_{\omega} \text{dist}^2(\mathbf{G}_h, SO(3)) \dd A\leq C(\epsilon) c_{\text{lb}}^{-1} E_{\mathbf{n}_0}^h(\mathbf{y}_h, \mathbf{n}_h) \leq C(\epsilon) c_{\text{lb}}^{-1}M
\end{aligned}
\end{equation}
for the $\epsilon > 0$ used to define $\boldsymbol{\pi}_{\text{ext}}$. The estimate in (\ref{eq:tenRotInequlities}) follows from that of (\ref{eq:tenInequality}) and (\ref{eq:rotInequality}). 
\end{proof}
Next we show that the two sequences of tensors fields obtained above converge to the same rotation field.
\begin{lem}\label{GhRhLemma}
The sequences $\{ \mathbf{G}_h \colon \omega \rightarrow \mathbb{R}^{3\times3}\}$ and $\{ \mathbf{R}_h \colon \omega \rightarrow SO(3)\}$ from Lemma \ref{tensorEstsLemma} satisfy $\mathbf{G}_h \rightharpoonup \mathbf{R}$ in $H^1(\omega)$ and $\mathbf{R}_h \rightarrow \mathbf{R}$ in $L^2(\omega)$ for the same rotation field $\mathbf{R} \in H^1(\omega, SO(3))$. 
\end{lem}
\begin{proof}
Per Lemma \ref{tensorEstsLemma} and standard estimates, we have the the sequence $\{\mathbf{G}_h\}$
\begin{align*}
\int_{\omega} |\mathbf{G}_h|^2 \dd A&\leq 2 \big( \int_{\Omega} \big\{ |\mathbf{G}_h - (\boldsymbol{\ell}^{\text{f}}_{\mathbf{n}_h})^{-1/2} \nabla_h \mathbf{y}_h (\boldsymbol{\ell}^{0}_{\mathbf{n}_0})^{1/2}\big|^2 + \big|(\boldsymbol{\ell}^{\text{f}}_{\mathbf{n}_h})^{-1/2} \nabla_h \mathbf{y}_h (\boldsymbol{\ell}^{0}_{\mathbf{n}_0})^{1/2}\big|^2 \big\} \dd V \\ 
&\leq 2 \Big( C h^2 + \tilde{c}(\lambda_{\text{f}}, \lambda_0) \int_{\Omega} |\nabla_h \mathbf{y}_h|^2 \dd V \Big)
\end{align*}
for $C$ as in the previous lemma and a $\tilde{c}(\lambda_{\text{f}}, \lambda_0) >0$ arising from the step-length tensors. Since $\nabla_h \mathbf{y}_h$ weakly converges from Lemma \ref{weakConLemmaY}, we conclude that $\mathbf{G}_h$ is bounded uniformly in $L^2(\omega)$. Lemma \ref{tensorEstsLemma} also established a uniform bound on $\nabla \mathbf{G}_h$ in $L^2(\omega)$. So we conclude that $\mathbf{G}_h \rightharpoonup \mathbf{R}$ in $H^1(\omega, \mathbb{R}^{3\times3})$. 

To prove that $\mathbf{R} \in SO(3)$ a.e., observe that 
\begin{align*}
&\int_{\omega} \text{dist}^2(\mathbf{R}, SO(3)) \dd A\leq 2 \int_{\omega}\big\{ \text{dist}^2(\mathbf{G}_h ,SO(3)) + |\mathbf{G}_h - \mathbf{R}|^2 \big\} \dd A \\
&\qquad \qquad \leq 4 c_{\text{lb}}^{-1}h^2 E_{\mathbf{n}_0}^h(\mathbf{y}_h, \mathbf{n}_h)+4 \int_{\Omega} \big\{ \big| \mathbf{G}_h - (\boldsymbol{\ell}^{\text{f}}_{\mathbf{n}_h})^{-1/2} \nabla_h \mathbf{y}_h (\boldsymbol{\ell}^{0}_{\mathbf{n}_0})^{1/2}|^2 + |\mathbf{G}_h - \mathbf{R}|^2 \big\} \dd V.
\end{align*}
From here, Lemma \ref{tensorEstsLemma} and the energy bound in (\ref{eq:boundedEnergy}) furnishes the inequality 
\begin{equation}
\begin{aligned}\label{eq:tensorEstsToRot}
\int_{\omega} \text{dist}^2(\mathbf{R}, SO(3)) \dd A\leq 4\big( c_{\text{lb}}^{-1} M + C\big) h^2 + 2\int_{\omega} |\mathbf{G}_h - \mathbf{R}|^2 \dd x.
\end{aligned}
\end{equation}
By Rellich's theorem, $\mathbf{G}_h \rightarrow \mathbf{R}$ in $L^2(\omega)$ since it weakly converges in $H^1(\omega)$. Thus, $\mathbf{R} \in SO(3)$ a.e. since (\ref{eq:tensorEstsToRot}) holds as $h \rightarrow 0$. 

For the sequence of rotations fields $\{ \mathbf{R}_h\}$, we have 
\begin{align*}
\int_{\omega} |\mathbf{R}_h - \mathbf{R}|^2 \dd A\leq 2 \Big( \int_{\omega}\big\{ | \mathbf{R}_h - \mathbf{G}_h|^2 + |\mathbf{G}_h - \mathbf{R}|^2 \big\} \dd A\Big) \leq 2Ch^2 + 2\int_{\omega} |\mathbf{G}_h - \mathbf{R}|^2 \dd A, 
\end{align*}
where the second inequality follows by Lemma \ref{tensorEstsLemma}. So $\mathbf{R}_h \rightarrow \mathbf{R}$ in $L^2(\omega)$ because $\mathbf{G}_h$ does. 
\end{proof}

We now link the rotation field $\mathbf{R}$ above to the sequences $\{ (\mathbf{y}_h, \mathbf{n}_h)\}$.
\begin{lem}\label{wArgLimLemma}
$ (\boldsymbol{\ell}^{\emph{f}}_{\mathbf{n}_h})^{-1/2} \nabla_h \mathbf{y}_h (\boldsymbol{\ell}^{0}_{\mathbf{n}_0})^{1/2} \rightarrow \mathbf{R} = (\boldsymbol{\ell}^{\emph{f}}_{\mathbf{n}})^{-1/2} (\nabla \mathbf{y}, \mathbf{b}) (\boldsymbol{\ell}^{0}_{\mathbf{n}_0})^{1/2}$ in $L^2(\Omega, \mathbb{R}^{3\times3})$ for the rotation field $\mathbf{R}$ from Lemma \ref{GhRhLemma} and the vector fields $\mathbf{n}, \mathbf{y}, \mathbf{b}$ from Lemma \ref{weakConLemmaN} and \ref{weakConLemmaY}. In particular, $(\boldsymbol{\ell}^{\emph{f}}_{\mathbf{n}})^{-1/2} (\nabla \mathbf{y}, \mathbf{b}) (\boldsymbol{\ell}^{0}_{\mathbf{n}_0})^{1/2}$ is a rotation field that is independent of $x_3$ and belongs to $H^1(\omega)$.
\end{lem}
\begin{proof}
First observe that 
\begin{align*}
\int_{\Omega} \big| (\boldsymbol{\ell}^{\text{f}}_{\mathbf{n}_h})^{-1/2} \nabla_h \mathbf{y}_h (\boldsymbol{\ell}^{0}_{\mathbf{n}_0})^{1/2} - \mathbf{R} \big|^2 \dd A\leq 2 \int_{\Omega} \big| (\boldsymbol{\ell}^{\text{f}}_{\mathbf{n}_h})^{-1/2} \nabla_h \mathbf{y}_h (\boldsymbol{\ell}^{0}_{\mathbf{n}_0})^{1/2} - \mathbf{G}_h \big|^2 \dd V + 2\int_{\omega} |\mathbf{G}_h -\mathbf{R}|^2 \dd A. 
\end{align*}
So $(\boldsymbol{\ell}^{\text{f}}_{\mathbf{n}_h})^{-1/2} \nabla_h \mathbf{y}_h (\boldsymbol{\ell}^{0}_{\mathbf{n}_0})^{1/2} \rightarrow \mathbf{R}$ in $L^2(\Omega, \mathbb{R}^{3\times3})$ given estimates in Lemma \ref{tensorEstsLemma} and the fact that $\mathbf{G}_h \rightarrow \mathbf{R}$ in $L^2(\omega)$ from Lemma \ref{GhRhLemma}. On the other hand, we have $\nabla_h \mathbf{y}_h \rightharpoonup (\nabla \mathbf{y} ,\mathbf{b})$ in $L^2(\Omega)$ from Lemma \ref{weakConLemmaY} and $(\boldsymbol{\ell}^{\text{f}}_{\mathbf{n}_h})^{-1/2} \rightarrow (\boldsymbol{\ell}^{\text{f}}_{\mathbf{n}})^{-1/2}$ in $L^2(\Omega)$ from Lemma \ref{weakConLemmaN}. Thus, since we are dealing with a sequence composed of products of sequences which weakly and strongly converge in $L^2$, we conclude that $(\boldsymbol{\ell}^{\text{f}}_{\mathbf{n}_h})^{-1/2} \nabla_h \mathbf{y}_h (\boldsymbol{\ell}^{0}_{\mathbf{n}_0})^{1/2} \rightharpoonup (\boldsymbol{\ell}^{\text{f}}_{\mathbf{n}})^{-1/2} (\nabla \mathbf{y}, \mathbf{b}) (\boldsymbol{\ell}^{0}_{\mathbf{n}_0})^{1/2}$ in $L^1(\Omega)$. Evidently then $\mathbf{R} = ( \boldsymbol{\ell}^{\text{f}}_{\mathbf{n}})^{-1/2} (\nabla \mathbf{y}, \mathbf{b}) (\boldsymbol{\ell}^{0}_{\mathbf{n}_0})^{1/2}$ a.e. since the weak-$L^1$ limit is unique. It follows from Lemma \ref{GhRhLemma} that $( \boldsymbol{\ell}^{\text{f}}_{\mathbf{n}})^{-1/2} (\nabla \mathbf{y}, \mathbf{b}) (\boldsymbol{\ell}^{0}_{\mathbf{n}_0})^{1/2}$ is a rotation field that is independent of $x_3$ and belongs to $H^1(\omega)$. 
\end{proof}

Next we build on the previous result to strengthen the convergence properties of $\{ \nabla_h \mathbf{y}_h\}$ and deduce additional regularity on the limit $(\nabla \mathbf{y}, \mathbf{b})$. 
\begin{lem}\label{strengthConLem}
The sequences in Lemma \ref{weakConLemmaY} actually strongly converge in their respective spaces. We also have improved regularity of the limits: $\nabla \mathbf{y}$ belongs to $H^1(\omega, \mathbb{R}^3)$ and $\mathbf{b}$ is independent of $x_3$ and in $H^1(\omega, \mathbb{R}^3)$. 
\end{lem}
\begin{proof}
For the strong convergence, observe that 
\begin{align*}
&\int_{\Omega} | \nabla_h \mathbf{y}_h - (\nabla \mathbf{y}, \mathbf{b})|^2 \dd V \\
&\qquad \leq 2\int_{\Omega} \big\{ | \nabla_h \mathbf{y}_h + ( \boldsymbol{\ell}^{\text{f}}_{\mathbf{n}_h})^{1/2} \mathbf{R} (\boldsymbol{\ell}^{0}_{\mathbf{n}_0})^{-1/2} |^2 +| (\nabla \mathbf{y}, \mathbf{b}) - ( \boldsymbol{\ell}^{\text{f}}_{\mathbf{n}_h})^{1/2} \mathbf{R} (\boldsymbol{\ell}^{0}_{\mathbf{n}_0})^{-1/2} |^2 \big\} \dd V \\
&\qquad \leq 2c(\lambda_{\text{f}}, \lambda_0) \int_{\Omega} \big| ( \boldsymbol{\ell}^{\text{f}}_{\mathbf{n}_h})^{-1/2} \nabla_h \mathbf{y}_h (\boldsymbol{\ell}^{0}_{\mathbf{n}_0})^{1/2} - \mathbf{R}|^2 + | ( \boldsymbol{\ell}^{\text{f}}_{\mathbf{n}_h})^{-1/2} - ( \boldsymbol{\ell}^{\text{f}}_{\mathbf{n}})^{-1/2} |^2 \big\} \dd V,
\end{align*}
where the equality $(\nabla \mathbf{y}, \mathbf{b}) = ( \boldsymbol{\ell}^{\text{f}}_{\mathbf{n}})^{1/2} \mathbf{R} ( \boldsymbol{\ell}^{\text{0}}_{\mathbf{n}_0})^{-1/2}$ a.e. from the previous lemma allows us to go from the first to second estimate above. Lemma \ref{weakConLemmaN} and \ref{wArgLimLemma} then give, respectively, the convergences $ ( \boldsymbol{\ell}^{\text{f}}_{\mathbf{n}_h})^{-1/2} \nabla_h \mathbf{y}_h (\boldsymbol{\ell}^{0}_{\mathbf{n}_0})^{1/2} \rightarrow \mathbf{R}$ and $( \boldsymbol{\ell}^{\text{f}}_{\mathbf{n}_h})^{-1/2} \rightarrow ( \boldsymbol{\ell}^{\text{f}}_{\mathbf{n}})^{-1/2}$ in $L^2(\Omega)$. The strong convergence $ \nabla_h \mathbf{y}_h \rightarrow (\nabla \mathbf{y}, \mathbf{b})$ in $L^2(\Omega)$ follows.

For the added regularity, note again from the previous lemma that $(\nabla \mathbf{y}, \mathbf{b}) = ( \boldsymbol{\ell}^{\text{f}}_{\mathbf{n}})^{1/2} \mathbf{R} ( \boldsymbol{\ell}^{\text{0}}_{\mathbf{n}_0})^{-1/2}$ a.e. Note also that $\mathbf{n}_0 \in H^1(\omega)$ by hypothesis, $\mathbf{n} \in H^1(\omega)$ by Lemma \ref{weakConLemmaN}, $\mathbf{R} \in H^1(\omega)$ by Lemma \ref{GhRhLemma}, and $( \boldsymbol{\ell}^{\text{f}}_{\mathbf{n}})^{1/2}$, $\mathbf{R}$ and $( \boldsymbol{\ell}^{\text{0}}_{\mathbf{n}_0})^{-1/2}$ are bounded in $L^{\infty}(\omega)$. So the product rule gives that $\nabla \mathbf{y}$ belongs to $H^1(\omega)$ and $\mathbf{b}$ is independent of $x_3$ and in $H^1(\omega)$.
\end{proof}

\subsubsection{Metric constraint and the nonideal energy}
At this stage, we have only used the ideal entropic elastic energy term and the Frank elastic term to achieve the stated compactness properties. We now analyze the non-ideal director anchoring term in the context of compactness.
\begin{lem}
The limits $\mathbf{n}$, $\mathbf{b}$ and $\mathbf{\mathbf{y}}$ satisfy 
\begin{equation}
\begin{aligned}\label{eq:anchoringPointwise}
\mathbf{n} = \sigma \frac{\nabla \mathbf{y}\mathbf{n}_0}{|\nabla \mathbf{y} \mathbf{n}_0|}, \quad \mathbf{b} = \lambda_0^{-1/8} \lambda_{\emph{f}}^{1/8} \boldsymbol{\nu}_{\mathbf{y}} , \quad (\nabla \mathbf{y})^T \nabla \mathbf{y} = \mathbf{g}_{\mathbf{n}_0} \quad \text{ a.e. on $\omega$}.
\end{aligned}
\end{equation}
for a fixed $\sigma \in \{ -1,1\}$ and for $\mathbf{g}_{\mathbf{v}_0}$ defined in (\ref{eq:metricDef}). 
\end{lem}
\begin{proof}
Given the hypothesis on $\mu_h$ in (\ref{eq:muH}) and the energy bound in (\ref{eq:boundedEnergy}), we have 
\begin{align*}
\int_{\Omega} |\mathbf{P}_{\mathbf{n}_0} (\nabla_h \mathbf{y}_h)^T\mathbf{n}_h|^2 \dd V \leq h^2 \mu_h^{-1} E_{\mathbf{n}_0}^h(\mathbf{y}_h, \mathbf{n}_h) \leq M h^2 \mu_h^{-1} \rightarrow 0 \quad \text{ as } h \rightarrow 0.
\end{align*}
So $\mathbf{P}_{\mathbf{n}_0} (\nabla_h \mathbf{y}_h)^T\mathbf{n}_h \rightarrow \mathbf{0}$ in $L^2(\Omega)$. We also have $(\nabla_h \mathbf{y}_h)^T \rightarrow (\nabla \mathbf{y}, \mathbf{b})^T$ and $\mathbf{n}_h \rightarrow \mathbf{n}$ in $L^2(\Omega)$ by Lemma \ref{strengthConLem}. Thus $\mathbf{P}_{\mathbf{n}_0} (\nabla_h \mathbf{y}_h)^T\mathbf{n}_h \rightarrow \mathbf{P}_{\mathbf{n}_0} (\nabla \mathbf{y}, \mathbf{b})^T \mathbf{n}$ in $L^1(\Omega)$ and the uniqueness of the $L^1$-limit implies $ \mathbf{P}_{\mathbf{n}_0} (\nabla \mathbf{y}, \mathbf{b})^T \mathbf{n} = \mathbf{0}$ a.e. We now combine this equality with that of $\mathbf{R} = ( \boldsymbol{\ell}^{\text{f}}_{\mathbf{n}})^{-1/2} (\nabla \mathbf{y}, \mathbf{b}) ( \boldsymbol{\ell}^{\text{0}}_{\mathbf{n}_0})^{1/2}$ a.e.\;from Lemma \ref{wArgLimLemma} to conclude that $\mathbf{P}_{\mathbf{n}_0} \mathbf{R}^T \mathbf{n} = \mathbf{0}$ a.e. Since $\mathbf{R}$ is the $H^1$ rotation field from Lemma \ref{GhRhLemma} and $\mathbf{n}$ is the $H^1$ unit vector field from Lemma \ref{weakConLemmaN}, this constraint implies that $\mathbf{n} = \sigma \mathbf{R} \begin{pmatrix} \mathbf{n}_0 \\ 0 \end{pmatrix}$ a.e.\;for some fixed $\sigma \in \{ -1,1\}$. To get to the identities in (\ref{eq:anchoringPointwise}) from here is a matter of linear algebra. Lemma \ref{linAlgLemma} in the appendix supplies the result. 
\end{proof}

\subsubsection{Strain and the incompressibility penalty}

As a final point, we establish convergence properties for the sequence of strain measures $\{ \mathbf{S}_h \colon \Omega \rightarrow \mathbb{R}^{3\times3}$\} defined by 
\begin{align*}
\mathbf{S}_h := \frac{1}{h}\Big( \mathbf{R}_h^T (\boldsymbol{\ell}^{\emph{f}}_{\mathbf{n}_h})^{-1/2} \nabla_h \mathbf{y}_h (\boldsymbol{\ell}^{0}_{\mathbf{n}_0})^{1/2} - \mathbf{I}\Big)
\end{align*}
which are key to the lowerbound argument. This step is the first in our compactness analysis to make use of the incompressibility penalty. 
\begin{lem}\label{ShLemma}
$\mathbf{S}_h \rightharpoonup \mathbf{S}$ in $L^2(\Omega, \mathbb{R}^{3\times3})$ for some tensor field $\mathbf{S}$ that satisfies $\emph{Tr}(\mathbf{S}) = 0$ a.e.
\end{lem}

\begin{proof}
As $\mathbf{R}_h$ is a rotation field, first observe that 
\begin{align*}
\int_{\Omega} |\mathbf{S}_h|^2 \dd V &= \int_{\Omega} h^{-2} | \mathbf{R}_h - (\boldsymbol{\ell}^{\emph{f}}_{\mathbf{n}_h})^{-1/2} \nabla_h \mathbf{y}_h (\boldsymbol{\ell}^{0}_{\mathbf{n}_0})^{1/2} |^2 \dd V \\
&\leq \frac{2}{h^2} \Big( \int_{\omega} | \mathbf{R}_h - \mathbf{G}_h|^2 \dd A+ \int_{\Omega} |\mathbf{G}_h - (\boldsymbol{\ell}^{\emph{f}}_{\mathbf{n}_h})^{-1/2} \nabla_h \mathbf{y}_h (\boldsymbol{\ell}^{0}_{\mathbf{n}_0})^{1/2}|^2 \dd V \Big) 
\end{align*}
It follows from Lemma \ref{tensorEstsLemma} that $\int_{\Omega} |\mathbf{S}_h|^2 \dd V \leq C$ and thus $\mathbf{S}_h \rightharpoonup \mathbf{S}$ in $L^2(\Omega)$. 

That $\mathbf{S}$ has zero trace follows from a careful Taylor expansion of the approximate incompressibility energy. 
First observe that
\begin{equation}
\begin{aligned}\label{eq:detIdent2}
 \det \nabla_h \mathbf{y}_h & = \det \big[ (\boldsymbol{\ell}^{\emph{f}}_{\mathbf{n}_h})^{-1/2} \nabla_h \mathbf{y}_h (\boldsymbol{\ell}^{0}_{\mathbf{n}_0})^{1/2} \big]= \det( \mathbf{I} + h \mathbf{S}_h) \quad \text{ a.e. on $\Omega$}
\end{aligned}
\end{equation}
 since the step length tensors satisfy $\det( \boldsymbol{\ell}_{\mathbf{v}}^{\text{f}}) = \det( \boldsymbol{\ell}_{\mathbf{v}_0}^{\text{0}}) = 1$. 
We aim to apply the Taylor expansion identity in (\ref{eq:detIdent}) on $\det( \mathbf{I} + h \mathbf{S}_h)$, but can only do so on a set for which $\mathbf{S}_h$ is bounded. Consider the characteristic function $\chi_h \colon \Omega \rightarrow \{ 0,1\}$ defined by 
 \begin{equation}
 \begin{aligned}\label{eq:chiHDef}
 \chi_h := \begin{cases}
 1 & \text{ if } |\mathbf{S}_h| \leq h^{-1/2} \\
 0 & \text{ otherwise},
 \end{cases}
 \end{aligned}
 \end{equation}
It follows from Chebychev's inequality that $|\Omega| - |\chi_h(\Omega)| \leq h \int_{\Omega} |\mathbf{S}_h|^2 dV \leq Ch \rightarrow 0$ and thus $\chi_h \rightarrow 1$ boundedly in measure. This fact, combined with weak convergence of $\mathbf{S}_h$, allows us to conclude that 
\begin{align}\label{eq:TaylorWeak}
 \chi_h \mathbf{S}_h \rightharpoonup \mathbf{S} \quad \text{ in $L^2(\Omega)$}.
\end{align}
The energy bound in (\ref{eq:boundedEnergy}) and the pointwise identity in (\ref{eq:detIdent2}) furnish
 \begin{align*}
 M \kappa_h^{-1} & \geq E_{\mathbf{n}_0}^h(\mathbf{y}_h, \mathbf{n}_h) \geq \int_{\Omega} \frac{\chi_h}{h^2} ( \det \nabla_h \mathbf{y}_h - 1)^2 \dd V \\
 &= \int_{\Omega} \frac{\chi_h}{h^2} ( \det [\mathbf{I} + h \mathbf{S}_h] - 1)^2 \dd V =\int_{\Omega} \chi_h \big\{ \text{Tr}(\mathbf{S}_h)^2 + o(|\mathbf{S}_h|^2) \big\} \dd V
 \end{align*}
where the last equality holds since $\mathbf{S}_h$ is bounded when $\chi_h =1$.
As $\chi_h \text{Tr}(\mathbf{S}_h)^2 = \text{Tr}(\chi_h \mathbf{S}_h)^2$ and $\int_{\Omega} \chi_h o(|\mathbf{S}_h|^2) dV = o\big( \int_{\Omega} \chi_h |\mathbf{S}_h|^2 dV\big) \rightarrow 0$, we take the limit of the last set of inequalities to obtain
\begin{align*}
0 = \liminf_{h\rightarrow 0} M \kappa_h^{-1} \geq \liminf_{h\rightarrow 0} \Big\{ \int_{\Omega} \text{Tr}(\chi_h \mathbf{S}_h)^2 \dd V + o\Big( \int_{\Omega} \chi_h |\mathbf{S}_h|^2 \dd V\Big) \Big\} \geq \int_{\Omega} \text{Tr}( \mathbf{S})^2 \dd V
\end{align*}
since $\kappa_h \rightarrow \infty$ from (\ref{eq:kappaH}), and by the convexity of $\text{Tr}(\mathbf{A})^2$ and convergence in (\ref{eq:TaylorWeak}). We conclude that $\text{Tr}(\mathbf{S}) = 0$ a.e. 
\end{proof}

\subsubsection{Summary on compactness.} For convenience, we now summarize all the compactness results established by Lemmas \ref{weakConLemmaN}-\ref{ShLemma} above. 

\begin{prop}\label{compactnessProp}
For any $\mathbf{n}_0 \in H^1(\omega, \mathbb{S}^1)$ and any sequence $\{ (\mathbf{y}_h, \mathbf{n}_h)\} \subset H^1(\Omega, \mathbb{R}^3) \times H^1(\Omega, \mathbb{S}^2)$ such that $\liminf_{h \rightarrow 0} E_{\mathbf{n}_0}^h(\mathbf{y}_h, \mathbf{n}_h) <\infty$, there is a subsequence (not relabeled) such that 
\begin{align*}
&\mathbf{y}_h - \frac{1}{|\Omega|}\int_{\Omega} \mathbf{y}_h \dd V \rightarrow \mathbf{y} \text{ in } H^1(\Omega) && \text{ for $\mathbf{y}$ a vector field on $\mathbb{R}^3$ independent of $x_3$ and in $H^2(\omega)$} \\
&h^{-1} \partial_3 \mathbf{y}_h \rightarrow \mathbf{b} \text{ in } L^2(\Omega) && \text{ for $\mathbf{b}$ a vector field on $\mathbb{R}^3$ independent of $x_3$ and in $H^1(\omega)$} , \\
&\mathbf{n}_h \rightharpoonup \mathbf{n} \text{ in } H^1(\Omega) && \text{ for $\mathbf{n}$ a unit vector field on $\mathbb{S}^2$ independent of $x_3$ }, \\
& (\boldsymbol{\ell}^{\emph{f}}_{\mathbf{n}_h})^{-1/2} \nabla_h \mathbf{y}_h (\boldsymbol{\ell}^{0}_{\mathbf{n}_0})^{1/2} \rightarrow \mathbf{R} \text{ in } H^1(\Omega) && \text{ for $\mathbf{R}$ a rotation field on $SO(3)$ independent of $x_3$ and in $H^1(\omega)$ } \\
&h^{-1} \partial_3 \mathbf{n}_h \rightharpoonup \boldsymbol{\tau} \text{ in } L^2(\Omega) && \text{ for $\boldsymbol{\tau}$ a vector field on $\mathbb{R}^3$ such that $\boldsymbol{\tau} \cdot \mathbf{n} = 0$ a.e.} 
\end{align*}
The limiting fields $\mathbf{n}, \mathbf{R}, \mathbf{y}$ and $\mathbf{b}$ are also coupled via the identities 
\begin{equation}
\begin{aligned}\label{eq:identsOnPlanarFields}
\mathbf{n} = \sigma \mathbf{R} \begin{pmatrix} \mathbf{n}_0 \\ 0 \end{pmatrix} = \sigma \frac{\nabla \mathbf{y} \mathbf{n}_0}{|\nabla \mathbf{y} \mathbf{n}_0|}, \quad \mathbf{R} = (\boldsymbol{\ell}_{\mathbf{n}}^{\emph{f}})^{-1/2} (\nabla \mathbf{y}, \mathbf{b}) (\boldsymbol{\ell}_{\mathbf{n}_0}^0)^{1/2}, \quad (\nabla \mathbf{y})^T \nabla \mathbf{y} = \mathbf{g}_{\mathbf{n}_0}, \quad \mathbf{b} = \lambda_{\emph{f}}^{-1/4} \lambda_0^{1/4} \boldsymbol{\nu}_{\mathbf{y}} \quad 
\end{aligned}
\end{equation}
 a.e.\;on $\omega$. In addition, there is a sequence of rotation fields $\{ \mathbf{R}_h \colon \omega \rightarrow SO(3)\}$ and a strain measure $\mathbf{S}_h := \frac{1}{h}\big( \mathbf{R}_h^T (\boldsymbol{\ell}^{\emph{f}}_{\mathbf{n}_h})^{-1/2} \nabla_h \mathbf{y}_h (\boldsymbol{\ell}^{0}_{\mathbf{n}_0})^{1/2} - \mathbf{I})$ such that 
 \begin{align*}
& \mathbf{R}_h \rightarrow \mathbf{R} \text{ in } L^2(\omega) &&\text{ for $\mathbf{R}$ the rotation field above} \\
 &\mathbf{S}_h \rightharpoonup \mathbf{S} \text{ in } L^2(\Omega) &&\text{ for $\mathbf{S}$ a tensor field on $\mathbb{R}^{3\times3}$ with $\emph{Tr}(\mathbf{S}) = 0$ a.e.} 
 \end{align*}
\end{prop}

\begin{rem}
The limiting fields in Proposition \ref{compactnessProp} 
satisfy $(\mathbf{y}, \mathbf{n}) \in \mathcal{A}_{\mathbf{n}_0}$ for $\mathcal{A}_{\mathbf{n}_0}$  in (\ref{eq:admissibleSet}).
\end{rem}
\section{Lowerbound}
 
% \subsection{Proof of the lowerbound}
This section completes the proof of Theorem \ref{MainTheorem} by establishing the lowerbound result. To begin, let $\mathbf{n}_0 \in H^1(\omega, \mathbb{S}^1)$ and let $\{ (\mathbf{y}_h, \mathbf{n}_h)\} \subset H^1(\Omega, \mathbb{R}^3) \times H^1(\Omega, \mathbb{S}^2)$ be a sequence such that $(\mathbf{y}_h , \mathbf{n}_h) \rightharpoonup (\mathbf{y}, \mathbf{n})$ in $H^1(\Omega, \mathbb{R}^3) \times H^1(\Omega, \mathbb{S}^2)$ and $\liminf_{h \rightarrow 0} E_{\mathbf{n}_0}^h(\mathbf{y}_h, \mathbf{n}_h) < \infty$. From the results of the prior section, we have $(\mathbf{y}, \mathbf{n}) \in \mathcal{A}_{\mathbf{n}_0}$. In addition, we extract a subsequence $\{(\mathbf{y}_h, \mathbf{n}_h) \}$ (not relabeled), and a sequence of rotation fields $\{ \mathbf{R}_h \colon \omega \rightarrow SO(3)\}$ and strain fields $\{ \mathbf{S}_h := \frac{1}{h}\big( \mathbf{R}_h^T (\boldsymbol{\ell}^{\emph{f}}_{\mathbf{n}_h})^{-1/2} \nabla_h \mathbf{y}_h (\boldsymbol{\ell}^{0}_{\mathbf{n}_0})^{1/2} - \mathbf{I})\}$ built from this subsequence which satisfy all the convergence properties stated in Proposition \ref{compactnessProp}. We break the rest of the proof up into two subsections. The first uses arguments similar to those in \cite{bhattacharya2016plates, friesecke2002theorem,lewicka2011scaling} to establish an initial lower bound in Lemma \ref{firstLiminfLemma} and identify the limit $\mathbf{S}_h\rightharpoonup \mathbf{S}$ in Lemma \ref{identifySLemma} as a function of other limiting quantities in Proposition \ref{compactnessProp}. The second refines this lower bound by optimizing over the free degrees of freedom contained in $\mathbf{S}$, which in turn completes the proof.

\subsection{Asymptotics and identification of the limiting strain.}

Here we   develop a lowerbound on the limit of the rescaled energy through a Taylor expansion argument, and identify $\mathbf{S}$ through analysis of the finite difference quotient of the sequence $\{ \mathbf{y}_h\}$ along the $x_3$-coordinate. 

\subsubsection{Asymptotics} Our first result concerns an initial estimate on the energy of this subsequence as $h \rightarrow 0$. 
\begin{lem}\label{firstLiminfLemma}
The energy of the subsequence satisfies
\begin{equation}
\begin{aligned}\label{eq:firstLiminf}
\liminf_{h \rightarrow 0} E_{\mathbf{n}_0}^h(\mathbf{y}_h, \mathbf{n}_h) \geq \int_{\Omega}\Big\{\frac{1}{2} Q_3(\mathbf{S}) + \gamma( |\nabla \mathbf{n}|^2 + |\boldsymbol{\tau}|^2) \Big\} \dd V.
\end{aligned}
\end{equation}
\end{lem}
\begin{proof}
Let's first observe that 
\begin{equation}
\begin{aligned}\label{eq:discard1}
E_{\mathbf{n}_0}^h(\mathbf{y}_h, \mathbf{n}_h) \geq \int_{\Omega} \Big\{\frac{1}{h^2} W\big( (\boldsymbol{\ell}^{\emph{f}}_{\mathbf{n}_h})^{-1/2} \nabla_h \mathbf{y}_h (\boldsymbol{\ell}^{0}_{\mathbf{n}_0})^{1/2}\big) + \frac{\gamma_h}{h^2} |\nabla_h \mathbf{n}_h|^2 \Big\} \dd V 
\end{aligned}
\end{equation}
after discarding the non-negative incompressibility and director anchoring terms from the energy. Since $h^{-2} \gamma_h \rightarrow \gamma$ and $\nabla_h \mathbf{n}_h \rightharpoonup (\nabla \mathbf{n}, \boldsymbol{\tau})$ in $L^2(\Omega)$, the Frank term in this bound satisfies 
\begin{equation}
\begin{aligned}\label{eq:FrankLiminf}
\liminf_{h \rightarrow 0} \int_{\Omega} \frac{\gamma_h}{h^2} |\nabla_h \mathbf{n}_h|^2 \dd V& \geq \liminf_{h \rightarrow 0} \Big( \int_{\Omega} \gamma |\nabla_h \mathbf{n}_h|^2\dd V + (\frac{\gamma_h}{h^2} - \gamma) \| \nabla_h \mathbf{n}_h\|_{L^2(\Omega)}^2 \Big)\\
&\geq \liminf_{h \rightarrow 0} \int_{\Omega} \gamma |\nabla_h \mathbf{n}_h|^2\dd V \geq \int_{\Omega} \gamma \big( |\nabla \mathbf{n}|^2 + |\boldsymbol{\tau}|^2\big) \dd V.
\end{aligned}
\end{equation}
The entropic energy density meanwhile satisfies 
\begin{equation}
\begin{aligned}\label{eq:WFrameIdent}
W\big( (\boldsymbol{\ell}^{\emph{f}}_{\mathbf{n}_h})^{-1/2} \nabla_h \mathbf{y}_h (\boldsymbol{\ell}^{0}_{\mathbf{n}_0})^{1/2}\big) &= W\big(\mathbf{R}_h^T (\boldsymbol{\ell}^{\emph{f}}_{\mathbf{n}_h})^{-1/2} \nabla_h \mathbf{y}_h (\boldsymbol{\ell}^{0}_{\mathbf{n}_0})^{1/2}\big) = W\big(\mathbf{I} + h\mathbf{S}_h)
\end{aligned}
\end{equation}
a.e.\;on $\Omega$ by the frame indifference of $W$ in (\ref{eq:hyperelastic}) and the definition of $\mathbf{R}_h$ and $\mathbf{S}_h$ and in Proposition \ref{compactnessProp}. 
Next, introduce again the characteristic function $\chi_h \colon \Omega \rightarrow \{ 0,1\}$ from (\ref{eq:chiHDef}) for the purpose of analyzing the entropic energy, and note that the convergence in (\ref{eq:TaylorWeak}) still holds. Since $\chi_h \leq 1$, it follows from (\ref{eq:WFrameIdent}) and the boundedness of $\mathbf{S}_h$ when $\chi_h =1$ that
\begin{align*}
\int_{\Omega} \frac{1}{h^2} W\big( (\boldsymbol{\ell}^{\emph{f}}_{\mathbf{n}_h})^{-1/2} \nabla_h \mathbf{y}_h (\boldsymbol{\ell}^{0}_{\mathbf{n}_0})^{1/2}\big) \dd V \geq \int_{\Omega}\frac{\chi_h}{h^2} W\big(\mathbf{I} + h\mathbf{S}_h) \dd V \geq \int_{\Omega} \chi_h \Big\{\frac{1}{2} Q_3(\mathbf{S}_h) + o(|\mathbf{S}_h|^2) \Big\} \dd V
\end{align*}
using (\ref{eq:TaylorW}) and (\ref{eq:Q3A}). As $\chi_h Q_3(\mathbf{S}_h) = Q_3(\chi_h \mathbf{S}_h) $ and $\int_{\Omega} \chi_h o( |\mathbf{S}_h|^2) \dd V = o(\int_{\Omega} \chi_h|\mathbf{S}_h|^2 \dd V) \rightarrow 0$,
\begin{equation}
\begin{aligned}\label{eq:EntropicLiminf}
&\liminf_{h \rightarrow 0} \int_{\Omega} \frac{1}{h^2} W\big( (\boldsymbol{\ell}^{\emph{f}}_{\mathbf{n}_h})^{-1/2} \nabla_h \mathbf{y}_h (\boldsymbol{\ell}^{0}_{\mathbf{n}_0})^{1/2}\big) \dd V \\
&\qquad \qquad \qquad \geq \liminf_{h \rightarrow 0} \Big\{ \int_{\Omega} \frac{1}{2} Q_3( \chi_h \mathbf{S}_h) \dd V + o\Big(\int_{\Omega} \chi_h|\mathbf{S}_h|^2 \dd V \Big) \Big\} \geq \int_{\Omega} \frac{1}{2} Q_3(\mathbf{S}) \dd V
\end{aligned}
\end{equation}
using the convexity of $Q_3(\mathbf{A})$ and (\ref{eq:TaylorWeak}). The desired estimate in (\ref{eq:firstLiminf}) follows from (\ref{eq:discard1}), (\ref{eq:FrankLiminf}) and (\ref{eq:EntropicLiminf}). 
\end{proof}

\subsubsection{Identification of $\mathbf{S}$}
Our next result identifies components the strain measure $\mathbf{S}$ as explicit functions of $\mathbf{R}, \mathbf{n}, \mathbf{b}$ and $\boldsymbol{\tau}$. 
\begin{lem}\label{identifySLemma}
The first two columns $\mathbf{S}$ satisfy
\begin{equation}
\begin{aligned}\label{eq:S3by2}
\big[ \mathbf{S}(\mathbf{x}, x_3) \big]_{3\times2} &= \big[ \mathbf{S}(\mathbf{x}, 0) \big]_{3\times2} + x_3 (\boldsymbol{\ell}_{\mathbf{n}_0}^{\emph{f}})^{-1/2}(\mathbf{x}) \mathbf{R}^T(\mathbf{x}) \nabla \mathbf{b}(\mathbf{x}) [(\boldsymbol{\ell}^{0}_{\mathbf{n}_0})^{1/2}(\mathbf{x}) ]_{2 \times2} \\
&\qquad - c_{\emph{f}}\;(\boldsymbol{\ell}_{\mathbf{n}_0}^{\emph{f}})^{-1/2}(\mathbf{x}) \Big[\emph{sym} \Big( \int_{0}^{x_3} \mathbf{R}^T(\mathbf{x}) \boldsymbol{\tau}(\mathbf{x},t) \dd t \otimes \mathbf{R}^T(\mathbf{x}) \mathbf{n}(\mathbf{x}) \Big)\Big]_{3\times2} 
\end{aligned}
\end{equation}
a.e.\;on $\Omega$ for some $\big[ \mathbf{S}(\mathbf{x}, 0) \big]_{3\times2} \in L^2(\omega, \mathbb{R}^{3\times2})$, for $c_\emph{f} := 2(\lambda_{\emph{f}}^{1/2} - \lambda_{\emph{f}}^{-1/4})$ and for $(\boldsymbol{\ell}_{\mathbf{v}_0}^{\emph{f}})$ as defined in (\ref{eq:mixedStepLength}). 
\end{lem}
\begin{proof}
Let $\Omega'$ be any compact subset of $\Omega$, let $s$ satisfy $|s| \leq \text{dist}(\Omega', \Omega)$, and define the $H^1(\Omega', \mathbb{R}^3)$ function 
\begin{align*}
\mathbf{f}_{h,s}(\mathbf{x}, x_3) := \frac{1}{hs} \Big( \mathbf{y}_h(\mathbf{x}, x_3 + s) -\mathbf{y}_h(\mathbf{x}, x_3)\Big).
\end{align*}
We first show that $\mathbf{f}_{h,s} \rightarrow \mathbf{b}$ in $L^2(\Omega')$ and that $(\nabla \mathbf{f}_{h,s}, \partial_3 \mathbf{f}_{h,s})$ weakly converges in $L^2(\Omega')$ to a function of $\mathbf{R}, \mathbf{S}, \boldsymbol{\tau}, \mathbf{n}$, thus allowing us to identify $\nabla \mathbf{b}$ on $\Omega'$. For this calculation, we find it useful to track the dependence on $x_3$ while suppressing the $\mathbf{x}$ dependence. Observe that 
\begin{equation}
\begin{aligned}\label{eq:firstConvergeFhs}
&\mathbf{f}_{h,s}(x_3) = \frac{1}{s} \int_{x_3}^{x_3+s}\frac{1}{h} \partial_3 \mathbf{y}_h(t) \dd t \rightarrow \mathbf{b} \quad \text{ in } L^2(\Omega') \\
&\partial_3\mathbf{f}_{h,s}( x_3) = \frac{1}{s}\big( \frac{1}{h} \partial_3 \mathbf{y}_h( x_3 + s) - \frac{1}{h} \partial_3 \mathbf{y}_h(x_3)\big) \rightarrow 0 \quad \text{ in } L^2(\Omega')
\end{aligned}
\end{equation}
since Proposition \ref{compactnessProp} gives $h^{-1} \partial_3 \mathbf{y}_h \rightarrow \mathbf{b}$ in $L^2(\Omega)$ for $\mathbf{b}$ independent of $x_3$. We also have that 
\begin{equation}
\begin{aligned}\label{eq:finiteDifference}
\partial_\alpha \mathbf{f}_{h,s}(x_3) &= \frac{1}{hs} \big( \partial_\alpha \mathbf{y}_h(x_3 + s) - \partial_\alpha \mathbf{y}_h(x_3)\big) \\
&= \frac{1}{hs} \Big( (\boldsymbol{\ell}_{\mathbf{n}_h}^{\text{f}})^{1/2}(x_3+s) \mathbf{R}_h\big( \mathbf{I} + h \mathbf{S}_h(x_3 + s) \big)- (\boldsymbol{\ell}_{\mathbf{n}_h}^{\text{f}})^{1/2}(x_3) \mathbf{R}_h\big( \mathbf{I} + h \mathbf{S}_h(x_3) \big) \Big) (\boldsymbol{\ell}_{\mathbf{n}_0}^0)^{-1/2} \mathbf{e}_{\alpha} \\
&= \Big( \frac{1}{hs}\big(( \boldsymbol{\ell}_{\mathbf{n}_h}^{\text{f}})^{1/2}(x_3 + s) - ( \boldsymbol{\ell}_{\mathbf{n}_h}^{\text{f}})^{1/2}(x_3) \big) \mathbf{R}_h\big( \mathbf{I} + h \mathbf{S}_h(x_3 + s) \big) \Big) (\boldsymbol{\ell}_{\mathbf{n}_0}^0)^{-1/2} \mathbf{e}_{\alpha} \\
&\qquad \qquad + \Big( (\boldsymbol{\ell}_{\mathbf{n}_h}^{\text{f}})^{1/2}(x_3) \mathbf{R}_h \frac{1}{s} \big(\mathbf{S}_h(x_3 + s) - \mathbf{S}_h(x_3) \big) \Big) (\boldsymbol{\ell}_{\mathbf{n}_0}^0)^{-1/2} \mathbf{e}_{\alpha}
\end{aligned}
\end{equation}
on $\Omega'$ for $\alpha =1,2$, since $\mathbf{S}_h := h^{-1}( \mathbf{R}_h^T (\boldsymbol{\ell}^{\text{f}}_{\mathbf{n}_h})^{-1/2} \nabla_h \mathbf{y}_h (\boldsymbol{\ell}^{0}_{\mathbf{n}_0})^{1/2} - \mathbf{I})$ and since $\mathbf{R}_h$, $(\boldsymbol{\ell}^{0}_{\mathbf{n}_0})^{-1/2}$ are independent of $x_3$. We now establish convergence properties for the terms on the right side of (\ref{eq:finiteDifference}). Observe that 
\begin{align*}
 \frac{1}{hs}\big(( \boldsymbol{\ell}_{\mathbf{n}_h}^{\text{f}})^{1/2}(x_3 + s) - ( \boldsymbol{\ell}_{\mathbf{n}_h}^{\text{f}})^{1/2}(x_3)\big) &= \frac{1}{s} \int_{x_3}^{x_3 +s} 2(\lambda_{\text{f}}^{1/2} - \lambda_{\text{f}}^{-1/4}) \text{sym}\big(h^{-1}\partial_3 \mathbf{n}_h(t) \otimes \mathbf{n}_h(t) \big) \dd t \\
 &\rightharpoonup 2(\lambda_{\text{f}}^{1/2} - \lambda_{\text{f}}^{-1/4}) \text{sym} \Big( \frac{1}{s} \int_{x_3}^{x_3+s} \boldsymbol{\tau}(t) \dd t \otimes \mathbf{n}\Big) \quad \text{ in } L^2(\Omega'),
\end{align*}
since $h^{-1} \partial_3 \mathbf{n}_h \rightharpoonup \boldsymbol{\tau}$ in $L^2(\Omega)$ and $\mathbf{n}_h \rightarrow \mathbf{n}$ in $L^2(\Omega)$ with $|\mathbf{n}_h| = 1$ a.e. Next, observe that 
\begin{align*}
 &\mathbf{R}_h\big( \mathbf{I} + h \mathbf{S}_h(x_3 + s) \big) = (\boldsymbol{\ell}^{\emph{f}}_{\mathbf{n}_h})^{-1/2}(x_3+s) \nabla_h \mathbf{y}_h(x_3 + s) (\boldsymbol{\ell}^{0}_{\mathbf{n}_0})^{1/2} \rightarrow \mathbf{R} \quad \text{ in } L^2(\Omega'), 
\end{align*}
by Proposition \ref{compactnessProp}. Next, since $\mathbf{R}_h \rightarrow \mathbf{R}$ in $L^2(\Omega)$ by Proposition \ref{compactnessProp} and $(\boldsymbol{\ell}_{\mathbf{n}_h}^{\text{f}})^{1/2}(x_3) \rightarrow (\boldsymbol{\ell}_{\mathbf{n}}^{\text{f}})^{1/2}$ in $L^2(\Omega)$ by Lemma \ref{weakConLemmaN} and both are bounded in $L^{\infty}$, 
\begin{align*}
(\boldsymbol{\ell}_{\mathbf{n}_h}^{\text{f}})^{1/2}(x_3) \mathbf{R}_h \rightarrow (\boldsymbol{\ell}_{\mathbf{n}}^{\text{f}})^{1/2}\mathbf{R} \quad \text{ in } L^2(\Omega').
\end{align*}
Finally, observe that 
\begin{equation}
\begin{aligned}\label{eq:ShWeakConv}
 \frac{1}{s} \big(\mathbf{S}_h(x_3 + s) - \mathbf{S}_h(x_3) \big) \rightharpoonup \frac{1}{s}\big( \mathbf{S}(x_3 + s) - \mathbf{S}(x_3)\big) \quad \text{ in } L^2(\Omega'). 
\end{aligned}
\end{equation}
since $\mathbf{S}_h \rightharpoonup \mathbf{S}$ in $L^2(\Omega)$ by Proposition \ref{compactnessProp}. It follows from (\ref{eq:finiteDifference}-\ref{eq:ShWeakConv}) that 
\begin{equation}
\begin{aligned}\label{eq:lastConvergeFhs}
\partial_\alpha \mathbf{f}_{h,s}(x_3) \rightharpoonup & \Big[ c_\text{f}\; \text{sym} \Big( \frac{1}{s} \int_{x_3}^{x_3+s} \boldsymbol{\tau}(t) \dd t \otimes \mathbf{n}\Big) \mathbf{R} + (\boldsymbol{\ell}_{\mathbf{n}}^{\text{f}})^{1/2}\mathbf{R}\frac{1}{s}\big( \mathbf{S}(x_3 + s) - \mathbf{S}(x_3)\big) \Big] (\boldsymbol{\ell}^{0}_{\mathbf{n}_0})^{-1/2} \mathbf{e}_{\alpha} \quad \text{ in } L^1(\Omega')
\end{aligned}
\end{equation}\
for $\alpha = 1,2$, where $c_\text{f} := 2(\lambda_{\text{f}}^{1/2} - \lambda_{\text{f}}^{-1/4})$. By (\ref{eq:firstConvergeFhs}) and (\ref{eq:lastConvergeFhs}), we identify $\nabla \mathbf{b}$ as 
\begin{align*}
\nabla \mathbf{b} = \Big[ c_{\text{f}}\; \text{sym} \Big( \frac{1}{s} \int_{x_3}^{x_3+s} \boldsymbol{\tau}(t) \dd t \otimes \mathbf{n}\Big) \mathbf{R} + (\boldsymbol{\ell}_{\mathbf{n}}^{\text{f}})^{1/2} \mathbf{R} \frac{1}{s}\big( \mathbf{S}(x_3 + s) - \mathbf{S}(x_3)\big) \Big] \big[(\boldsymbol{\ell}^{0}_{\mathbf{n}_0})^{-1/2} \big]_{3 \times2} \quad \text{ a.e. on $\Omega'$}.
\end{align*}

We now rearrange this formula to isolate the first two columns of the strain measure $\mathbf{S}$ in the liminf inequality in (\ref{eq:firstLiminf}). First observe that 
\begin{align*}
&\text{sym} \Big( \frac{1}{s} \int_{x_3}^{x_3+s} \boldsymbol{\tau}(t) \dd t \otimes \mathbf{n}\Big) \mathbf{R} = \mathbf{R} \text{sym} \Big( \frac{1}{s} \int_{x_3}^{x_3+s} \mathbf{R}^T \boldsymbol{\tau}(t) \dd t \otimes \mathbf{R}^T \mathbf{n}\Big) \quad \text{ and } \quad (\boldsymbol{\ell}_{\mathbf{n}}^{\text{f}})^{1/2} \mathbf{R} = \mathbf{R} (\boldsymbol{\ell}_{\mathbf{n}_0}^{\text{f}})^{1/2} 
\end{align*}
a.e. on $\Omega'$, where $\boldsymbol{\ell}_{\mathbf{v}_0}^{\text{f}}$ is as defined below (\ref{eq:mixedStepLength}) in Appendix \ref{sec:linAlgebra}. Thus, the identity $\mathbf{A} [(\boldsymbol{\ell}^{0}_{\mathbf{v}_0})^{-1/2} ]_{3 \times2} [(\boldsymbol{\ell}^{0}_{\mathbf{v}_0})^{1/2} ]_{2 \times2} = [ \mathbf{A}]_{3\times2}$ for all $\mathbf{A} \in \mathbb{R}^{3\times3}$ and $\mathbf{v}_0 \in \mathbb{S}^1$ gives that 
\begin{equation}
\begin{aligned}\label{eq:almostDoneWithS}
 \Big[ \frac{1}{s}\big( \mathbf{S}(x_3 + s) - \mathbf{S}(x_3)\big) \Big]_{3\times2} = (\boldsymbol{\ell}_{\mathbf{n}_0}^{\text{f}})^{-1/2} \Big( \mathbf{R}^T \nabla \mathbf{b} [(\boldsymbol{\ell}^{0}_{\mathbf{n}_0})^{1/2} ]_{2 \times2} - c_{\text{f}}\; \Big[\text{sym} \Big( \frac{1}{s} \int_{x_3}^{x_3+s} \mathbf{R}^T \boldsymbol{\tau}(t) \dd t \otimes \mathbf{R}^T \mathbf{n}\Big)\Big]_{3\times2} \Big) 
\end{aligned}
\end{equation}
a.e.\;on $\Omega'$. Note that $\int_{\Omega'} |2(\boldsymbol{\ell}_{\mathbf{n}_0}^{\text{f}})^{-1/2} \text{sym} \big( \mathbf{R}^T \boldsymbol{\tau}(t) \otimes \mathbf{R}^T \mathbf{n} \big)|^2 \dd V \leq \| (\boldsymbol{\ell}_{\mathbf{n}_0}^{\text{f}})^{-1/2}\|_{L^{\infty}(\Omega')}^2 \int_{\Omega'} 2 | \boldsymbol{\tau}(t)|^2 \dd V$ since $\mathbf{R}$ is a rotation field and $\mathbf{n}$ is a unit vector field. As such, $c_{\text{f}}(\boldsymbol{\ell}_{\mathbf{n}_0}^{\text{f}})^{-1/2}\text{sym} \big( \mathbf{R}^T \boldsymbol{\tau}(t) \dd t \otimes \mathbf{R}^T \mathbf{n} \big)$ belongs to $L^2(\Omega')$. Thus, $s^{-1}( \mathbf{S}(x_3 + s) - \mathbf{S}(x_3))$ is uniformly bounded in $L^2(\Omega')$ as $s\to0$ because the right side of (\ref{eq:almostDoneWithS}) is so. It follows that $\partial_3\mathbf{S}$ exists in $L^2(\Omega')$ as the weak limit of a subsequence of $s^{-1}( \mathbf{S}(x_3 + s) - \mathbf{S}(x_3))$, and is given by 
\begin{equation}
\begin{aligned}\label{eq:partial3S}
\big[\partial_3 \mathbf{S}(x_3) \big]_{3\times2} = (\boldsymbol{\ell}_{\mathbf{n}_0}^{\text{f}})^{-1/2} \Big( \mathbf{R}^T \nabla \mathbf{b} [(\boldsymbol{\ell}^{0}_{\mathbf{n}_0})^{1/2} ]_{2 \times2} - c_{\text{f}}\; \Big[\text{sym} \Big( \mathbf{R}^T \boldsymbol{\tau}(x_3) \otimes \mathbf{R}^T \mathbf{n}\Big)\Big]_{3\times2} \Big) 
\end{aligned}
\end{equation}
a.e.\;on $\Omega'$ as a consequence of Lebesgue differentiation theorem. 

As (\ref{eq:partial3S}) holds for any compact subset $\Omega'$ of $\Omega$, it holds on $\Omega$. The identity in (\ref{eq:S3by2}) follows after integrating up the indentity (\ref{eq:partial3S}) in $x_3$. 
\end{proof}

\subsection{Optimization of strain and the proof of Theorem \ref{MainTheorem}} Notice that the strain $\mathbf{S}$ depends on the behavior of the director in the $x_3$ direction through $\boldsymbol{\tau}$. We now depart from the arguments in \cite{bhattacharya2016plates, friesecke2002theorem, lewicka2011scaling} by carefully optimizing $\boldsymbol{\tau}$ to eliminate this $x_3$ dependence. This optimization refines the lowerbound, achieving the desired result in Theorem \ref{MainTheorem}. The main ingredient is an analytical solution to a 1D calculus of variations problem for the even part of the planar component $\boldsymbol{\tau}$ not parallel to $\mathbf{n}_0$ (see Lemma \ref{infimumLemma}) .
% Next 

\subsubsection{Lower bound in terms of the even part of $\boldsymbol{\tau}$}
First,
we establish a bound on the right side of (\ref{eq:firstLiminf}) in Lemma \ref{firstLiminfLemma}. 
\begin{lem}\label{secondLBIneqLemma}
The limiting fields obey the energy inequality
\begin{equation}
\begin{aligned}\label{eq:boundBelow2}
\int_{\Omega}\Big\{\frac{1}{2} Q_3(\mathbf{S}) + \gamma( |\nabla \mathbf{n}|^2 + |\boldsymbol{\tau}|^2) \Big\} \dd V \geq \int_{\Omega} \Big\{ \frac{1}{2}Q_2\Big( x_3 \mathbf{S}_1 + \mathcal{S}_2 \cdot \int_{0}^{x_3} \boldsymbol{\tau}^{\emph{e}} \dd t \Big) + \gamma( |\nabla \mathbf{n}|^2 + |\boldsymbol{\tau}^{\emph{e}}|^2 ) \Big\} \dd V 
\end{aligned}
\end{equation}
where $Q_2(\mathbf{A}) := 2\mu \Big\{ \big| \emph{sym} \big( [ \mathbf{A}]_{2\times2}\big)|^2 + \emph{Tr} \big( \emph{sym}( [ \mathbf{A}]_{2\times2} ) \big)^2\Big\}$ for any $\mathbf{A} \in \mathbb{R}^{3\times2}$ and 
\begin{equation}
\begin{aligned}\label{eq:strainDefsLB}
&x_3 \mathbf{S}_1(\mathbf{x}) := x_3 (\boldsymbol{\ell}_{\mathbf{n}_0}^{\emph{f}})^{-1/2}(\mathbf{x}) \mathbf{R}^T(\mathbf{x}) \nabla \mathbf{b}(\mathbf{x}) [(\boldsymbol{\ell}^{0}_{\mathbf{n}_0})^{1/2}(\mathbf{x}) ]_{2 \times2}, \\
&\mathcal{S}_2(\mathbf{x}) \cdot \int_{0}^{x_3} \boldsymbol{\tau}^{\emph{e}}(\mathbf{x}, t) \dd t := - c_{\emph{f}}\;(\boldsymbol{\ell}_{\mathbf{n}_0}^{\emph{f}})^{-1/2}(\mathbf{x}) \Big[\emph{sym} \Big( \mathbf{R}^T(\mathbf{x}) \int_{0}^{x_3} \boldsymbol{\tau}^{\emph{e}}(\mathbf{x},t) \dd t \otimes \mathbf{R}^T(\mathbf{x}) \mathbf{n}(\mathbf{x}) \Big)\Big]_{3\times2} \\
&\boldsymbol{\tau}^{\emph{e}}(\mathbf{x}, x_3) := \frac{1}{2}\big( \boldsymbol{\tau}(\mathbf{x}, x_3) + \boldsymbol{\tau}(\mathbf{x}, -x_3)\big).
\end{aligned}
\end{equation}
\end{lem}
\begin{proof}
In Lemma \ref{Q3toQ2Lemma} of Appendix \ref{sec:linAlgebra}, we establish that $Q_3(\mathbf{A}) \geq Q_2( [ \mathbf{A}]_{3\times2})$ for all $\mathbf{A} \in \mathbb{R}^{3\times3}$ that satisfy $\text{Tr}(\mathbf{A}) = 0$. As Proposition \ref{compactnessProp} gives $\text{Tr}(\mathbf{S}) = 0$ a.e., we deduce that 
 \begin{equation}
\begin{aligned}\label{eq:firstQ2Inequality}
\int_{\Omega} \frac{1}{2} Q_3\big(\mathbf{S}(\mathbf{x}, x_3) \big)\dd V &\geq \int_{\Omega} \frac{1}{2} Q_2\big(\big[\mathbf{S}(\mathbf{x}, x_3)\big]_{3\times2} \big)\dd V.
\end{aligned}
\end{equation}
Observe that the expression for the first two columns of $\mathbf{S}$ in (\ref{eq:S3by2}) of Lemma \ref{identifySLemma} is of the form 
\begin{equation}
\begin{aligned}\label{eq:Sform}
\big[\mathbf{S}(\mathbf{x}, x_3) \big]_{3\times2} = \mathbf{S}_0(\mathbf{x}) + x_3 \mathbf{S}_1(\mathbf{x}) + \mathcal{S}_2(\mathbf{x}) \cdot \boldsymbol{i}_{\boldsymbol{\tau}} (\mathbf{x}, x_3), \quad \boldsymbol{i}_{\boldsymbol{\tau}} (\mathbf{x}, x_3) := \int_{0}^{x_3} \boldsymbol{\tau}(\mathbf{x}, t)\dd t
\end{aligned}
\end{equation}
for $\mathbf{S}_{0,1} \in L^2(\omega, \mathbb{R}^{3\times2})$ and a third order tensor $\mathcal{S}_2 \in L^{\infty}(\omega, \mathbb{R}^{3\times2 \times 3})$. We bound the right side of (\ref{eq:firstQ2Inequality}) from below by exploiting the form of $x_3$ dependence in (\ref{eq:Sform}). Let $\overline{\mathbf{S}}(\mathbf{x}) := \int_{-1/2}^{1/2} \big[\mathbf{S}(\mathbf{x}, x_3) \big]_{3\times2} \dd x_3$. Since $\big[\mathbf{S}(\mathbf{x}, x_3)\big]_{3\times2} = \big[\mathbf{S}(\mathbf{x}, x_3)\big]_{3\times2} - \overline{\mathbf{S}}(\mathbf{x}) + \overline{\mathbf{S}}(\mathbf{x})$, it follows that 
\begin{align*}
 \int_{\Omega} \frac{1}{2} Q_2\big(\big[\mathbf{S}(\mathbf{x}, x_3)\big]_{3\times2} \big)\dd V &= \int_{\Omega} \frac{1}{2} \Big\{ Q_2\Big(\big[\mathbf{S}(\mathbf{x}, x_3)\big]_{3\times2} - \overline{\mathbf{S}}(\mathbf{x}) \Big) + Q_2\Big( \overline{\mathbf{S}}(\mathbf{x}) \Big) \Big\} \dd V \\
 &\geq \int_{\Omega} \frac{1}{2} Q_2\Big(\big[\mathbf{S}(\mathbf{x}, x_3)\big]_{3\times2} - \overline{\mathbf{S}}(\mathbf{x}) \Big) \dd V \\
 &= \int_{\Omega} \frac{1}{2} Q_2 \Big( x_3 \mathbf{S}_1(\mathbf{x}) + \mathcal{S}_2(\mathbf{x}) \cdot \big( \boldsymbol{i}_{\boldsymbol{\tau}}(\mathbf{x}, x_3) - \overline{\boldsymbol{i}_{\boldsymbol{\tau}}}(\mathbf{x})\big) \Big) \dd V 
\end{align*}
where $\overline{\boldsymbol{i}_{\boldsymbol{\tau}}}(\mathbf{x}) := \int_{-1/2}^{1/2} \boldsymbol{i}_{\boldsymbol{\tau}}(\mathbf{x},x_3) \dd x_3$. The cross-terms that arise when expanding out the quadratic form are not present in first equality above since they are odd functions of $x_3$ and thus vanish when integrating through the thickness. Now break $\boldsymbol{i}_{\boldsymbol{\tau}}(\mathbf{x},x_3)$ up into its even and odd parts via $\boldsymbol{i}^{\text{o}}_{\boldsymbol{\tau}}(\mathbf{x},x_3) := \frac{1}{2} (\boldsymbol{i}_{\boldsymbol{\tau}}(\mathbf{x},x_3) - \boldsymbol{i}_{\boldsymbol{\tau}}(\mathbf{x},-x_3))$ and $\boldsymbol{i}^{\text{e}}_{\boldsymbol{\tau}}(\mathbf{x},x_3) := \frac{1}{2} (\boldsymbol{i}_{\boldsymbol{\tau}}(\mathbf{x},x_3) + \boldsymbol{i}_{\boldsymbol{\tau}}(\mathbf{x},-x_3))$, respectively. Since $\boldsymbol{i}_{\boldsymbol{\tau}}(\mathbf{x},x_3) = \boldsymbol{i}^{\text{o}}_{\boldsymbol{\tau}}(\mathbf{x},x_3) + \boldsymbol{i}^{\text{e}}_{\boldsymbol{\tau}}(\mathbf{x},x_3)$, 
\begin{equation}
\begin{aligned}\label{eq:lastQ2Inequality}
&\int_{\Omega} \frac{1}{2} Q_2 \Big( x_3 \mathbf{S}_1(\mathbf{x}) + \mathcal{S}_2(\mathbf{x}) \cdot \big( \boldsymbol{i}_{\boldsymbol{\tau}}(\mathbf{x}, x_3) - \overline{\boldsymbol{i}_{\boldsymbol{\tau}}}(\mathbf{x})\big) \Big) \dd V \\
&\qquad = \int_{\Omega} \frac{1}{2} \Big\{ Q_2 \Big( x_3 \mathbf{S}_1(\mathbf{x}) + \mathcal{S}_2(\mathbf{x}) \cdot \boldsymbol{i}^{\text{o}}_{\boldsymbol{\tau}}(\mathbf{x}, x_3) \Big) + Q_2 \Big( \mathcal{S}_2(\mathbf{x}) \cdot \big( \boldsymbol{i}^{\text{e}}_{\boldsymbol{\tau}}(\mathbf{x}, x_3) - \overline{\boldsymbol{i}_{\boldsymbol{\tau}}}(\mathbf{x}) \big) \Big) \Big\} \dd V \\
&\qquad \geq \int_{\Omega} \frac{1}{2} Q_2 \Big( x_3 \mathbf{S}_1(\mathbf{x}) + \mathcal{S}_2(\mathbf{x}) \cdot \boldsymbol{i}^{\text{o}}_{\boldsymbol{\tau}}(\mathbf{x}, x_3) \Big) \dd V,
\end{aligned}
\end{equation}
where again the cross-terms vanish on integration since the product of an even and odd function in $x_3$ is an odd function. Note that $\boldsymbol{i}^{\text{o}}_{\boldsymbol{\tau}}(\mathbf{x}, x_3) = \int_{0}^{x_3} \boldsymbol{\tau}^{\text{e}}(\mathbf{x}, t) \dd t$ for $\boldsymbol{\tau}^{\text{e}}$ the even part of $\boldsymbol{\tau}$ defined in (\ref{eq:strainDefsLB}). Thus, 
\begin{equation}
\begin{aligned}\label{eq:lastQ2}
\int_{\Omega} \frac{1}{2} Q_3\big(\mathbf{S}(\mathbf{x}, x_3) \big) \dd V \geq \int_{\Omega} \frac{1}{2} Q_2 \Big( x_3 \mathbf{S}_1(\mathbf{x}) + \mathcal{S}_2(\mathbf{x}) \cdot \int_0^{x_3} \boldsymbol{\tau}^{\text{e}}(\mathbf{x}, t)\dd t \Big) \dd V
\end{aligned}
\end{equation}
using (\ref{eq:firstQ2Inequality}-\ref{eq:lastQ2Inequality}). The formulas in (\ref{eq:strainDefsLB}) follow from matching the terms in (\ref{eq:Sform}) with that of Lemma \ref{identifySLemma} and applying the various definitions above. To complete the proof, we simply note that
\begin{equation}
\begin{aligned}\label{eq:ineqTauE}
\int_{\Omega} \gamma |\boldsymbol{\tau}(\mathbf{x}, x_3)|^2 \dd V = \int_{\Omega} \gamma \big\{ | \boldsymbol{\tau}^{\text{e}}(\mathbf{x}, x_3)|^2 + | \boldsymbol{\tau}^{\text{o}}(\mathbf{x}, x_3)|^2 \big\} \dd V \geq \int_{\Omega} \gamma | \boldsymbol{\tau}^{\text{e}}(\mathbf{x}, x_3)|^2 \dd V,
\end{aligned}
\end{equation}
 since $\boldsymbol{\tau}(\mathbf{x}, x_3) = \boldsymbol{\tau}^{\text{e}}(\mathbf{x}, x_3) + \boldsymbol{\tau}^{\text{o}}(\mathbf{x},x_3)$ for its even and odd parts and since the cross-terms once again vanish on integration. The inequalities in (\ref{eq:lastQ2}) and (\ref{eq:ineqTauE}) imply (\ref{eq:boundBelow2}).
\end{proof}

We now simplify the expressions for the strain measures from the prior lemma. To do so, we exploit the fact that the quadratic form $Q_2(\mathbf{A})$, $\mathbf{A} \in \mathbb{R}^{3\times2}$, only depends on the symmetric part of its principal $2\times2$ submatrix $\text{sym} ([ \mathbf{A}]_{2\times2})$; we also make use of the fact that $\boldsymbol{\tau}^{\emph{e}}$ is parameterized without loss of generality as 
\begin{equation}
\begin{aligned}\label{eq:tauEparam}
\boldsymbol{\tau}^{\text{e}}(\mathbf{x}, x_3) := \sigma \Big( \tau^{\text{e}}_{\perp}(\mathbf{x}, x_3) \mathbf{R}(\mathbf{x}) \begin{pmatrix} \mathbf{n}_0^{\perp}(\mathbf{x}) \\ 0 \end{pmatrix} + \tau^{\text{e}}_3(\mathbf{x}, x_3) \mathbf{R}(\mathbf{x}) \mathbf{e}_3\Big)
\end{aligned}
\end{equation}
a.e\;on $\Omega$ for some scalar fields $\tau^{\text{e}}_{\perp}, \tau^{\text{e}}_{3} \in L^2(\Omega)$ that are even functions of $x_3$. 
(The latter follows by Proposition \ref{compactnessProp}, namely, because $\boldsymbol{\tau}^{\text{e}}$ and $\mathbf{n}$ satisfy the conditions $\boldsymbol{\tau}^{\text{e}} \cdot \mathbf{n} = 0$ and $\mathbf{n} = \sigma \mathbf{R}\begin{pmatrix} \mathbf{n}_0 \\ 0 \end{pmatrix}$ a.e.\;on $\Omega$.)

\begin{lem}\label{get2x2StrainsLemma}
The expressions in (\ref{eq:strainDefsLB}) satisfy 
\begin{equation}
\begin{aligned}\label{eq:getProjectStrains}
\emph{sym} \Big( \big[ x_3 \mathbf{S}_1 + \mathcal{S}_2 \cdot \int_{0}^{x_3} \boldsymbol{\tau}^e \dd t \big]_{2\times2}\Big) &= x_3 \lambda_{\emph{f}}^{-1/4} \lambda_0^{1/4} \emph{sym} \big([( \boldsymbol{\ell}_{\mathbf{n}_0}^{\emph{f}})^{-1} ( \boldsymbol{\ell}_{\mathbf{n}_0}^{\emph{0}})^{1/2}]_{2\times2} \mathbf{II}_{\mathbf{y}} [( \boldsymbol{\ell}_{\mathbf{n}_0}^{\emph{0}})^{1/2}]_{2\times2}\big) \\
&\qquad + ( \lambda_{\emph{f}}^{-3/4} - \lambda_{\emph{f}}^{3/4}) \Big( \int_0^{x_3} \tau_{\perp}^{\emph{e}} \dd t \Big) \emph{sym} \big( \mathbf{n}_0^{\perp} \otimes \mathbf{n}_0 \big),
\end{aligned} 
\end{equation}
a.e.\;on $\Omega$, where $\mathbf{II}_{\mathbf{y}}$ is the second fundamental form defined in (\ref{eq:secFund}). 
\end{lem}
\begin{proof}
All stated identities in this proof hold pointwise a.e.\;on $\Omega$; this notation is suppressed below for conciseness. Per Remark \ref{lotsOfIdentsRem}, $\mathbf{R}^T = (\boldsymbol{\ell}_{\mathbf{n}_0}^{\text{f}})^{-1/2} (\boldsymbol{\ell}_{\mathbf{n}_0}^0)^{1/2} ( \nabla \mathbf{y}, \mathbf{b})^T$, so $\mathbf{S}_1$ in (\ref{eq:strainDefsLB}) satisfies 
\begin{equation}
\begin{aligned}\label{eq:S1FirstManip}
\mathbf{S}_1 = (\boldsymbol{\ell}_{\mathbf{n}_0}^{\text{f}})^{-1} (\boldsymbol{\ell}_{\mathbf{n}_0}^0)^{1/2} ( \nabla \mathbf{y}, \mathbf{b})^T \nabla \mathbf{b} [ (\boldsymbol{\ell}_{\mathbf{n}_0}^0)^{1/2} ]_{2\times2} .
\end{aligned}
\end{equation}
 Now observe that $(\nabla \mathbf{y})^T\nabla \mathbf{b} = \lambda_{\text{f}}^{-1/4} \lambda_0^{1/4} \mathbf{II}_{\mathbf{y}}$ and $\mathbf{b} \cdot \nabla \mathbf{b} = 0$, since $\mathbf{b} = \lambda_{\text{f}}^{-1/4} \lambda_0^{1/4} \boldsymbol{\nu}_{\mathbf{y}}$ for the surface normal define in (\ref{eq:secFund}). It follow that (\ref{eq:S1FirstManip}) can be written as 
 \begin{equation}
 \begin{aligned}\label{eq:S1SecManip}
 \mathbf{S}_1 = (\boldsymbol{\ell}_{\mathbf{n}_0}^{\text{f}})^{-1} (\boldsymbol{\ell}_{\mathbf{n}_0}^0)^{1/2} \begin{pmatrix} \lambda_{\text{f}}^{-1/4} \lambda_0^{1/4} \mathbf{II}_{\mathbf{y}} \\ [\mathbf{0}]_{2\times1} \end{pmatrix} [ (\boldsymbol{\ell}_{\mathbf{n}_0}^0)^{1/2} ]_{2\times2} = \lambda_{\text{f}}^{-1/4} \lambda_0^{1/4} \begin{pmatrix}
[ (\boldsymbol{\ell}_{\mathbf{n}_0}^{\text{f}})^{-1} (\boldsymbol{\ell}_{\mathbf{n}_0}^0)^{1/2} ]_{2\times2} \mathbf{II}_{\mathbf{y}} [ (\boldsymbol{\ell}_{\mathbf{n}_0}^0)^{1/2} ]_{2\times2} \\ [\mathbf{0}]_{2\times1} 
 \end{pmatrix} 
 \end{aligned}
 \end{equation}
 because $\mathbf{n}_0$ is a planar director field. Next we use the parameterization of $\boldsymbol{\tau}^{\text{e}}$ in (\ref{eq:tauEparam}) and that of $\mathbf{n} = \sigma \mathbf{R}\begin{pmatrix} \mathbf{n}_0 \\ 0 \end{pmatrix}$ from Proposition \ref{compactnessProp} to furnish the identities 
 \begin{equation}
 \begin{aligned}\label{eq:S2FirstManip}
 \mathcal{S}_2 \cdot \int_0^{x_3} \boldsymbol{\tau}^{\text{e}} \dd t &= -c_{\text{f}} (\boldsymbol{\ell}_{\mathbf{n}_0}^{\text{f}})^{-1/2} \text{sym}\Big( \begin{pmatrix} \int_{0}^{x_3} \tau^{\text{e}}_{\perp} \dd t \mathbf{n}_0^{\perp} \\ \int_{0}^{x_3} \tau^{\text{e}}_{3}\dd t \end{pmatrix} \otimes \begin{pmatrix} \mathbf{n}_0 \\ 0 \end{pmatrix} \Big) \\
 &= -\frac{c_{\text{f}}}{2}\Big[ \lambda_{\text{f}}^{1/4} \begin{pmatrix} \int_{0}^{x_3} \tau^{\text{e}}_{\perp} \dd t \mathbf{n}_0^{\perp} \\ \int_{0}^{x_3} \tau^{\text{e}}_{3}\dd t \end{pmatrix} \otimes \begin{pmatrix} \mathbf{n}_0 \\ 0 \end{pmatrix} + \lambda_{\text{f}}^{-1/2} \begin{pmatrix} \mathbf{n}_0 \\ 0 \end{pmatrix} \otimes \begin{pmatrix} \int_{0}^{x_3} \tau^{\text{e}}_{\perp} \dd t \mathbf{n}_0^{\perp} \\ \int_{0}^{x_3} \tau^{\text{e}}_{3}\dd t \end{pmatrix} \Big].
 \end{aligned}
 \end{equation}
The result in (\ref{eq:getProjectStrains}) clearly follows from (\ref{eq:S1SecManip}) and (\ref{eq:S2FirstManip}) since $c_{\text{f}} = 2 (\lambda_{\text{f}}^{1/2} - \lambda_{\text{f}}^{-1/4}) $.
\end{proof}

We now bound the energy from below in terms of the scalar field $\tau_{\perp}^{\text{e}}$ in (\ref{eq:tauEparam}), the gradient of the convected director $\nabla \mathbf{y} \mathbf{n}_0$, and the $2 \times2$ symmetric strain tensors
\begin{equation}
\begin{aligned}\label{eq:E1E2Strains}
&\mathbf{E}_{1}(\mathbf{y}, \mathbf{n}_0) := \lambda_{\text{f}}^{-1/4} \lambda_0^{1/4} \text{sym} \big([( \boldsymbol{\ell}_{\mathbf{n}_0}^{\text{f}})^{-1} ( \boldsymbol{\ell}_{\mathbf{n}_0}^{0})^{1/2}]_{2\times2} \mathbf{II}_{\mathbf{y}} [( \boldsymbol{\ell}_{\mathbf{n}_0}^{\emph{0}})^{1/2}]_{2\times2}\big), \\
&\mathbf{E}_{2}(\mathbf{n}_0) := ( \lambda_{\text{f}}^{-3/4} - \lambda_{\text{f}}^{3/4}) \text{sym} \big( \mathbf{n}_0^{\perp} \otimes \mathbf{n}_0 \big).
\end{aligned}
\end{equation}
\begin{lem}\label{thirdLBIneqLemma}
The limiting fields obey the energy inequality
\begin{equation}
\begin{aligned}\label{eq:LBInequalityAlmostFinal}
 &\int_{\Omega} \Big\{ \frac{1}{2}Q_2\Big( x_3 \mathbf{S}_1 + \mathcal{S}_2 \cdot \int_{0}^{x_3} \boldsymbol{\tau}^{\emph{e}} \dd t \Big) + \gamma( |\nabla \mathbf{n}|^2 + |\boldsymbol{\tau}^{\emph{e}}|^2 ) \Big\} \dd V \\
& \quad \geq \int_{\Omega} \Big\{ \mu \Big| x_3 \mathbf{E}_1(\mathbf{y}, \mathbf{n}_0) + \int_0^{x_3} \tau^{\emph{e}}_{\perp} \dd t \mathbf{E}_2(\mathbf{n}_0) \Big|^2 + \gamma (\tau_{\perp}^{\emph{e}})^2 + \mu \emph{Tr}\big( x_3 \mathbf{E}_1(\mathbf{y}, \mathbf{n}_0)\big)^2 + \gamma \lambda_{\emph{f}}^{-1} \lambda_0 \big| \nabla(\nabla \mathbf{y} \mathbf{n}_0) \big|^2 \Big\} \dd V
\end{aligned}
\end{equation}
\end{lem} 
\begin{proof}
From Lemma \ref{get2x2StrainsLemma} and the definitions in (\ref{eq:E1E2Strains}), $\text{sym}(\big[x_3 \mathbf{S}_1 + \mathcal{S}_2 \cdot \int_{0}^{x_3} \boldsymbol{\tau}^{\emph{e}} \dd t\big]_{2\times2}) = x_3 \mathbf{E}_1(\mathbf{y}, \mathbf{n}_0) +( \int_{0}^{x_3} \tau_{\perp}^{\text{e}}\dd t ) \mathbf{E}_2( \mathbf{n}_0)$ a.e.\;on $\Omega$. Thus, the definition of $Q_2(\mathbf{A})$ in Lemma \ref{Q3toQ2Lemma} gives that 
\begin{equation}
\begin{aligned}\label{eq:firstLBIneqLemm}
&\int_{\Omega} \frac{1}{2}Q_2\Big( x_3 \mathbf{S}_1 + \mathcal{S}_2 \cdot \int_{0}^{x_3} \boldsymbol{\tau}^{\emph{e}} \dd t \Big) \dd V = \int_{\Omega} \mu \Big\{ \Big| x_3 \mathbf{E}_1(\mathbf{y}, \mathbf{n}_0) + \int_0^{x_3} \tau^{\emph{e}}_{\perp} \dd t \mathbf{E}_2(\mathbf{n}_0) \Big|^2 + \text{Tr}\big( x_3 \mathbf{E}_1(\mathbf{y}, \mathbf{n}_0)\big)^2 \Big\} \dd V
\end{aligned}
\end{equation}
We now consider the other terms in the energy. First observe that the general parameterization in (\ref{eq:tauEparam}) furnishes the inequality 
\begin{align*}
\int_{\Omega} \gamma |\boldsymbol{\tau}^{\text{e}}|^2 \dd V = \int_{\Omega} \gamma \big| \tau^{\text{e}}_{\perp} \begin{pmatrix} \mathbf{n}_0^{\perp} \\ 0 \end{pmatrix} + \tau_3^{\text{e}} \mathbf{e}_3 \big|^2 \dd V = \int_{\Omega} \gamma \big\{ ( \tau^{\text{e}}_{\perp})^2 + ( \tau^{\text{e}}_{3})^2 \big\} \dd V \geq \int_{\Omega} \gamma ( \tau^{\text{e}}_{\perp})^2 \dd V
\end{align*}
since $\mathbf{R}$ is a rotation field and $\mathbf{n}_0^{\perp}$ and $\mathbf{e}_3$ are perpendicular with unit magnitude. For the remaining term in the energy, observe that $|\nabla \mathbf{y} \mathbf{n}_0|^2 = \mathbf{n}_0 \cdot \mathbf{g}_{\mathbf{n}_0} \mathbf{n}_0 = \lambda_{\text{f}} \lambda_0^{-1}$ a.e.\;on $\Omega$ by (\ref{eq:identsOnPlanarFields}). Thus, 
\begin{equation}
\begin{aligned}\label{eq:lastLBIneqLemm}
\int_{\Omega}\gamma |\nabla \mathbf{n}|^2 \dd V = \int_{\Omega} \gamma | \nabla \Big( \frac{\nabla \mathbf{y} \mathbf{n}_0}{|\nabla \mathbf{y} \mathbf{n}_0|} \Big)|^2 \dd V = \int_{\Omega}\gamma \lambda_{\text{f}}^{-1} \lambda_0 \big| \nabla ( \nabla \mathbf{y} \mathbf{n}_0) \big|^2 \dd V 
\end{aligned}
\end{equation}
from (\ref{eq:identsOnPlanarFields}) in Proposition \ref{compactnessProp}. The result in (\ref{eq:LBInequalityAlmostFinal}) follows from (\ref{eq:firstLBIneqLemm}-\ref{eq:lastLBIneqLemm}).
\end{proof}

\subsubsection{Optimization of $\tau^e_\perp$}
The next major step in the lowerbound is to minimize the energy on the right of (\ref{eq:LBInequalityAlmostFinal}) amongst all $\tau_{\perp}^{\text{e}} \in L^2(\Omega)$ that are even functions of $x_3$. As a preliminary, we consider the energy functional 
\begin{align*}
E_0(f) := \int_{-1/2}^{1/2} \big\{ \mu \big| t \mathbf{E}_1 + f(t) \mathbf{E}_2 \big|^2 + \gamma (f'(t))^2 \big\} \dd t , \quad \mathbf{E}_{1,2} \in \mathbb{R}^{2\times2}_{\text{sym}}, \quad \mathbf{E}_2 \neq \mathbf{0},
\end{align*}
in the scalar fields $f \colon (-1/2,1/2) \rightarrow \mathbb{R}$ and seek the infimum of this energy with respect to odd functions 
\begin{align*}
E_0^{\star} = \inf \big\{ E_0(f) \colon f \in H^1((-1/2,1/2)), f(-t) = f(t) \text{ a.e.} \big\}.
\end{align*}
\begin{lem}\label{infimumLemma}
$E_0^{\star}$ satisfies 
\begin{align*}
E_0^{\star} &=\frac{\mu}{12} \Big( |\mathbf{E}_1|^2 - g \Big( \big(\tfrac{\mu}{\gamma}\big)^{1/2} |\mathbf{E}_2| \Big) \frac{( \mathbf{E}_1 \colon \mathbf{E}_2)^2}{|\mathbf{E}_2|^2} \Big)
\end{align*}
for a monotonically increasing function 
\begin{equation}
\begin{aligned}\label{eq:gLambda}
g(\alpha) := 1 - 12\alpha^{-2} + 24 \alpha^{-3} \tanh\Big( \frac{\alpha}{2}\Big) , \quad \alpha >0,
\end{aligned}
\end{equation}
with the limiting properties $\lim_{\alpha \rightarrow 0} g(\alpha) = 0$ and $\lim_{\alpha \rightarrow \infty} g(\alpha) = 1$.
The minimizer $f^{\star}$ to $E_0^{\star}$ is
\begin{equation}
\begin{aligned}\label{eq:fminimizer}
f^{\star}(t) =\Big( \frac{\mathbf{E}_1 \colon \mathbf{E}_2}{|\mathbf{E}_2|^2} \Big)\Big( \frac{\gamma}{\mu}\Big)^{1/2} \bigg[ \frac{\sinh \big(\big(\tfrac{\mu}{\gamma}\big)^{1/2} |\mathbf{E}_2| t\big)}{\cosh \big(\tfrac{1}{2}\big(\tfrac{\mu}{\gamma}\big)^{1/2} |\mathbf{E}_2| \big)} - \big(\tfrac{\mu}{\gamma}\big)^{1/2} |\mathbf{E}_2| t \bigg] 
\end{aligned}
\end{equation}
\end{lem}
\begin{proof}
First observe that $w(t) = \mu \big| t \mathbf{E}_1 + f(t) \mathbf{E}_2 \big|^2 + \gamma (f'(t))^2$ is an even function when $f(t)$ is odd. Thus, for any odd function $f\in H^1(-1/2,1/2)$ we have
\[
{E}_0(f) = 2\int_{0}^{1/2} \{\mu|t \mathbf{E}_1 +f(t) \mathbf{E}_2|^2 +\gamma(f'(t))^2 \} \dd t =: \tilde{E}_0(f).
\]
As a consequence, we have that $f \in H^1(-1/2,1/2) $ minimizes $E_0$ over all odd $H^1(-1/2,1/2)$ functions if and only if $f|_{(0,1/2)} \in H^1(0,1/2)$ minimizes $\tilde{E}_0$ subject to $f(0) = 0$. Furthermore, the minimum value of $E_0^{\star}$ satisfies
\[
E_0^{\star} = \inf \{ \tilde{E}_0(f) \colon f \in H^1(0,1/2), f(0) = 0\}.
\]

The first variation of the latter formulation of the energy furnishes the Euler-Lagrange equation, with Dirichlet and natural boundary conditions. The strong form reads 
\begin{align*}
\begin{cases}
\mu t \mathbf{E}_1\colon \mathbf{E}_2 + \mu f(t) |\mathbf{E}_2|^2 - \gamma f''(t) = 0 \\
f(0) = 0, \; f'(1/2) = 0
\end{cases}
\end{align*}
A particular solution to the Euler-Lagrange equation is $f_{\text{P}}(t) = -t |\mathbf{E}_2|^{-2} \mathbf{E}_1 \colon \mathbf{E}_2$, while the general solution to the homogeneous problem $\mu f(t) |\mathbf{E}_2|^2 - \gamma f''(t) = 0$ is $f_{\text{H}}(t) = A \exp (\big(\tfrac{\mu}{\gamma}\big)^{1/2} |\mathbf{E}_2| t) + B \exp (-\big(\tfrac{\mu}{\gamma}\big)^{1/2} |\mathbf{E}_2| t)$. Supplying the boundary conditions to $f^{\star}(t) = f_{\text{P}}(t) + f_{\text{H}}(t)$ gives the minimizer in (\ref{eq:fminimizer}). 

To calculate the energy, define $\alpha := \big(\tfrac{\mu}{\gamma}\big)^{1/2} |\mathbf{E}_2|$ and $f_{\alpha}(t) := \frac{1}{\alpha} \Big( \frac{\sinh (\alpha t)}{\cosh(\alpha/2)} - \alpha t \Big)$, and observe that the minimizer satisfies $f^{\star}(t) = |\mathbf{E}_2|^{-2} (\mathbf{E}_1 \colon \mathbf{E}_2) f_{\alpha}(t)$. It follows by expanding out $E_0^{\star} = \tilde{E}_0(f^{\star})$ that 
\begin{align*}
E_0^{\star} &= \frac{\mu}{12} |\mathbf{E}_1|^2 + \Big( \int_0^{1/2} \big\{ 2 t f_{\alpha}(t) + \big(f_{\alpha}(t)\big)^2 \big\} \dd t \Big) \Big( 2 \mu \frac{(\mathbf{E}_1 \colon \mathbf{E}_2)^2}{|\mathbf{E}_2|^2} \Big) + \Big( \int_0^{1/2} (f_{\alpha}'(t))^2 \dd t \Big) \Big(2 \gamma \frac{(\mathbf{E}_1 \colon \mathbf{E}_2)^2}{|\mathbf{E}_2|^4} \Big) \\ 
& =\frac{\mu}{12} |\mathbf{E}_1|^2 + \Big( \int_0^{1/2} \big\{ 2 t f_{\alpha}(t) + \big(f_{\alpha}(t)\big)^2 + \alpha^{-2} (f_{\alpha}'(t))^2 \big\} \dd t \Big) \Big(2 \mu \frac{(\mathbf{E}_1 \colon \mathbf{E}_2)^2}{|\mathbf{E}_2|^2} \Big)
\end{align*}
where the second equality combines the two integrals using the fact that $\gamma |\mathbf{E}_2|^{-2} = \mu \alpha^{-2}$. It is not hard to see that $g(\alpha)$ in (\ref{eq:gLambda}) is, in fact, given by the integration $g(\alpha) = - 24 \int_0^{1/2} \big\{ 2 t f_{\alpha}(t) + \big(f_{\alpha}(t)\big)^2 + \alpha^{-2} (f_{\alpha}'(t))^2 \big\} \dd t$. Its stated monotonicity and limit properties are also straightforward to verify. 
\end{proof}

Having quantitatively addressed the minimization problem in $\tau_{\perp}^{\text{e}}$ with the lemma above, we now attain the desired limiting energy as an ansatz-free lowerbound. 

\begin{lem}\label{finalLBIneqLemma}
The limiting fields obey the energy inequality
\begin{equation}
\begin{aligned}\label{eq:theFinalLowerbound}
& \int_{\Omega} \Big\{ \mu \Big| x_3 \mathbf{E}_1(\mathbf{y}, \mathbf{n}_0) + \int_0^{x_3} \tau^{\emph{e}}_{\perp} \dd t \mathbf{E}_2(\mathbf{n}_0) \Big|^2 + \gamma (\tau_{\perp}^{\emph{e}})^2 + \mu \emph{Tr}\big( x_3 \mathbf{E}_1(\mathbf{y}, \mathbf{n}_0)\big)^2 + \gamma \lambda_{\emph{f}}^{-1} \lambda_0 \big| \nabla(\nabla \mathbf{y} \mathbf{n}_0) \big|^2 \Big\} \dd V \\
 &\qquad \geq \int_{\omega} \Big\{ \mathbf{II}_{\mathbf{y}} \colon \mathbb{B}(\mathbf{n}_0) \colon \mathbf{II}_{\mathbf{y}} + \gamma \lambda_{\emph{f}}^{-1} \lambda_0 \big| \nabla(\nabla \mathbf{y} \mathbf{n}_0) \big|^2 \Big\} \dd A
\end{aligned}
\end{equation}
for $\mathbb{B}(\mathbf{v}_0)$ defined in (\ref{eq:Bv0Moduli}).
\end{lem}

\begin{proof} 

Let $E_1(\mathbf{y}, \mathbf{n}_0, \tau_{\perp}^{\text{e}}) := \int_{\Omega} \big\{ \mu \big| x_3 \mathbf{E}_1(\mathbf{y}, \mathbf{n}_0) + \int_0^{x_3} \tau^{\emph{e}}_{\perp} \dd t \mathbf{E}_2(\mathbf{n}_0) \big|^2 + \gamma (\tau_{\perp}^{\emph{e}})^2 \big\} \dd V$. We observe that for any $\psi^\text{e}\in C^0(\Omega)$ even in $x_3$, the function $i_{\psi^\text{e}}$ defined by $i_{\psi^\text{e}}(\mathbf{x}, x_3) := \int_0^{x_3} \psi^\text{e}(\mathbf{x}, t) \dd t$ is odd in $x_3$.
Hence,
\begin{equation*}
\begin{aligned} 
E_1(\mathbf{y}, \mathbf{n}_0, \psi^{\text{e}})
&= \int_\omega\int_{-1/2}^{1/2} \Big\{ \mu \Big| x_3 \mathbf{E}_1(\mathbf{y}, \mathbf{n}_0) + i_{\psi^{\text{e}}}\mathbf{E}_2(\mathbf{n}_0) \Big|^2 + \gamma |\partial_3 i_{\psi^{\text{e}}}|^2 \Big\} \dd x_3 \dd A\\
&\geq \int_\omega \inf_{ \substack{ f \in H^1((-1/2,1/2)) \\ f \text{ odd in $x_3$} } }\int_{-1/2}^{1/2} \Big\{ \mu \Big| x_3 \mathbf{E}_1(\mathbf{y}, \mathbf{n}_0) + f \mathbf{E}_2(\mathbf{n}_0) \Big|^2 + \gamma |\partial_3 f|^2 \Big\} \dd x_3 \dd A\\
& = \int_{\omega }\frac{\mu}{12} \Big( |\mathbf{E}_1(\mathbf{y}, \mathbf{n}_0) |^2 - g \Big( \big(\tfrac{\mu}{\gamma}\big)^{1/2} |\mathbf{E}_2(\mathbf{n}_0) | \Big) \frac{( \mathbf{E}_1(\mathbf{y}, \mathbf{n}_0) \colon \mathbf{E}_2(\mathbf{n}_0))^2}{|\mathbf{E}_2(\mathbf{n}_0)|^2} \Big) \dd A
\end{aligned}
\end{equation*}
for $g(\alpha)$ in (\ref{eq:gLambda}) due to Lemma \ref{infimumLemma}. By density of $C^0(\Omega)$ functions in $L^2(\Omega)$, the above inequality also holds for $\tau_{\perp}^{\text{e}}$, i.e., 
\begin{equation}\label{eq:getLBinfimum}
E_1(\mathbf{y}, \mathbf{n}_0, \tau_{\perp}^{\text{e}}) \geq \int_{\omega }\frac{\mu}{12} \Big( |\mathbf{E}_1(\mathbf{y}, \mathbf{n}_0) |^2 - g \Big( \big(\tfrac{\mu}{\gamma}\big)^{1/2} |\mathbf{E}_2(\mathbf{n}_0) | \Big) \frac{( \mathbf{E}_1(\mathbf{y}, \mathbf{n}_0) \colon \mathbf{E}_2(\mathbf{n}_0))^2}{|\mathbf{E}_2(\mathbf{n}_0)|^2} \Big) \dd A.
\end{equation}

Expanding out the first term in the energy density on the right using definition of the strain $\mathbf{E}_1(\mathbf{y}, \mathbf{n}_0)$ in (\ref{eq:E1E2Strains}) and the isotropy of $|\cdot|^2$ yields
\begin{equation}
\begin{aligned}\label{eq:dotN01}
 |\mathbf{E}_1(\mathbf{y}, \mathbf{n}_0) |^2 &= ( \mathbf{n}_0 \cdot \mathbf{E}_1(\mathbf{y}, \mathbf{n}_0) \mathbf{n}_0)^2 + ( \mathbf{n}^{\perp}_0 \cdot \mathbf{E}_1(\mathbf{y}, \mathbf{n}_0) \mathbf{n}^{\perp}_0))^2 + 2 ( \mathbf{n}^{\perp}_0 \cdot \mathbf{E}_1(\mathbf{y}, \mathbf{n}_0) \mathbf{n}_0)^2 \\
 &= \lambda_{\text{f}}^{-5/2} \lambda_0^{5/2} \big( \mathbf{n}_0 \cdot \mathbf{II}_{\mathbf{y}} \mathbf{n}_0\big)^2 + \lambda_{\text{f}}^{1/2} \lambda_0^{-1/2} \big( \mathbf{n}^{\perp}_0 \cdot \mathbf{II}_{\mathbf{y}} \mathbf{n}^{\perp}_0\big)^2 + \frac{1}{2} \lambda_0 (\lambda_\text{f}^{-5/4} + \lambda_\text{f}^{1/4} )^2 \big( \mathbf{n}_0 \cdot \mathbf{II}_{\mathbf{y}} \mathbf{n}^{\perp}_0\big)^2
\end{aligned}
\end{equation} 
Likewise, the definition of $\mathbf{E}_2(\mathbf{n}_0)$ in (\ref{eq:E1E2Strains}) yields $|\mathbf{E}_2(\mathbf{n}_0) | = \frac{1}{\sqrt{2}}(\lambda_{\text{f}}^{3/4} - \lambda_{\text{f}}^{-3/4})$. Thus,
\begin{equation}
\begin{aligned}\label{eq:dotN02}
g \Big( \big(\tfrac{\mu}{\gamma}\big)^{1/2} |\mathbf{E}_2(\mathbf{n}_0) | \Big) \frac{( \mathbf{E}_1(\mathbf{y}, \mathbf{n}_0) \colon \mathbf{E}_2(\mathbf{n}_0))^2}{|\mathbf{E}_2(\mathbf{n}_0)|^2} = g \Big(\big(\tfrac{\mu}{2\gamma}\big)^{1/2} (\lambda_{\text{f}}^{3/4} - \lambda_{\text{f}}^{-3/4})\Big) \frac{\lambda_0}{2} (\lambda_\text{f}^{-5/4} + \lambda_\text{f}^{1/4} )^2 ( \mathbf{n}_0 \cdot \mathbf{II}_{\mathbf{y}} \mathbf{n}_0^{\perp})^2.
\end{aligned}
\end{equation}
The remaining term in the energy that depends on $\mathbf{E}_1(\mathbf{y}, \mathbf{n}_0)$ satisfies 
\begin{equation}
\begin{aligned}\label{eq:integrateTrace}
\int_{\Omega} \mu \text{Tr} \big( x_3 \mathbf{E}_1(\mathbf{y}, \mathbf{n}_0)\big)^2 \dd V = \int_{\omega} \frac{\mu}{12} \text{Tr} \big( \mathbf{E}_1(\mathbf{y}, \mathbf{n}_0)\big)^2 \dd A
\end{aligned}
\end{equation} 
Its energy density is 
\begin{equation}
\begin{aligned}\label{eq:dotN03}
\text{Tr} \big( \mathbf{E}_1(\mathbf{y}, \mathbf{n}_0)\big)^2 &= \big( \mathbf{n}_0 \cdot \mathbf{E}_1(\mathbf{y}, \mathbf{n}_0) \mathbf{n}_0 + \mathbf{n}^{\perp}_0 \cdot \mathbf{E}_1(\mathbf{y}, \mathbf{n}_0) \mathbf{n}^{\perp}_0 \big)^2 \\
&= \big( \mathbf{n}_0 \cdot \mathbf{E}_1(\mathbf{y}, \mathbf{n}_0) \mathbf{n}_0\big)^2 + \big( \mathbf{n}^{\perp}_0 \cdot \mathbf{E}_1(\mathbf{y}, \mathbf{n}_0) \mathbf{n}^{\perp}_0\big)^2 + 2 \big( \mathbf{n}_0 \cdot \mathbf{E}_1(\mathbf{y}, \mathbf{n}_0) \mathbf{n}_0 \big)\big( \mathbf{n}^{\perp}_0 \cdot \mathbf{E}_1(\mathbf{y}, \mathbf{n}_0) \mathbf{n}^{\perp}_0 \big) \\
&= \lambda_{\text{f}}^{-5/2} \lambda_0^{5/2} \big( \mathbf{n}_0 \cdot \mathbf{II}_{\mathbf{y}} \mathbf{n}_0\big)^2 + \lambda_{\text{f}}^{1/2} \lambda_0^{-1/2} \big( \mathbf{n}^{\perp}_0 \cdot \mathbf{II}_{\mathbf{y}} \mathbf{n}^{\perp}_0\big)^2 + 2 \lambda_{\text{f}}^{-1} \lambda_0\big( \mathbf{n}_0 \cdot \mathbf{II}_{\mathbf{y}} \mathbf{n}_0\big)\big( \mathbf{n}^{\perp}_0 \cdot \mathbf{II}_{\mathbf{y}} \mathbf{n}^{\perp}_0\big)
\end{aligned}
\end{equation}
since the trace is an isotropic function.
The identities in (\ref{eq:dotN01}), (\ref{eq:dotN02}) and (\ref{eq:dotN03}) then furnish 
\begin{equation}
\begin{aligned}\label{eq:finallyDoneDone}
&\int_{\omega }\frac{\mu}{12} \Big( |\mathbf{E}_1(\mathbf{y}, \mathbf{n}_0) |^2 - g \Big( \big(\tfrac{\mu}{\gamma}\big)^{1/2} |\mathbf{E}_2(\mathbf{n}_0) | \Big) \frac{( \mathbf{E}_1(\mathbf{y}, \mathbf{n}_0) \colon \mathbf{E}_2(\mathbf{n}_0))^2}{|\mathbf{E}_2(\mathbf{n}_0)|^2} + \text{Tr} \big( \mathbf{E}_1(\mathbf{y}, \mathbf{n}_0)\big)^2 \Big) \dd A \\
&\qquad= \int_{\omega} \frac{\mu}{12} \Big\{ 2\lambda_{\text{f}}^{-5/2} \lambda_0^{5/2} \big( \mathbf{n}_0 \cdot \mathbf{II}_{\mathbf{y}} \mathbf{n}_0\big)^2 + 2 \lambda_{\text{f}}^{1/2} \lambda_0^{-1/2} \big( \mathbf{n}^{\perp}_0 \cdot \mathbf{II}_{\mathbf{y}} \mathbf{n}^{\perp}_0\big)^2 + 2 \lambda_{\text{f}}^{-1} \lambda_0\big( \mathbf{n}_0 \cdot \mathbf{II}_{\mathbf{y}} \mathbf{n}_0\big)\big( \mathbf{n}^{\perp}_0 \cdot \mathbf{II}_{\mathbf{y}} \mathbf{n}^{\perp}_0\big) \\
&\qquad \qquad \qquad \quad \Big[ 1 - g \Big(\big(\tfrac{\mu}{2\gamma}\big)^{1/2} (\lambda_{\text{f}}^{3/4} - \lambda_{\text{f}}^{-3/4})\Big) \Big] \frac{\lambda_0}{2} (\lambda_\text{f}^{-5/4} + \lambda_\text{f}^{1/4} )^2 ( \mathbf{n}_0 \cdot \mathbf{II}_{\mathbf{y}} \mathbf{n}_0^{\perp})^2 \Big\} \dd A \\
&\qquad = \int_{\omega} \mathbf{II}_{\mathbf{y}} \colon \mathbb{B}(\mathbf{n}_0) \colon \mathbf{II}_{\mathbf{y}} \dd A 
\end{aligned}
\end{equation}
for $\mathbb{B}(\mathbf{v}_0)$ defined in (\ref{eq:Bv0Moduli}). The desired inequality in (\ref{eq:theFinalLowerbound}) follows from 
(\ref{eq:getLBinfimum}), (\ref{eq:integrateTrace}), and (\ref{eq:finallyDoneDone}). 
\end{proof}

\subsubsection{Summary on the lowerbound} The lemmas above collectively prove the lowerbound in Theorem \ref{MainTheorem}. For clarity, we summarize the result below.
\begin{prop}
For every sequence $\{ (\mathbf{y}_h, \mathbf{n}_h)\} \subset H^1(\Omega, \mathbb{R}^3) \times H^1(\Omega, \mathbb{S}^2)$ such that $(\mathbf{y}_h , \mathbf{n}_h) \rightharpoonup (\mathbf{y}, \mathbf{n})$ in $H^1(\Omega, \mathbb{R}^3) \times H^1(\Omega, \mathbb{S}^2)$, $\liminf_{h \rightarrow 0} E_{\mathbf{n}_0}^h(\mathbf{y}_h, \mathbf{n}_h) \geq E_{\mathbf{n}_0} (\mathbf{y}, \mathbf{n})$. 
\end{prop}
\begin{proof}
Let $\{ (\mathbf{y}_h, \mathbf{n}_h)\} \subset H^1(\Omega, \mathbb{R}^3) \times H^1(\Omega, \mathbb{S}^2)$ such that $(\mathbf{y}_h , \mathbf{n}_h) \rightharpoonup (\mathbf{y}, \mathbf{n})$ in $H^1(\Omega, \mathbb{R}^3) \times H^1(\Omega, \mathbb{S}^2)$. If $\liminf_{h \rightarrow 0} E_{\mathbf{n}_0}^h(\mathbf{y}_h, \mathbf{n}_h) = \infty$, there is nothing to prove. So we assume $M := \liminf_{h \rightarrow 0} E_{\mathbf{n}_0}^h(\mathbf{y}_h, \mathbf{n}_h) < \infty$. Extract a subsequence (not relabeled) such that $\lim_{h \rightarrow 0} E_{\mathbf{n}_0}^h(\mathbf{y}_h, \mathbf{n}_h) = M$ and a further subsequence (also not relabeled) such that Proposition \ref{compactnessProp} holds. The latter subsequence satisfies the hypotheses of Lemma \ref{firstLiminfLemma}. Thus, 
\begin{align*}
M &\geq \int_{\Omega}\Big\{\frac{1}{2} Q_3(\mathbf{S}) + \gamma( |\nabla \mathbf{n}|^2 + |\boldsymbol{\tau}|^2) \Big\} \dd V \geq \int_{\Omega} \Big\{ \frac{1}{2}Q_2\Big( x_3 \mathbf{S}_1 + \mathcal{S}_2 \cdot \int_{0}^{x_3} \boldsymbol{\tau}^{\text{e}} \dd t \Big) + \gamma( |\nabla \mathbf{n}|^2 + |\boldsymbol{\tau}^{\text{e}}|^2 ) \Big\} \dd V \\ & \geq \int_{\Omega} \Big\{ \mu \Big| x_3 \mathbf{E}_1(\mathbf{y}, \mathbf{n}_0) + \int_0^{x_3} \tau^{\text{e}}_{\perp} \dd t \mathbf{E}_2(\mathbf{n}_0) \Big|^2 + \gamma (\tau_{\perp}^{\emph{e}})^2 + \mu \text{Tr}\big( x_3 \mathbf{E}_1(\mathbf{y}, \mathbf{n}_0)\big)^2 + \gamma \lambda_{\text{f}}^{-1} \lambda_0 \big| \nabla(\nabla \mathbf{y} \mathbf{n}_0) \big|^2 \Big\} \dd V \\
&\geq \int_{\omega} \Big\{ \mathbf{II}_{\mathbf{y}} \colon \mathbb{B}(\mathbf{n}_0) \colon \mathbf{II}_{\mathbf{y}} + \gamma \lambda_{\text{f}}^{-1} \lambda_0 \big| \nabla(\nabla \mathbf{y} \mathbf{n}_0) \big|^2 \Big\} \dd A = E_{\mathbf{n}_0}(\mathbf{y}, \mathbf{n})
\end{align*}
by Lemmas \ref{firstLiminfLemma}, \ref{secondLBIneqLemma}, \ref{thirdLBIneqLemma} and \ref{finalLBIneqLemma}.
\end{proof}

\section{Recovery Sequence}\label{sec:Recovery}

This section addresses the topic of recovery sequences, namely, sequences of 3D deformations and director fields whose bulk energy limits to the desired plate energy as $h \rightarrow 0$. We start by focusing on recovery sequences in the smooth setting in order to deal with the many non-trivial algebraic manipulations involved in the dimension reduction without fussing over delicate questions of regularity. We then turn to a proof of Theorem \ref{RecoveryTheorem}, which involves producing a recovery sequence for metric-constrained midplane deformations and director fields with fractional Sobolev regularity.

\subsection{Smooth setting.}\label{ssec:smoothSetting}
Consider any $\mathbf{n}_0 \in C^{\infty}(\overline{\omega}, \mathbb{S}^1)$, $\mathbf{n} \in C^{\infty}(\overline{\omega}, \mathbb{S}^2)$ and $\mathbf{y} \in C^{\infty}(\overline{\omega}, \mathbb{R}^3)$ that satisfy $(\mathbf{y}, \mathbf{n}) \in \mathcal{A}_{\mathbf{n}_0}$, i.e., 
\begin{equation}
\begin{aligned}\label{eq:ConstructUB1}
(\nabla \mathbf{y} ) ^T \nabla \mathbf{y} = \mathbf{g}_{\mathbf{n}_0} \quad \text{ and } \quad \mathbf{n} = \sigma \frac{\nabla \mathbf{y} \mathbf{n}_0}{|\nabla \mathbf{y} \mathbf{n}_0|}
\end{aligned}
\end{equation}
for some constant $\sigma \in \{ -1,1\}$. 
\subsubsection{Construction of recovery sequences in the smooth setting}
We build the full ansatz of our recovery sequence based on $\mathbf{y}$ and $\mathbf{n}_0$ by defining the 3D deformation $\mathbf{y}_h \colon \Omega \rightarrow \mathbb{R}^3$ and director field $\mathbf{n}_h \colon \Omega \rightarrow \mathbb{S}^2$ as
\begin{equation}
\begin{aligned}\label{eq:ConstructUBAnsatz}
&\mathbf{y}_h(\mathbf{x}, x_3) := \mathbf{y}(\mathbf{x}) + h x_3 \mathbf{b}(\mathbf{x}) + h^2 \frac{x_3^2}{2} \mathbf{d}(\mathbf{x}) \quad \text{ and } \quad \mathbf{n}_h(\mathbf{x}, x_3) := \frac{\mathbf{n}(\mathbf{x}) + h \int_0^{x_3} \boldsymbol{\tau}(\mathbf{x}, t)\dd t }{|\mathbf{n}(\mathbf{x}) + h \int_0^{x_3} \boldsymbol{\tau} (\mathbf{x}, t) \dd t|} 
\end{aligned}
\end{equation}
for smooth vector fields $\mathbf{b},\mathbf{d}, \boldsymbol{\tau}$ chosen below to be consistent with the compactness and lowerbound arguments of the previous sections. 

The first of these vectors fields, $\mathbf{b} \in C^{\infty}(\overline{\omega}, \mathbb{R}^3)$, is identified by the compactness result in Proposition \ref{compactnessProp} as being proportional to the surface normal $\boldsymbol{\nu}_{\mathbf{y}}$ in (\ref{eq:secFund}), specifically, via 
\begin{equation}
\begin{aligned}\label{eq:ConstructUB2}
 \mathbf{b} = \lambda_{\text{f}}^{-1/4} \lambda_{0}^{1/4} \boldsymbol{\nu}_{\mathbf{y}}.
\end{aligned}
\end{equation}
The parameterizations in (\ref{eq:ConstructUB1}) and (\ref{eq:ConstructUB2}) then yield a rotation field $\mathbf{R} \in C^{\infty}(\overline{\omega}, SO(3))$ defined by 
\begin{equation}
\begin{aligned}\label{eq:ConstructR}
\mathbf{R} := (\boldsymbol{\ell}_{\mathbf{n}}^{\text{f}})^{-1/2} (\nabla \mathbf{y}, \mathbf{b}) (\boldsymbol{\ell}_{\mathbf{n}_0}^{0})^{1/2},
\end{aligned}
\end{equation}
as noted by the linear algebra supplied in Lemma \ref{linAlgLemma} of the appendix. It happens that $\mathbf{n}$ and $\mathbf{R}$ are also coupled per Lemma \ref{linAlgLemma} via 
\begin{equation}
\begin{aligned}\label{eq:rotNCoupled}
\mathbf{n} = \sigma \mathbf{R} \begin{pmatrix} \mathbf{n}_0 \\ 0 \end{pmatrix},
\end{aligned}
\end{equation}
which proves useful in the derivation below. Next, note that $\boldsymbol{\tau}$ should be perpendicular to $\mathbf{n}$ for it to be asymptotically consistent with the compactness properties of the director field in Proposition \ref{compactnessProp}. In fact,
\begin{equation}
\begin{aligned}\label{eq:tauPerpN}
\mathbf{n}_h = \mathbf{n} + h \int_{0}^{x_3} \boldsymbol{\tau} \dd t + O(h^2) \quad \text{ and } \quad h^{-1} \partial_3 \mathbf{n}_h = \boldsymbol{\tau} + O(h)
\end{aligned}
\end{equation}
for any $\boldsymbol{\tau}$ perpendicular to $\mathbf{n}$ in (\ref{eq:ConstructUBAnsatz}).

\subsubsection{Prescribing $\boldsymbol{\tau}$} A detailed prescription $\boldsymbol{\tau}$ is motivated by the lowerbound argument of Lemma \ref{thirdLBIneqLemma} and \ref{infimumLemma}. We take 
\begin{equation}
\begin{aligned}\label{eq:constructTau}
\boldsymbol{\tau} = \sigma \tau_{\perp} \mathbf{R}\begin{pmatrix} \mathbf{n}_0^{\perp} \\ 0 \end{pmatrix} ,
\end{aligned}
\end{equation}
to hard-encode by (\ref{eq:rotNCoupled}) that $\boldsymbol{\tau}$ is orthogonal to $\mathbf{n}$. Then, we choose the scalar field $\tau_{\perp} \in C^{\infty}(\overline{\Omega}, \mathbb{R})$ as 
\begin{equation}
\begin{aligned}\label{eq:tauPerpFormula}
\tau_{\perp} = \frac{\lambda_0}{2} (\lambda_\text{f}^{-5/4} + \lambda_\text{f}^{1/4} )^2 \bigg[ \frac{\cosh \big(\big(\tfrac{\mu}{2\gamma}\big)^{1/2} (\lambda_{\text{f}}^{3/4} - \lambda_{\text{f}}^{-3/4})x_3 \big)}{\cosh \big(\big(\tfrac{\mu}{2\gamma}\big)^{1/2} (\lambda_{\text{f}}^{3/4} - \lambda_{\text{f}}^{-3/4}) \frac{1}{2} \big)} -1 \bigg] ( \mathbf{n}_0 \cdot \mathbf{II}_{\mathbf{y}} \mathbf{n}_0^{\perp})^2 ,
\end{aligned}
\end{equation}
so that $\boldsymbol{\tau}$ belongs to $C^{\infty}(\overline{\Omega}, \mathbb{R}^3)$ with an $x_3$-dependence that stems from making $\tau_{\perp}$ in (\ref{eq:tauPerpFormula}) consistent with Lemma \ref{infimumLemma}. In particular, $\tau_{\perp}(\mathbf{x}, x_3) = (f^{\star})'(x_3; \mathbf{E}_1(\mathbf{x}), \mathbf{E}_2(\mathbf{x}))$ for $f^{\star}$ in (\ref{eq:fminimizer}) and the strain measures $\mathbf{E}_1 \equiv \mathbf{E}_1(\mathbf{y}, \mathbf{n}_0)$ and $\mathbf{E}_2 \equiv \mathbf{E}_2(\mathbf{n}_0)$ from (\ref{eq:E1E2Strains}). 

\subsubsection{Prescribing $\mathbf{d}$} The choice of the last vector field $\mathbf{d}$ emerges most naturally in the course of the reduction from $3D$ to $2D$. First observe from (\ref{eq:ConstructUBAnsatz}) and (\ref{eq:tauPerpN}) that 
\begin{equation}
\begin{aligned}\label{eq:gradHyHFormula}
&\nabla_h \mathbf{y}_h = (\nabla \mathbf{y}, \mathbf{b}) + h x_3 (\nabla \mathbf{b}, \mathbf{d}) + O(h^2), \\
&(\boldsymbol{\ell}_{\mathbf{n}_h}^{\text{f}})^{-1/2} = (\boldsymbol{\ell}_{\mathbf{n}}^{\text{f}})^{-1/2} + 2 h (\lambda_{\text{f}}^{-1/2} - \lambda_{\text{f}}^{1/4} ) \text{sym} \big( \int_{0}^{x_3} \boldsymbol{\tau} \dd t \otimes \mathbf{n} \Big) + O(h^2).
\end{aligned}
\end{equation}
Thus from (\ref{eq:ConstructR}) the strain measure entering the entropic elastic term is 
\begin{align*}
(\boldsymbol{\ell}^{\text{f}}_{\mathbf{n}_h})^{-1/2}\nabla_h \mathbf{y}_h (\boldsymbol{\ell}_{\mathbf{n}_0}^0)^{1/2} &= \mathbf{R} + h\big( \mathbf{S}_1 + \mathbf{S}_2) + O(h^2)
\end{align*}
for tensor fields $\mathbf{S}_{1,2} \in C^{\infty}(\overline{\Omega}, \mathbb{R}^{3\times3})$ distinguished by their $x_3$-dependence via
\begin{align*}
\mathbf{S}_1:= x_3(\boldsymbol{\ell}_{\mathbf{n}} ^{\text{f}})^{-1/2} (\nabla \mathbf{b} , \mathbf{d} ) (\boldsymbol{\ell}_{\mathbf{n}_0}^0)^{1/2} ,\quad \mathbf{S}_2 := 2 (\lambda_{\text{f}}^{-1/2} - \lambda_{\text{f}}^{1/4} ) \text{sym} \big( \int_{0}^{x_3} \boldsymbol{\tau} \dd t \otimes \mathbf{n} \Big) (\nabla \mathbf{y}, \mathbf{b}) (\boldsymbol{\ell}_{\mathbf{n}_0}^0)^{1/2}.
\end{align*}
Manipulations on these tensor fields reveal the natural choice of $\mathbf{d}$. Specifically, $\mathbf{S}_1$ satisfies
\begin{equation}
\begin{aligned}\label{eq:S1Manip}
\mathbf{S}_1 &= x_3 \mathbf{R} (\boldsymbol{\ell}_{\mathbf{n}_0} ^{\text{f}})^{-1/2} \mathbf{R}^T \Big(\nabla \mathbf{b}[(\boldsymbol{\ell}_{\mathbf{n}_0}^0)^{1/2}]_{2\times2} , \lambda_0^{-1/4} \mathbf{d} \Big) \\
&= x_3 \mathbf{R} \Big( ( \boldsymbol{\ell}_{\mathbf{n}_0}^{\text{f}} )^{-1}( \boldsymbol{\ell}_{\mathbf{n}_0}^{\text{0}} )^{1/2} (\nabla \mathbf{y}, \mathbf{b})^T \nabla \mathbf{b}[(\boldsymbol{\ell}_{\mathbf{n}_0}^0)^{1/2}]_{2\times2}, \lambda_0^{-1/4} (\boldsymbol{\ell}_{\mathbf{n}_0} ^{\text{f}})^{-1/2} \mathbf{R}^T \mathbf{d} \Big) \\
&= x_3 \mathbf{R} \begin{pmatrix} \lambda_{\text{f}}^{-1/4} \lambda_0^{1/4} \big[(\boldsymbol{\ell}_{\mathbf{n}_0}^{\text{f}} )^{-1}( \boldsymbol{\ell}_{\mathbf{n}_0}^{\text{0}} )^{1/2}\big]_{2\times2} \mathbf{II}_{\mathbf{y}} [(\boldsymbol{\ell}_{\mathbf{n}_0}^0)^{1/2}]_{2\times2} & [\tilde{\mathbf{d}}]_{2\times1} \\ [ \mathbf{0}]_{2\times1} & \tilde{\mathbf{d}} \cdot \mathbf{e}_3 \end{pmatrix}
\end{aligned}
\end{equation}
for $\tilde{\mathbf{d}} := \lambda_0^{-1/4} (\boldsymbol{\ell}_{\mathbf{n}_0} ^{\text{f}})^{-1/2} \mathbf{R}^T\mathbf{d}$ by (\ref{eq:rotNCoupled}), Remark \ref{lotsOfIdentsRem}, and the definition of $\mathbf{b}$ in (\ref{eq:ConstructUB2}). $\mathbf{S}_2$ satisfies 
\begin{equation}
\begin{aligned}\label{eq:S2Manip}
\mathbf{S}_2 &= 2 (\lambda_{\text{f}}^{-1/2} - \lambda_{\text{f}}^{1/4} ) \Big( \int_{0}^{x_3} \tau_{\perp} \dd t \Big) \mathbf{R}\; \text{sym} \Big( \begin{pmatrix} \mathbf{n}_0^{\perp} \\ 0 \end{pmatrix} \otimes \begin{pmatrix} \mathbf{n}_0 \\ 0 \end{pmatrix} \Big) \mathbf{R}^T (\nabla \mathbf{y}, \mathbf{b}) (\boldsymbol{\ell}_{\mathbf{n}_0}^0)^{1/2} \\
&= 2 (\lambda_{\text{f}}^{-1/2} - \lambda_{\text{f}}^{1/4} ) \Big( \int_{0}^{x_3} \tau_{\perp} \dd t \Big) \mathbf{R}\; \text{sym} \Big( \begin{pmatrix} \mathbf{n}_0^{\perp} \\ 0 \end{pmatrix} \otimes \begin{pmatrix} \mathbf{n}_0 \\ 0 \end{pmatrix} \Big) ( \boldsymbol{\ell}_{\mathbf{n}_0}^{\text{f}} )^{1/2} \\
&=(\lambda_{\text{f}}^{-1/2} - \lambda_{\text{f}}^{1/4} )\Big( \int_{0}^{x_3} \tau_{\perp} \dd t \Big) \mathbf{R} \Big[ \lambda_{\text{f}}^{1/2} \begin{pmatrix} \mathbf{n}_0^{\perp} \\ 0 \end{pmatrix} \otimes \begin{pmatrix} \mathbf{n}_0 \\ 0 \end{pmatrix} + \lambda_{\text{f}}^{-1/4} \begin{pmatrix} \mathbf{n}_0 \\ 0 \end{pmatrix} \otimes \begin{pmatrix} \mathbf{n}^{\perp}_0 \\ 0 \end{pmatrix} \Big]
\end{aligned}
\end{equation}
given (\ref{eq:rotNCoupled}) and (\ref{eq:tauPerpFormula}) and because $\mathbf{R}^T (\nabla \mathbf{y}, \mathbf{b}) = ( \boldsymbol{\ell}_{\mathbf{n}_0}^{\text{f}} )^{1/2} (\boldsymbol{\ell}_{\mathbf{n}_0}^0)^{-1/2}$ by Remark \ref{lotsOfIdentsRem} and (\ref{eq:FtransposeF}) of Appendix \ref{sec:linAlgebra}. As a final bit of manipulation, the last of the identities in (\ref{eq:S1Manip}) and (\ref{eq:S2Manip}) furnishes 
\begin{align*}
\text{sym} (\mathbf{R}^T \mathbf{S}_1) = x_3 \begin{pmatrix} \mathbf{E}_1(\mathbf{y}, \mathbf{n}_0) & \frac{1}{2} [\tilde{\mathbf{d}}]_{2\times1} \\
 \frac{1}{2} [\tilde{\mathbf{d}}^T]_{1 \times 2} & \tilde{\mathbf{d}} \cdot \mathbf{e}_3 \end{pmatrix} , \quad \text{sym} (\mathbf{R}^T \mathbf{S}_2) = \Big( \int_0^{x_3} \tau_{\perp} \dd t \Big) \begin{pmatrix} \mathbf{E}_2(\mathbf{n}_0) & [\mathbf{0}]_{2\times1} \\ [ \mathbf{0} ]_{1\times2} & 0 \end{pmatrix}
 \end{align*}
for the $2\times2$ symmetric tensors fields $\mathbf{E}_1(\mathbf{y}, \mathbf{n}_0)$ and $ \mathbf{E}_2(\mathbf{n}_0)$ defined in (\ref{eq:E1E2Strains}).
We build on the identities above to derive the leading order expressions for the entropic and determinant energy densities, which in turn motivates our choice of $\tilde{\mathbf{d}}$ and thus $\mathbf{d} = \lambda_0^{1/4} \mathbf{R} (\boldsymbol{\ell}_{\mathbf{n}_0} ^{\text{f}})^{1/2} \tilde{\mathbf{d}}$. The entropic energy density satisfies 
\begin{equation}
\begin{aligned}\label{eq:TaylorWSmooth}
h^{-2} W\big((\boldsymbol{\ell}^{\text{f}}_{\mathbf{n}_h})^{-1/2}\nabla_h \mathbf{y}_h (\boldsymbol{\ell}_{\mathbf{n}_0}^0)^{1/2}\big) &= h^{-2} W\big( \mathbf{R} + h (\mathbf{S}_1 + \mathbf{S}_2) + O(h^2) \big) \\
&= h^{-2} W\big( \mathbf{I} + h \mathbf{R}^T (\mathbf{S}_1 + \mathbf{S}_2) + O(h^2) \big) \\
&= \frac{1}{2} Q_3\Big( \text{sym} \big (\mathbf{R}^T (\mathbf{S}_1 + \mathbf{S}_2) \big) + O(h) \Big) + o(1) \\
&= \frac{1}{2} Q_3 \Big( \begin{pmatrix} x_3 \mathbf{E}_1(\mathbf{y}, \mathbf{n}_0) + \big( \int_0^{x_3} \tau_{\perp} \dd t \big) \mathbf{E}_2(\mathbf{n}_0) & \frac{1}{2} x_3[\tilde{\mathbf{d}}]_{2\times1} \\
\frac{1}{2} x_3 [\tilde{\mathbf{d}}^T]_{1 \times 2} & x_3\tilde{\mathbf{d}} \cdot \mathbf{e}_3 \end{pmatrix} \Big) + o(1).
\end{aligned}
\end{equation}
due to the frame-indifference of $W$, (\ref{eq:Q3A}), and the smoothness of all fields. Along the same lines, the determinant energy density satisfies 
\begin{equation}
\begin{aligned}\label{eq:detManips}
h^{-2} \kappa_h ( \det \nabla_{h} \mathbf{y}_h - 1)^2 &= h^{-2} \kappa_h ( \det \big( \mathbf{R}^T (\boldsymbol{\ell}^{\text{f}}_{\mathbf{n}_h})^{-1/2}\nabla_h \mathbf{y}_h (\boldsymbol{\ell}_{\mathbf{n}_0}^0)^{1/2} \big) - 1)^2 \\
&= h^{-2} \kappa_h ( \det \big( \mathbf{I} + h \mathbf{R}^T (\mathbf{S}_1 + \mathbf{S}_2) + O(h^2) \big) - 1)^2 \\
&= \kappa_h \Big[\text{Tr}\Big( \text{sym}\big( \mathbf{R}^T (\mathbf{S}_1 + \mathbf{S}_2) \big) \Big) + O(h)\Big]^2 \\
&= \kappa_h \Big[ \ \text{Tr}\Big( \begin{pmatrix} x_3 \mathbf{E}_1(\mathbf{y}, \mathbf{n}_0) +\big( \int_0^{x_3} \tau_{\perp} \dd t \big) \mathbf{E}_2(\mathbf{n}_0) & \frac{1}{2} x_3 [\tilde{\mathbf{d}}]_{2\times1} \\
 \frac{1}{2} x_3 [\tilde{\mathbf{d}}^T]_{1 \times 2} & x_3 \tilde{\mathbf{d}} \cdot \mathbf{e}_3 \end{pmatrix} \Big) + O(h) \Big]^2 \\
 &= \kappa_h \Big[ x_3 \text{Tr}\Big( \begin{pmatrix} \mathbf{E}_1(\mathbf{y}, \mathbf{n}_0) & \frac{1}{2} [\tilde{\mathbf{d}}]_{2\times1} \\
 \frac{1}{2} [\tilde{\mathbf{d}}^T]_{1 \times 2} & \tilde{\mathbf{d}} \cdot \mathbf{e}_3 \end{pmatrix} \Big) + O(h) \Big]^2 
\end{aligned}
\end{equation}
using the the identity in (\ref{eq:detIdent}) and the fact that $\det \mathbf{R} = \det\big( (\boldsymbol{\ell}^{\text{f}}_{\mathbf{n}_h})^{-1/2}\big) = \det\big( (\boldsymbol{\ell}_{\mathbf{n}_0}^0)^{1/2} \big) = 1$ and $\text{Tr}(\mathbf{E}_2(\mathbf{n}_0)) = 0$. We therefore choose $\tilde{\mathbf{d}} \in C^{\infty}(\overline{\omega}, \mathbb{R}^3)$ (equivalently $\mathbf{d} \in C^{\infty}(\overline{\omega}, \mathbb{R}^3)$) as
\begin{equation}
\begin{aligned}\label{eq:solveDVec}
\tilde{\mathbf{d}} := -\text{Tr} \big( \mathbf{E}_1(\mathbf{y},\mathbf{n}_0) \big) \mathbf{e}_3 \quad \text{ equivalently } \quad \mathbf{d} := -\text{Tr} \big( \mathbf{E}_1(\mathbf{y},\mathbf{n}_0) \big) \mathbf{R} \mathbf{e}_3
\end{aligned}
\end{equation}
to obtain the result
\begin{equation}
\begin{aligned}\label{eq:comboEnergyDensities}
&\frac{1}{h^2} \Big( W\big((\boldsymbol{\ell}^{\text{f}}_{\mathbf{n}_h})^{-1/2}\nabla_h \mathbf{y}_h (\boldsymbol{\ell}_{\mathbf{n}_0}^0)^{1/2}\big) + \kappa_h ( \det \nabla_{h} \mathbf{y}_h - 1)^2 \Big)\\
&\qquad = \mu \Big\{ \Big| x_3 \mathbf{E}_1(\mathbf{y}, \mathbf{n}_0) + \int_0^{x_3} \tau^{\emph{e}}_{\perp} \dd t \mathbf{E}_2(\mathbf{n}_0) \Big|^2 + \text{Tr}\big( x_3 \mathbf{E}_1(\mathbf{y}, \mathbf{n}_0)\big)^2 \Big\} + o(1) + O(\kappa_h h^2) .
\end{aligned}
\end{equation}

\subsubsection{Asymptotic analysis of the energy} Having prescribed all the fields of our ansatz in (\ref{eq:ConstructUBAnsatz}), it remains only to calculate the limiting energy. Focusing first on the director anchoring term, we have that
\begin{equation}
\begin{aligned}
\int_{\Omega} \frac{\mu_h}{h^{2}} \big| \mathbf{P}_{\mathbf{n}_0} (\nabla_h \mathbf{y}_h)^T \mathbf{n}_h \big|^2 \dd V &= \int_{\Omega} \frac{\mu_h}{h^2} \Big| \mathbf{P}_{\mathbf{n}_0} \Big( \big(\nabla \mathbf{y}, \lambda_0^{1/4} \lambda^{-1/4}_\text{f} \boldsymbol{\nu}_{\mathbf{y}} \big) + O(h)\Big) \Big( \sigma \frac{\nabla \mathbf{y} \mathbf{n}_0}{|\nabla \mathbf{y} \mathbf{n}_0|} + O(h) \Big) \Big|^2 \dd V \\
&= \int_{\Omega} \frac{\mu_h}{h^2} \Big| \mathbf{P}_{\mathbf{n}_0} \Big(\frac{\mathbf{g}_{\mathbf{n}_0} \mathbf{n}_0}{|\nabla \mathbf{y} \mathbf{n}_0|} + O(h)\Big) \Big|^2 \dd V = O(\mu_h).
\end{aligned}
\end{equation} 
by (\ref{eq:ConstructUB1}), (\ref{eq:ConstructUB2}), (\ref{eq:tauPerpN}), (\ref{eq:gradHyHFormula}) and since $\mathbf{g}_{\mathbf{n}_0} \mathbf{n}_0$ is parallel to $\mathbf{n}_0$. The Frank elastic terms satisfies 
\begin{equation}
\begin{aligned}\label{eq:FrankLeadingOrder}
\int_{\Omega} \frac{\gamma_h}{h^2} |\nabla_h \mathbf{n}_h|^2 \dd V &= \int_{\Omega} \frac{\gamma_h}{h^2} \{ |\nabla \mathbf{n} + O(h)|^2 + |\boldsymbol{\tau} + O(h)|^2 \} \dd V \\
&= \int_{\Omega} \frac{\gamma_h}{h^2} \big\{ \lambda_{\text{f}}^{-1} \lambda_0 | \nabla (\nabla \mathbf{y} \mathbf{n}_0)|^2 + (\tau_{\perp})^2 \big\} \dd V + O(\frac{\gamma_h}{h^2} h) \\
&= \int_{\Omega} \gamma \big\{ \lambda_{\text{f}}^{-1} \lambda_0 | \nabla (\nabla \mathbf{y} \mathbf{n}_0)|^2 + (\tau_{\perp})^2 \big\} \dd V + O(\frac{\gamma_h}{h^2} h) + O(|\frac{\gamma_h}{h^2} - \gamma|) 
\end{aligned}
\end{equation}
by (\ref{eq:ConstructUB1}), (\ref{eq:tauPerpN}), (\ref{eq:constructTau}) and because $|\nabla \mathbf{y} \mathbf{n}_0|^2 = \mathbf{n}_0 \cdot \mathbf{g}_{\mathbf{n}_0} \mathbf{n}_0 = \lambda_{\text{f}} \lambda_0^{-1}$. The expressions in (\ref{eq:comboEnergyDensities}-\ref{eq:FrankLeadingOrder}) together yield 
\begin{equation}
\begin{aligned}\label{eq:getToLimit1}
E_{\mathbf{n}_0}^h(\mathbf{y}_h , \mathbf{n}_h) &= \int_{\Omega} \Big\{ \mu \Big| x_3 \mathbf{E}_1(\mathbf{y}, \mathbf{n}_0) + \int_0^{x_3} \tau_{\perp} \dd t \mathbf{E}_2(\mathbf{n}_0) \Big|^2 + \gamma (\tau_{\perp})^2 + \mu \text{Tr}\big( x_3 \mathbf{E}_1(\mathbf{y}, \mathbf{n}_0)\big)^2 \Big\} \dd V \\
&\qquad + \int_{\omega} \gamma \lambda_{\text{f}}^{-1} \lambda_0 | \nabla (\nabla \mathbf{y} \mathbf{n}_0)|^2 \dd A + o(1) + O(\kappa_h h^2) + O(\mu_h) + O(\frac{\gamma_h}{h^2} h) + O(|\frac{\gamma_h}{h^2} - \gamma|).
\end{aligned}
\end{equation}
Thus, after invoking the prescription of $\tau_{\perp}$ in (\ref{eq:tauPerpFormula}), the first integral becomes
\begin{align*}
&\int_{\Omega} \Big\{ \mu \Big| x_3 \mathbf{E}_1(\mathbf{y}, \mathbf{n}_0) + \int_0^{x_3} \tau_{\perp} \dd t \mathbf{E}_2(\mathbf{n}_0) \Big|^2 + \gamma (\tau_{\perp})^2 + \mu \text{Tr}\big( x_3 \mathbf{E}_1(\mathbf{y}, \mathbf{n}_0)\big)^2 \Big\} \dd V \\
&\qquad = \int_{\omega }\frac{\mu}{12} \Big( |\mathbf{E}_1(\mathbf{y}, \mathbf{n}_0) |^2 - g \Big( \big(\tfrac{\mu}{\gamma}\big)^{1/2} |\mathbf{E}_2(\mathbf{n}_0) | \Big) \frac{( \mathbf{E}_1(\mathbf{y}, \mathbf{n}_0) \colon \mathbf{E}_2(\mathbf{n}_0))^2}{|\mathbf{E}_2(\mathbf{n}_0)|^2} + \text{Tr} \big( \mathbf{E}_1(\mathbf{y}, \mathbf{n}_0)\big)^2 \Big) \dd A.
\end{align*}
It then follows from the same line of reasoning as that of (\ref{eq:finallyDoneDone}) in the proof of Lemma \ref{finalLBIneqLemma} that 
\begin{equation}
\begin{aligned}\label{eq:getToLimit3}
E_{\mathbf{n}_0}^h(\mathbf{y}_h , \mathbf{n}_h)&= \int_{\omega}\big\{ \mathbf{II}_{\mathbf{y}} \colon \mathbb{B}(\mathbf{n}_0) \colon \mathbf{II}_{\mathbf{y}} + \gamma \lambda_{\text{f}}^{-1} \lambda_0 | \nabla (\nabla \mathbf{y} \mathbf{n}_0)|^2 \big\} \dd A \\
&\qquad \quad + o(1) + O(\kappa_h h^2) + O(\mu_h) + O(\frac{\gamma_h}{h^2} h) + O(|\frac{\gamma_h}{h^2} - \gamma|).
\end{aligned}
\end{equation}
Consequently,
\begin{align*}
\lim_{h \rightarrow 0} E_{\mathbf{n}_0}^h(\mathbf{y}_h , \mathbf{n}_h) = \int_{\omega}\big\{ \mathbf{II}_{\mathbf{y}} \colon \mathbb{B}(\mathbf{n}_0) \colon \mathbf{II}_{\mathbf{y}} + \gamma \lambda_{\text{f}}^{-1} \lambda_0 | \nabla (\nabla \mathbf{y} \mathbf{n}_0)|^2 \big\} \dd A
\end{align*}
by the scaling assumptions of the moduli in (\ref{eq:kappaH}), (\ref{eq:muH}) and (\ref{eq:gammaH}). This result completes the proof of a recovery sequence in the case that $(\mathbf{y}, \mathbf{n}) \in \mathcal{A}_{\mathbf{n}_0}$ such that $\mathbf{y}, \mathbf{n}$ and $\mathbf{n}_0$ all belong to $C^{\infty}(\overline{\omega})$.

\subsection{Sobolev-regular setting.}\label{ssec:SobolevSetting} 

This section  proves Theorem \ref{RecoveryTheorem}. In particular, we construct a recovery sequence under the assumption that the limiting deformation has fractional Sobolev regularity. For all estimates in this setting, $C> 0$ refers to a sufficiently large constant independent of $h$ that can change from line to line. We also make ample use of big-$O$ notation applied to non-smooth fields. To clarify, for any $v_h, w_h: \Omega \to V$, we say that $v_h = O(|w_h|)$ if there is a constant $C> 0$  such that $|v_h(\mathbf{x}) | \leq C |w_h(\mathbf{x})|$ for a.e.\;$\mathbf{x} \in \Omega$ for all $h> 0$ sufficiently small. 
\subsubsection{Construction of recovery sequence in the non-smooth setting.}

Let $(\mathbf{y}, \mathbf{n}) \in \mathcal{A}_{\mathbf{n}_0}$ so that (\ref{eq:ConstructUB1}) holds a.e.\;on $\omega$, and assume $\mathbf{II}_{\mathbf{y}} \in L^\infty (\omega, \mathbb{R}^{2\times2}_{\text{sym}}) \cap H^{s}(\omega, \mathbb{R}^{2 \times2}_{\text{sym}})$. We build off the construction in the smooth setting by defining $\mathbf{b} \in H^1(\omega, \mathbb{R}^3)$ as in (\ref{eq:ConstructUB2}) and $\mathbf{R} \in H^1(\omega, SO(3))$ as in (\ref{eq:ConstructR}). The identity in (\ref{eq:rotNCoupled}) then holds a.e.\;on $\omega$ for some fixed $\sigma \in \{ -1,1\}$. Since $\omega$ is Lipschitz, we can extend $\mathbf{n}_0$ to a function in $H^1(\mathbb{R}^2,\mathbb{R}^2)$, $\mathbf{n}$ and $\mathbf{b}$ to functions in $H^1(\mathbb{R}^2, \mathbb{R}^3)$, $\mathbf{y}$ to a function in $H^{2}(\mathbb{R}^2, \mathbb{R}^3)$, $\mathbf{R}$ to a function $H^{1}(\mathbb{R}^2, \mathbb{R}^{3\times3})$, and define $\mathbf{d} \colon \mathbb{R}^2 \rightarrow \mathbb{R}^3$ by the parameterization in (\ref{eq:solveDVec}) and $\boldsymbol{\tau} \colon \mathbb{R}^2 \times (-1/2,1/2) \rightarrow \mathbb{R}^3$ by the parameterizations in (\ref{eq:constructTau}) and (\ref{eq:tauPerpFormula}).
Note that, when restricted to $\omega$, $\mathbf{d}$ belongs to $L^{\infty}(\omega,\mathbb{R}^3) \cap H^s(\omega, \mathbb{R}^3)$ since it is the product of functions that belong to $H^1(\omega) \cap L^{\infty}(\omega)$ and $\mathbf{II}_{\mathbf{y}} \in L^\infty(\omega) \cap H^s(\omega)$. By the same reasoning, $\boldsymbol{\tau}$ belongs to $L^{\infty}(\omega, \mathbb{R}^3) \cap H^s(\omega, \mathbb{R}^3)$ for a.e.\;$x_3 \in (-1/2,1/2)$ and to $C^{\infty}([-1/2,1/2],\mathbb{R}^3)$ for a.e.\;$\mathbf{x} \in \omega$. 

A key challenge for the recovery sequence here as compared to the smooth setting in Section \ref{ssec:smoothSetting} is that we must avoid directly taking the gradients of $\mathbf{d}$ and $\boldsymbol{\tau}$ --- they are now not necessarily weakly differentiable. Instead, we develop smooth approximations of these fields and justify that the approximations contribute negligible energy. Let $m \in C^{\infty}_c(\mathbb{R}^2, \mathbb{R})$ denote a standard mollifier and set $m_{h}(\mathbf{x}) := h^{-2} m(h^{-1} \mathbf{x})$. 
We mollify $\mathbf{d}$ in the standard way via $\mathbf{d}_{h} := \mathbf{d} \ast m_{h}$.
For $\boldsymbol{\tau}$, we instead smooth it in a way that preserves its orthogonality with $\mathbf{n}$ by setting 
\begin{align*}
&\tau_{\perp,h} := \frac{\lambda_0}{2} (\lambda_\text{f}^{-5/4} + \lambda_\text{f}^{1/4} )^2 \bigg[ \frac{\cosh \big(\big(\tfrac{\mu}{2\gamma}\big)^{1/2} (\lambda_{\text{f}}^{3/4} - \lambda_{\text{f}}^{-3/4})x_3 \big)}{\cosh \big(\big(\tfrac{\mu}{2\gamma}\big)^{1/2} (\lambda_{\text{f}}^{3/4} - \lambda_{\text{f}}^{-3/4}) \frac{1}{2} \big)} -1 \bigg] \big( (\mathbf{n}_0 \cdot \mathbf{II}_{\mathbf{y}} \mathbf{n}_0^{\perp})^2 \ast m_{h}\big), \\
&\boldsymbol{\tau}_{h} := \sigma \tau_{\perp,h} \mathbf{R}\begin{pmatrix} \mathbf{n}_0^{\perp} \\ 0 \end{pmatrix} . 
\end{align*}
Observe that $\boldsymbol{\tau}_{h}$ is the same parameterization as that of $\boldsymbol{\tau}$, except the field $(\mathbf{n}_0 \cdot \mathbf{II}_{\mathbf{y}} \mathbf{n}_0^{\perp})^2$ is replaced by its mollification. 
 Note that $\mathbf{d}_{h} \in C^{\infty}(\overline{\omega},\mathbb{R}^3)$ and $\tau_{\perp,h} \in C^{\infty}(\overline{\Omega}, \mathbb{R}^3)$. Furthermore, since the fields $\mathbf{d}$ and $\tau_{\perp}$ belong to $L^{\infty}(\omega)$, their smoothings satisfy the pointwise estimates 
 \begin{equation}
\begin{aligned}\label{eq:pointwiseEstMoll}
 &\|\mathbf{d}_{h}\|_{L^{\infty}(\omega)} + h \|\nabla \mathbf{d}_{h}\|_{L^{\infty}(\omega)} \leq C \|\mathbf{d}\|_{L^{\infty}(\omega)}, \quad 
 &\|\tau_{\perp, h} \|_{L^{\infty}(\Omega)} + h \|\nabla \tau_{\perp,h}\|_{L^{\infty}(\Omega)} \leq C \| \tau_{\perp} \|_{L^{\infty}(\Omega)} .
\end{aligned}
\end{equation}
Since $\mathbf{d}$ and $\tau_{\perp}$ also belong to $H^s(\Omega)$, their smoothings obey the square-integrable estimates (see, for instance, Exercise 6.64 \cite{leoni2023first})
\begin{equation}
\begin{aligned}\label{eq:HSEsts}
 &\|\mathbf{d} _{h} - \mathbf{d}\|^2_{L^2(\omega)} \leq C h^{2s} \| \mathbf{d}\|_{H^{s}(\omega)}^2, && \|\nabla \mathbf{d}_{h} \|_{L^2(\omega)}^2 \leq C h^{2(s-1)} \| \mathbf{d}\|_{H^{s}(\omega)}^2, \\
 &\|\tau_{\perp,h} - \tau_{\perp} \|^2_{L^2(\Omega)} \leq C h^{2s} \| \tau_{\perp} \|_{H^{s}(\Omega)}^2, && \|\nabla \tau_{\perp,h} \|_{L^2(\Omega)}^2 \leq C h^{2(s-1)} \| \tau_{\perp}\|_{H^{s}(\Omega)}^2.
\end{aligned}
\end{equation}

The construction now proceeds as in the smooth setting. For $h \ll 1$, consider an ansatz of the form 
\begin{align*}
\mathbf{y}_{h}(\mathbf{x}, x_3) = \mathbf{y}(\mathbf{x}) + h x_3 \mathbf{b}(\mathbf{x}) + h^2 \frac{x_3^2}{2} \mathbf{d}_{h} (\mathbf{x}) \quad \text{ and } \quad \mathbf{n}_{h} (\mathbf{x}, x_3) := \frac{\mathbf{n}(\mathbf{x}) + h \int_0^{x_3} \boldsymbol{\tau}_{h}(\mathbf{x}, t)\dd t }{|\mathbf{n}(\mathbf{x}) + h \int_0^{x_3} \boldsymbol{\tau}_{h} (\mathbf{x}, t) \dd t|} .
\end{align*}
The denominator of $\mathbf{n}_h$ satisfies $|\mathbf{n} + h \int_0^{x_3}\boldsymbol{\tau}_{h}dt |^{-1} = (1 + h^2 |\int_0^{x_3}\boldsymbol{\tau}_{h}dt|^2)^{-1/2} 
 = 1 + O(h^2)$  and $\nabla \big( |\mathbf{n} + h \int_0^{x_3}\boldsymbol{\tau}_{h}dt |^{-1} \big) = O(h^2 )\int_0^{x_3} \nabla \boldsymbol{\tau}_hdt$,   since $\boldsymbol{\tau}_{h}(\mathbf{x},x_3)$ is perpendicular to the unit vector field $\mathbf{n}(\mathbf{x})$ and since $|\boldsymbol{\tau}_{h}| \leq \|\boldsymbol{\tau}\|_{L^{\infty}(\Omega)} = O(1)$. 
 So it follows that 
 \begin{equation}
\begin{aligned}\label{eq:dirDeltaEpsCompute}
&\mathbf{n}_{h} = \mathbf{n} + h \int_0^{x_3}\boldsymbol{\tau}_{h} dt + O(h^2) = \mathbf{n} + h \int_0^{x_3} \boldsymbol{\tau} dt + O(h^2) + O( h \int_{-1/2}^{1/2} |\tau_{\perp, h} - \tau_{\perp}| dx_3) , \\
&\nabla \mathbf{n}_{h} = ( 1+ O(h^2)) \nabla \mathbf{n} + \big( h+ O(h^2)\big) \int_0^{x_3} \nabla \boldsymbol{\tau}_{h} dt , \\
&h^{-1} \partial_3 \mathbf{n}_{h} = \boldsymbol{\tau}_{h} + O(h) = \boldsymbol{\tau} + O(h) + O( |\tau_{\perp,h} - \tau_{\perp}|) .
\end{aligned}
\end{equation}
We also observe that 
\begin{align*}
&\nabla_h \mathbf{y}_{h} = (\nabla \mathbf{y}, \mathbf{b}) + h x_3 (\nabla \mathbf{b}, \mathbf{d}) + O(h|\mathbf{d}_{h} - \mathbf{d}|) + O(h^2 |\nabla \mathbf{d}_h|) \\
&(\boldsymbol{\ell}_{\mathbf{n}_h}^{\text{f}})^{-1/2} = (\boldsymbol{\ell}_{\mathbf{n}}^{\text{f}})^{-1/2} + 2 h (\lambda_{\text{f}}^{-1/2} - \lambda_{\text{f}}^{1/4} ) \text{sym} \big( \int_{0}^{x_3} \boldsymbol{\tau} \dd t \otimes \mathbf{n} \Big) + O(h^2) + O(h \int_{-1/2}^{1/2} |\tau_{\perp, h} - \tau_{\perp}| dx_3).
\end{align*}

Looking to the strain that enters into the entropic energy, these calculations when combined with the definition of $\mathbf{R}$ yield
\begin{equation}
\begin{aligned}\label{eq:getStrainEntropic}
(\boldsymbol{\ell}^{\text{f}}_{\mathbf{n}_{h}})^{-1/2}\nabla_h \mathbf{y}_{h} (\boldsymbol{\ell}_{\mathbf{n}_0}^0)^{1/2} &= \mathbf{R} + h\big( \mathbf{S}_{1} + \mathbf{S}_{2}) + \boldsymbol{\delta} \mathbf{S}_h . 
\end{aligned}
\end{equation}
The terms $\mathbf{S}_{1,2} \in L^{\infty}(\Omega, \mathbb{R}^{2\times 2})$ are exactly as defined in the recovery sequence in the smooth setting in Section \ref{ssec:smoothSetting}. They satisfy the identities in (\ref{eq:S1Manip}) and (\ref{eq:S2Manip}) a.e.\;on $\Omega$ for $\tilde{\mathbf{d}} := \lambda_0^{-1/4} (\boldsymbol{\ell}_{\mathbf{n}_0}^{\text{f}})^{-1/2} \mathbf{R}^T \mathbf{d}$ defined in (\ref{eq:solveDVec}). This in turn implies that 
\begin{align}\label{eq:symStrain}
 \text{sym}\big[ \mathbf{R}^T(\mathbf{S}_1 + \mathbf{S}_2)\big] = \begin{pmatrix} x_3 \mathbf{E}_1(\mathbf{y}, \mathbf{n}_0) + (\int_0^{x_3} \tau_{\perp} \dd t) \mathbf{E}_2(\mathbf{n}_0) & [\mathbf{0}]_{2\times 1} \\
 [\mathbf{0}]_{1 \times 2} & - x_3 \text{Tr}( \mathbf{E}_1(\mathbf{y}, \mathbf{n}_0)) 
 \end{pmatrix}
\end{align}
a.e.\;on $\Omega$ for $\mathbf{E}_1(\mathbf{y}, \mathbf{n}_0)$ and $ \mathbf{E}_2(\mathbf{n}_0)$ defined in (\ref{eq:E1E2Strains}). $\boldsymbol{\delta} \mathbf{S}_h \in L^{\infty}(\Omega, \mathbb{R}^{2\times 2})$, meanwhile, collects the remainder terms for this non-smooth setting. They are of the form 
\begin{align*}
\boldsymbol{\delta} \mathbf{S}_h = O(h^2) + O(h |\mathbf{d}_h - \mathbf{d}|) + O(h^2|\nabla \mathbf{d}_{h}|) + O(h \int_{-1/2}^{1/2} |\tau_{\perp, h} - \tau_{\perp}| dx_3)
\end{align*}
and satisfy the following estimates 
\begin{align}\label{eq:remainderEsts}
 \|\boldsymbol{\delta} \mathbf{S}_h \|_{L^{\infty}(\Omega)} \leq C h, \quad \| \boldsymbol{\delta} \mathbf{S}_h \|_{L^2(\Omega)} \leq C h^{1+s}.
\end{align}
In particular, the remainder terms are pointwise small, as indicated, because of (\ref{eq:pointwiseEstMoll}) and even smaller in a square-integrable sense because of (\ref{eq:HSEsts}).
Both estimates are crucial in calculating the leading order energy of the sheet, as we now show. 
\subsubsection{Asymptotic analysis of the energy and the proof of Theorem \ref{RecoveryTheorem}}

We start with the entropic energy. A Taylor expansion of $W$ using (\ref{eq:TaylorW}), the fact that all the field in (\ref{eq:getStrainEntropic}) are bounded in $L^{\infty}(\Omega)$, and the manipulations in (\ref{eq:TaylorWSmooth}) gives
\begin{equation}
\begin{aligned}\label{eq:sobolevEnergyEst1}
&\int_{\Omega} h^{-2} W\big((\boldsymbol{\ell}^{\text{f}}_{\mathbf{n}_{h}})^{-1/2}\nabla_h \mathbf{y}_{h} (\boldsymbol{\ell}_{\mathbf{n}_0}^0)^{1/2})\big) \dd V \\
&\qquad = \int_{\Omega} h^{-2}W\Big( \mathbf{I} + h \mathbf{R}^T( \mathbf{S}_{1} + \mathbf{S}_{2} + h^{-1} \boldsymbol{\delta} \mathbf{S}_h) \Big) \dd V \\
 & \qquad = \int_{\Omega} \Big\{ \frac{1}{2} Q_3 \Big( \text{sym}\big[ \mathbf{R}^T(\mathbf{S}_1 + \mathbf{S}_2 + h^{-1} \boldsymbol{\delta} \mathbf{S}_h) \big] \Big) + o( |h^{-1} \boldsymbol{\delta} \mathbf{S}_h|^2) + o( 1 ) \Big \} \dd V \\
 &\qquad = \int_{\Omega} \Big\{ \frac{1}{2} Q_3 \Big( \text{sym}\big[ \mathbf{R}^T(\mathbf{S}_1 + \mathbf{S}_2)\big]\Big) + \underbrace{O(|h^{-1}\boldsymbol{\delta} \mathbf{S}_h|) + O(|h^{-1} \boldsymbol{\delta}\mathbf{S}_h|^2) + o(|h^{-1} \boldsymbol{\delta} \mathbf{S}_h|^2)}_{O(h^{-1} |\boldsymbol{\delta} \mathbf{S}_h|)} + o(1) \Big \} 
 \dd V \\ 
 &\qquad = \int_{\Omega} \mu \Big\{ \Big| x_3 \mathbf{E}_1(\mathbf{y}, \mathbf{n}_0) + \int_0^{x_3} \tau_{\perp} \dd t \mathbf{E}_2(\mathbf{n}_0) \Big|^2 + \text{Tr}\big( x_3 \mathbf{E}_1(\mathbf{y}, \mathbf{n}_0)\big)^2 \Big \} \dd V + \underbrace{O\Big( h^{-1} \|\boldsymbol{\delta} \mathbf{S}_h\|_{L^1(\Omega)} \Big)}_{=O(h^s)} + o(1) 
\end{aligned}
\end{equation}
Here, the third equality collects the remainder terms into a term of $O(h^{-1}|\boldsymbol{\delta} \mathbf{S}_h|)$ due to the pointwise estimate in (\ref{eq:remainderEsts}). The fourth uses (\ref{eq:symStrain}) and the definition of $Q_3$ in (\ref{eq:Q3A}), and then uses H\"{o}lder's inequality and the square-integrable estimate (\ref{eq:remainderEsts}) to conclude that $O(h^{-1} \| \boldsymbol{\delta} \mathbf{S}_h\|_{L^1(\Omega)}) = O(h^s)$.

Next, we calculate the incompressibility penalty. Using the first identity in (\ref{eq:detIdent}) and the manipulations in (\ref{eq:detManips}), this term has the form
\begin{equation}
\begin{aligned}\label{eq:sobolveEnergyEst2}
& \int_{\Omega} h^{-2} \kappa_h ( \det \nabla_{h} \mathbf{y}_{h} - 1)^2 \dd V \\
&\qquad = \int_{\Omega} h^{-2} \kappa_h \big( \det \big[ \mathbf{I} + h \mathbf{R}^T (\mathbf{S}_{1} + \mathbf{S}_{2} + h^{-1} \boldsymbol{\delta} \mathbf{S}_h ) \big] - 1) \dd V \\
&\qquad = 
\int_{\Omega} h^{-2}\kappa_h \Big[ h \underbrace{\text{Tr}\big( \text{sym} \big[ \mathbf{R}^T(\mathbf{S}_1 + \mathbf{S}_2) \big] \big)}_{=0} + \underbrace{O(|\boldsymbol{\delta} \mathbf{S}_h| ) + O(h^2) + O(|\boldsymbol{\delta} \mathbf{S}_h|^2) + O(h^3) + O(|\boldsymbol{\delta} \mathbf{S}_h|^3)}_{ = O(|\boldsymbol{\delta} \mathbf{S}_h|) + O(h^2)} \Big]^2 \\
 &\qquad = O \Big( \kappa_h h^{-2} \|\boldsymbol{\delta}\mathbf{S}_h\|_{L^2(\Omega)}^2 \Big) + O (\kappa_h h^2) = O(\kappa_h h^{2s}) 
\end{aligned}
\end{equation}
The second equality here uses the estimates in (\ref{eq:detEsts}) to supply the remainder terms and the pointwise estimates in (\ref{eq:remainderEsts}) to simplify these terms. The third equality then follows because $\text{sym}[ \mathbf{R}^T (\mathbf{S}_1 + \mathbf{S}_2)]$ in (\ref{eq:symStrain}) has zero trace a.e.\;on $\Omega$. The fourth follows from the square-integrable estimates in (\ref{eq:remainderEsts}).

Turning to the director anchoring term, we note that $(\nabla \mathbf{y}, \mathbf{b}) \in L^\infty (\omega)$ and that $\mathbf{P}_{\mathbf{n}_0}(\nabla \mathbf{y}, \mathbf{b})^T \mathbf{n} = \mathbf{0}$ a.e.\;on $\omega$ since the metric of $\mathbf{y}$ is $\mathbf{g}_{\mathbf{n}_0}\in L^{\infty}(\omega)$ and since $\mathbf{b}$ satisfies (\ref{eq:ConstructUB2}). It follows that
\begin{equation}
\begin{aligned}\label{eq:sobolevEnergyEst4}
 &\int_{\Omega} \frac{\mu_h}{h^2} \big| \mathbf{P}_{\mathbf{n}_0} ( \nabla_h \mathbf{y}_{h})^T \mathbf{n}_{h}|^2 \dd V \\
 &\qquad = \int_{\Omega} \frac{\mu_h}{h^2} \big| \mathbf{P}_{\mathbf{n}_0} \big( h x_3 (\nabla \mathbf{b}, \mathbf{d}_{h})^T + \frac{h^2 x_3^2}{2} ( \nabla \mathbf{d}_{h} , \mathbf{0} )^T\big) \mathbf{n}_{h} + \mathbf{P}_{\mathbf{n}_0} (\nabla \mathbf{y}, \mathbf{b})^T (\mathbf{n}_{h} - \mathbf{n}) \Big|^2 \dd V \\
 &\qquad = O \Big( \mu_h \int_{\Omega} |(\nabla \mathbf{b}, \mathbf{d}_{h})|^2 + h^2 |\nabla \mathbf{d}_{h}|^2 + h^{-2} | \mathbf{n}_{h} - \mathbf{n}|^2 \dd V \Big) = O(\mu_h) + O(\mu_h h^{2s} ) = O(\mu_h), 
\end{aligned}
\end{equation}
where the last line uses (\ref{eq:HSEsts}) for the $\nabla \mathbf{d}_h$ term and that $|\mathbf{d}_h| \leq C \| \mathbf{d}\|_{L^{\infty}(\omega)}$ and $|\mathbf{n}_{h} - \mathbf{n}| = O(h)$, which follow from (\ref{eq:pointwiseEstMoll}) and (\ref{eq:dirDeltaEpsCompute}). 

Finally, for the Frank elastic term, we invoke both (\ref{eq:HSEsts}) and (\ref{eq:dirDeltaEpsCompute}) to obtain the estimate 
\begin{align*}
 \|\nabla_h \mathbf{n}_h - (\nabla \mathbf{n}, \boldsymbol{\tau}) \|_{L^2(\Omega)} &\leq C( h^2 \| \nabla \mathbf{n}\|_{L^{2}(\omega)} + h \| \nabla \boldsymbol{\tau}_h \|_{L^2(\Omega)} + \|\tau_{\perp,h} - \tau_{\perp} \|_{L^2(\Omega)} + h) \leq Ch^{s}.
\end{align*}
 It is therefore straightforward to deduce from standard estimates and H\"{o}lder's inequality that 
\begin{equation}
\begin{aligned}\label{eq:sobolevEnergyEst5}
&\int_{\Omega} \frac{\gamma_h}{h^2} |\nabla_h \mathbf{n}_{h}|^2 \dd V \\
&\qquad = \int_{\Omega} \frac{\gamma_h}{h^2} |(\nabla \mathbf{n}, \boldsymbol{\tau})|^2 \dd V + O\Big( \frac{\gamma_h}{h^2} \|(\nabla \mathbf{n}, \boldsymbol{\tau})\|_{L^2(\Omega)} \| \nabla_h \mathbf{n}_h - (\nabla \mathbf{n}, \boldsymbol{\tau}) \|_{L^2(\Omega)}\Big) + O\Big( \frac{\gamma_{h}}{h^2} \| \nabla_h \mathbf{n}_h - (\nabla \mathbf{n}, \boldsymbol{\tau})\|_{L^2(\Omega)}^2 \Big) \\
&\qquad =\int_{\Omega} \gamma \big\{ \lambda_{\text{f}}^{-1} \lambda_0 | \nabla (\nabla \mathbf{y} \mathbf{n}_0)|^2 + (\tau_{\perp})^2 \big\} \dd V + O\Big( |\frac{\gamma_h}{h^2} - \gamma | \Big) + O(\frac{\gamma_h}{h^2}h^s). 
\end{aligned}
\end{equation}

We are now ready to complete the proof. Observe that the expressions in (\ref{eq:sobolevEnergyEst1}-\ref{eq:sobolevEnergyEst5}) together yield 
\begin{align*}
 E_{\mathbf{n}_0}^h(\mathbf{y}_h, \mathbf{n}_h) &= \int_{\Omega} \Big\{ \mu \Big| x_3 \mathbf{E}_1(\mathbf{y}, \mathbf{n}_0) + \int_0^{x_3} \tau_{\perp} \dd t \mathbf{E}_2(\mathbf{n}_0) \Big|^2 + \gamma (\tau_{\perp})^2 + \mu \text{Tr}\big( x_3 \mathbf{E}_1(\mathbf{y}, \mathbf{n}_0)\big)^2 \Big\} \dd V \\
&\qquad + \int_{\omega} \gamma \lambda_{\text{f}}^{-1} \lambda_0 | \nabla (\nabla \mathbf{y} \mathbf{n}_0)|^2 \dd A + o(1) + O(\kappa_h h^{2s}) + O(\mu_h) + O(\frac{\gamma_h}{h^2} h^s) + O(|\frac{\gamma_h}{h^2} - \gamma|). 
\end{align*}
The first integral then becomes $\int_{\omega} \mathbf{II}_{\mathbf{y}} \colon \mathbb{B}(\mathbf{n}_0) \colon \mathbf{II}_{\mathbf{y}} \dd V$ after employing the definition of $\tau_{\perp}$ (see (\ref{eq:getToLimit1}-\ref{eq:getToLimit3})). Meanwhile, the remainder terms vanish as $h\rightarrow 0$ by the scaling assumptions on $\kappa_h, \mu_h$ and $\gamma_h$ in (\ref{eq:kappaH}), (\ref{eq:muH}) and (\ref{eq:gammaH}). Consequently, $\lim_{h \rightarrow 0} E_{\mathbf{n}_0}^h (\mathbf{y}_h, \mathbf{n}_h) = E_{\mathbf{n}_0}(\mathbf{y},\mathbf{n})$ as desired.  It is also clear  that $(\mathbf{y}_h, \mathbf{n}_h) \rightarrow (\mathbf{y}, \mathbf{n})$ in $H^1(\Omega, \mathbb{R}^3) \times H^1(\Omega, \mathbb{S}^2)$.

\section{Examples of the theory}\label{sec:Examples}
We end the paper by illustrating the rich design space covered by our plate theory, specifically, through a series of canonical examples from the literature. In particular, we add to the literature by now computing and discussing the bending energy of these examples. We assume throughout that the ratio $\lambda_{\text{a}} := \lambda_\text{f}/\lambda_0$ is $< 1$, so that the actuation corresponds to heating the LCE sheet.
\subsection{Surfaces of revolution from smooth directors} The first set of examples come from Aharoni \textit{et al.}\;\cite{aharoni2014geometry}, which produced families of smooth director profiles and deformations that solved the metric constraint, leading to surfaces of revolution on actuation. 

The basic idea from that work is as follows. Start with a general parameterization of a surface of revolution in isothermal coordinates 
\begin{align}\label{eq:surfOfRevParam}
\mathbf{w}(\xi ,\eta) = \begin{pmatrix} R(\xi) \cos(\eta) \\
R(\xi) \sin(\eta) \\ z(\xi) \end{pmatrix} \quad \text{ subject to } \quad 
R'(\xi)^2 + z'(\xi)^2 = R(\xi)^2 
\end{align}
and general 2D director profile in those coordinates 
\begin{align*}
 \mathbf{m}_0( \xi,\eta) =\begin{pmatrix} \cos \theta(\xi, \eta) \\ \sin \theta(\xi, \eta) \end{pmatrix}.
\end{align*}
The goal is to seek out the right change of coordinates $\mathbf{x} = (x_1, x_2) = \boldsymbol{\Phi}(\eta, \xi)$ to the physical domain such that
\begin{align}\label{eq:designProblem}
\mathbf{y} \circ \boldsymbol{\Phi} := \mathbf{w}, \quad \mathbf{n}_0 \circ \boldsymbol{\Phi} := \mathbf{m}_0 \quad \text{ solves } \quad \nabla \mathbf{y}^T \nabla \mathbf{y} = \mathbf{g}_{\mathbf{n}_0}. 
\end{align}
Leaving aside the details, Aharoni\;\textit{et al.}\;showed that the coordinate transformation
\begin{align*}
 &\boldsymbol{\Phi}(\xi, \eta):= \begin{pmatrix} \lambda_{\text{a}}^{-1/4} \int_0^{\xi} R^2(\tilde{\xi}) \dd\tilde{\xi} \\[6pt]
 \eta + \lambda_{\text{a}}^{-1/4} \int_0^{\xi} \sqrt{(\lambda_{\text{a}} - R^2(\tilde{\xi}))(R^2(\tilde{\xi}) - \lambda_{\text{a}}^{-1/2})} \dd\tilde{\xi} \end{pmatrix} 
\end{align*}
along with the family of director profiles 
\begin{align*}
&\theta( \xi, \eta) = \arctan |\phi(\xi)| , \quad \phi(\xi) = \sqrt{ \frac{R^2(\xi) - \lambda^{-1/2 }_\text{a}}{\lambda_{\text{a}} - R^2(\xi)}} 
\end{align*}
provides a large class of solutions to (\ref{eq:designProblem}). 

\begin{figure}[t!]\hspace*{-.4cm}
 \includegraphics[width=1\textwidth]{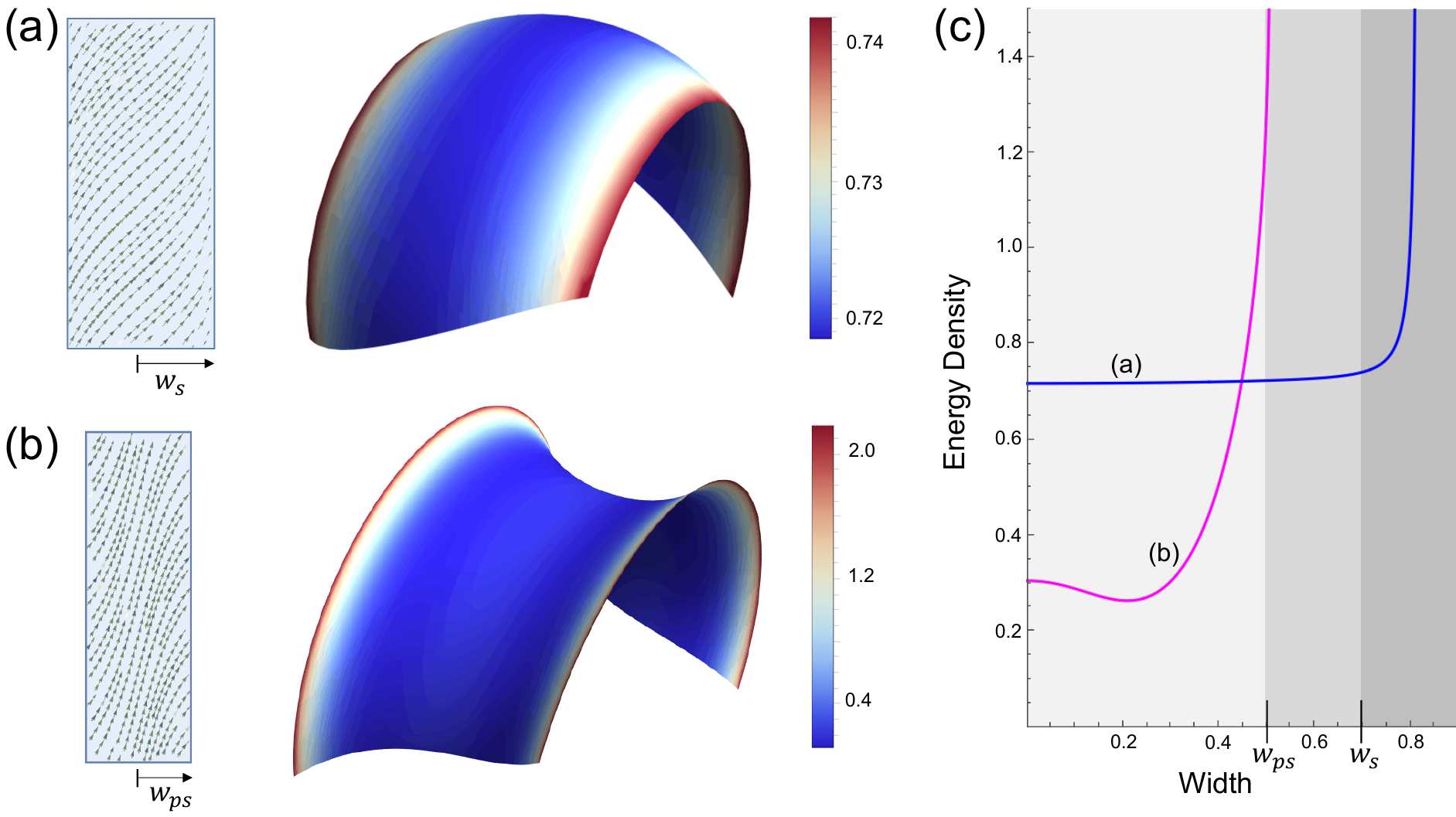} 
 \caption{ Metric-constrained surfaces of revolution. (a) Reference and deformed configuration of the spherical actuation of Gauss curvature $K = 1$. (b) Reference and deformed configuration of the pseudospherical actuation of Gauss curvature $K = -1.8$. 
 Both examples have parameters $\lambda_{\text{f}} = 1.5, \lambda_0 = 3, \mu = 1,$ and $ \gamma = 1/100$, with their plate energy densities overlaid onto the deformed configuration. (c) Plot of the plate energy densities as a function of $x_1$; the densities are independent of $x_2$. The different shadings denote the extent of each sheet domain: $w_s = 0.7$ for (a) and $w_{ps} = 0.5$ for (b).} \label{fig:aharoni-examples}
\end{figure}

Here we highlight two such examples from that work, obtained through appropriate choices of $R(\xi)$ and $z(\xi)$ subject to
 the ODE constraint in (\ref{eq:surfOfRevParam}). The first is a spherical actuation of constant positive Gauss curvature $K >0$, specified by the choices 
\begin{align*}
 R(\xi) = \frac{1}{\sqrt{K}} \sech (\xi), \quad z(\xi) = \frac{1}{\sqrt{K}} \tanh(\xi) .
\end{align*}
The second is the pseudospherical actuation of constant negative Gauss curvature $K < 0$, specified by 
\begin{align*}
 R(\xi) = \frac{1}{\sqrt{-K}} \sec (\xi), \quad z(\xi) = \frac{1}{\sqrt{-K}} \int_0^{\xi} \sqrt{ \sec^2(\tilde{\xi}) (1- \tan^2(\tilde{\xi})) } d \tilde{\xi} .
\end{align*}
Substituting these choices for $R(\xi)$ and $z(\xi)$ into (\ref{eq:surfOfRevParam}) and (\ref{eq:designProblem}) yields explicit characterizations of the metric-constrained pair $(\mathbf{y}, \mathbf{n}_0)$ for these two examples. All quantities of interest in our plate theory --- the director profiles and actuations as a function of $\mathbf{x}$ and energy densities --- can be determined from this characterization through standard manipulations. For brevity, we suppress the details of these calculations and simply illustrate the results using Fig.\;\ref{fig:aharoni-examples}.

Fig.\;\ref{fig:aharoni-examples}(a) shows the director profile and actuation in the spherical case; Fig.\;\ref{fig:aharoni-examples}(b) shows the same in the pseudospherical case. Notice that the director varies only along the horizontal $x_1$ coordinate in both cases. In particular, the director in the spherical case starts more to the horizontal at $x_1=0$ and tilts vertically as $x_1$ increases. This profile causes the sheet to expand on actuation in the $x_2$ direction at $x_1 = 0$ and contract in said direction at $x_1 = \pm w_s$, resulting in the 3D spherical shape shown. In contrast, the director in the pseudspherical case starts more vertically at $x_1 =0$ and tilts horizontally as $x_1$ increases, causing the sheet to contract in the $x_2$ direction at $x_1 = 0$ and expand at $x_2 = \pm w_{ps}$, thus leading to the negative Gauss curvature shape shown. In both cases, the energy density is independent of $x_2$ and is plotted in Fig.\;\ref{fig:aharoni-examples}(c). Observe that the density in the spherical case remains relatively constant but would quickly start to blow up if we were to extend the width a bit beyond $w_s = 0.7$. The density in the pseudospherical case is more varied, but exhibits the same type of blow up should the width be extended beyond $w_{ps} = 0.5$.

\subsection{Rotationally symmetric director fields }
We now turn our attention to examples of metric-constrained actuation involving rotationally symmetric director fields, including the conical actuation of Modes \textit{et al.}\;\cite{modes2011gaussian} and some of its log-spiral generalizations by Mostejeran and Warner \cite{mostajeran2016encoding, warner2018nematic}. 

These types of examples emerge from an ansatz of the form 
\begin{align*}
 \mathbf{n}_0(\mathbf{x}) = \begin{pmatrix} \cos\left(\theta + \frac{1}{2} \arccos \beta(r)\right) \\[3pt] \sin\left(\theta + \frac{1}{2} \arccos \beta(r)\right) \end{pmatrix}, \quad \mathbf{y}(\mathbf{x}) = \begin{pmatrix} d(r) \cos (\theta + \gamma(r)) \\ d(r) \sin(\theta + \gamma(r)) \\ p(r) \end{pmatrix} \quad \text{for} \quad \mathbf{x} = r \mathbf{e}_r(\theta),
\end{align*}
where $\mathbf{e}_r(\theta) \in \mathbb{R}^2$ is the radial basis vector in Polar coordinates. In particular, the necessary and sufficient conditions for the metric constraint $\nabla \mathbf{y}^T \nabla \mathbf{y} = \mathbf{g}_{\mathbf{n}_0}$ to hold under this ansatz are that $\beta(r), \gamma(r), d(r)$ and $p(r)$ solve
\begin{align}\label{eq:RadialExamples}
 \begin{cases}
 d'(r)^2 + p'(r)^2+ (d(r) \gamma'(r) )^2 = c_1 - c_2 \beta(r) \\
   d(r)^2 \gamma'(r) = c_2 r\sqrt{1- \beta(r)^2} \\
     d(r)^2 = r^2 (c_1 + c_2\beta(r) ) 
    \end{cases}
\end{align}
where $c_1 = \tfrac{1}{2} (\lambda_{\text{a}} + \lambda_{\text{a}}^{-1/2})$ and $c_2 = \tfrac{1}{2} (\lambda_{\text{a}}^{-1/2} - \lambda_{\text{a}})$. Note that $\nabla \mathbf{n}_0$ is singular as $r \rightarrow 0$. We therefore consider examples on a punctured disc $\omega := \{ r \mathbf{e}_r(\theta) \in \mathbb{R}^2 \colon r \in (r_i, r_0), \theta \in [0,2\pi)\}$ for some $0 < r_i < r_o$ to maintain that $\nabla \mathbf{n}_0 \in L^2(\omega)$, as demanded by our plate theory. We can then solve the middle equation in (\ref{eq:RadialExamples}) by choosing 
\begin{align*}
\gamma(r) = \int_{r_i}^r \frac{c_2 \sqrt{1 - \beta(v)^2}}{r(c_1+c_2\beta(v))} dv.
\end{align*}
The remaining two equations in (\ref{eq:RadialExamples}) can be easily solved to produce large families of examples,

\begin{figure}[t!]
\centering
\hspace*{-.4cm}
\includegraphics[width=1\textwidth]{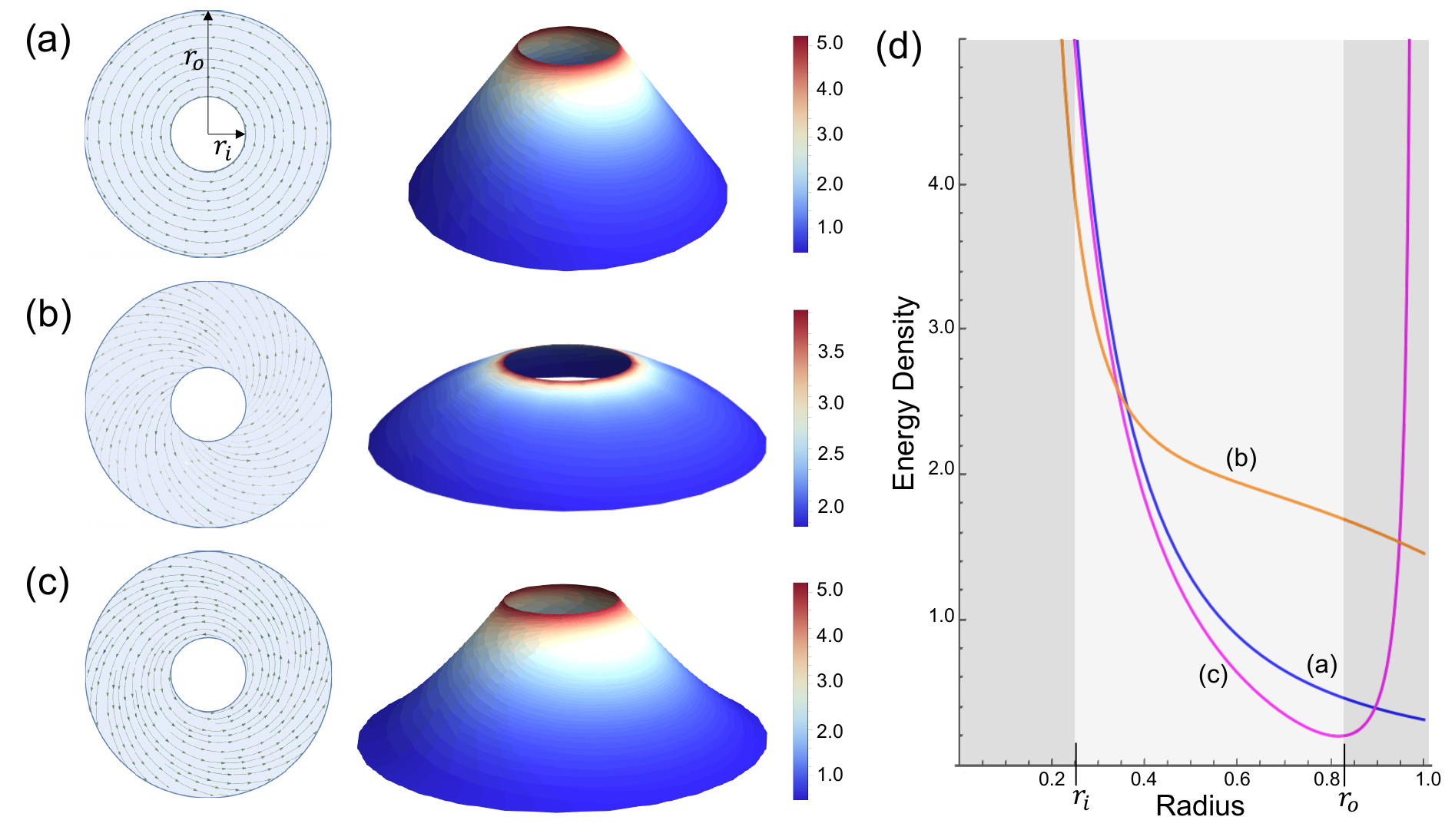} 
\caption{Metric-constrained actuations with rotationally symmetric director fields. Reference and deformed configuration of (a) the classical cone with zero Gauss curvature; (b) the spherical cap with Gauss curvature $K = 1$; (c) the hyperbolic cone with Gauss curvature $K = - 1$.
 All examples have parameters $r_i = 0.25$, $r_o = 0.825$, $\lambda_{\text{f}} = 1.5, \lambda_0 = 3, \mu = 1,$ and $ \gamma = 1/100$. Their plate energy densities are overlaid onto the deformed configuration. (d) Plot of the plate energy densities as a function of the radius $r$; the densities are independent of $\theta$. The grey shadings denote the extent of the annuli domains.}
\label{fig:LogSpiralExample}
\end{figure}

We highlight three specific examples of metric-constrained actuation on the punctured disk, namely, the cone, spherical cap, and hyperbolic cone. The classical cone example comes from \cite{modes2011gaussian}. It has zero Gauss curvature and is specified by 
\begin{align*}
 \beta(r) = 0, \quad d(r)= \lambda_{\text{a}}^{1/2} r, \quad p(r) = \sqrt{\lambda_{\text{a}}^{1/2} - \lambda_{\text{a}}}\, r.
\end{align*}
The spherical cap and hyperbolic cone belong to the family of log-spiral designs from \cite{mostajeran2016encoding, warner2018nematic}. The spherical cap, in particular, has constant positive Gauss curvature $K >0$ and is specified by 
\begin{align*}
 \beta(r) = 1-\frac{2}{1+\lambda_{a}^{3 / 4}}-c_{K} r^{2}, \quad d(r) = \frac{1}{\sqrt{K}}\sin(\sqrt{K} s(r)), \quad p(r) = \frac{1}{\sqrt{K}}\sin(\sqrt{K} s(r)),
\end{align*}
where $c_K := \frac{K}{2(\lambda_{a}^{-1}-\lambda_{a}^{1 / 2})}$ and $s(r) :=\int_0^r \frac{\lambda_{\text{a}}^{1/4}}{\sqrt{c_1+c_2\beta(v)}} dv$. The hyperbolic cone has constant negative Gauss curvature $K <0$ and is specified by
\begin{align*}
 \beta(r) = -1 -c_{K} r^{2}, \quad d(r) = \rho \sinh\left(\sqrt{|K|} s(r)\right), \quad p(r) = \int_0^{s(r)\sqrt{|K|}}& \sqrt{\frac{1}{|K|} - \rho^2\cosh^2(v)}dv,
\end{align*}
where $c_K$ is the same as the spherical cap case, $\rho := \frac{\lambda_{\text{a}}}{2\sqrt{|c_2c_K|}}$ and $s(r) := \tfrac{2}{\sqrt{|K|}}\text{arcsinh}(\sqrt{\tfrac{|c_2c_K|}{\lambda_{\text{a}} }}r)$. With these prescriptions, all quantities of interest in our plate theory can be determined via standard manipulations. As before, we do not run through these tedious but straightforward details and instead simply illustrate the results using Fig.\;\ref{fig:LogSpiralExample}. 

Fig.\;\ref{fig:LogSpiralExample}(a) shows the director profile and actuation of the classical cone, Fig.\;\ref{fig:LogSpiralExample}(b) shows the same for the spherical cap, and Fig.\;\ref{fig:LogSpiralExample}(c) shows that of the hyperbolic cone. Notice that the director in the classical cone is purely azimuthal and does not vary along $r$. Thus, $d(r)/r$ in this case describes a constant circumferential contraction  $= \lambda_\text{a}^{1/2}$, which produces the conical shape shown. In the log-spiral cases, the director varies as a function of $r$,  so the circumferential contraction $d(r)/r <1$ also varies in $r$. This apparently  dictates  the Gauss curvature in each case. Specifically, 
 $d(r)/r$ is monotonically decreasing for the spherical cap and monotonically increasing for the hyperbolic cone. Intuitively, the latter implies an actuation that starts conical but must flatten out as $r$ increases. These features leads naturally to the hyperbolic cone on display. In contrast, the former must result in an actuation that starts more flat in $r$ and then gradually increases its curvature as $r$ increases, hence spherical cap actuation. 
 
 Notice from Fig.\;\ref{fig:LogSpiralExample}(d) that all of these examples exhibit a blow up in their energy densities as $r \rightarrow 0$. This is fully expected since the director field is singular at the origin. Interestingly, the energy density also blows up for the hyperbolic case as $r$ increases beyond $0.825.$ There are probably analogies to be made with growth or swelling processes that conform to a hyperbolic metric \cite{klein2007shaping}. In this setting, one often 
 experimentally observes a waviness at the outer boundaries of the sample, a phenomenon that has been attributed mathematically to non-smooth $H^2$ isometric immersions \cite{gemmer2016isometric,yamamoto2021nature}. Let us recall that we did not quite settle the question of $\Gamma$-convergence for negatively-curved surfaces. The presence of unbounded energy at the extent of the domain and possible non-smooth competitors is perhaps at that heart of the issue. We leave this topic an open question for future research.

\section*{Acknowledgments}
 \noindent P.P. acknowledges support from the National Science Foundation (CMMI-CAREER-2237243) and Army Research Office (ARO-W911NF2310137).  L.B.'s acknowledges support from the NSF through the Center for Nonlinear Analysis, and D.P.G. was supported in part by the
Zuckerman STEM Leadership Program.

\appendix

\section{Linear algebra and inequalities}\label{sec:linAlgebra}

Here we collect various pointwise estimates and linear algebra identities used throughout the work. 

\begin{lem}\label{AppendLemma2}
There are constants $c_1, c_2> 0$ such that
\begin{equation}\label{eq:firstAppendProof}
\emph{dist}^2(\mathbf{F}, SO(3)) \geq c_1|\mathbf{F}|^2 - c_2 \quad \text{ for all } \mathbf{F} \in \mathbb{R}^{3\times3}.
\end{equation}
\end{lem}
\begin{proof} First observe that 
\begin{align*}
\text{dist}^2(\mathbf{F}, SO(3)) = \inf_{\mathbf{Q} \in SO(3)} |\mathbf{F}|^2 - 2 \text{Tr}(\mathbf{F}^T\mathbf{Q}) + 3 \geq |\mathbf{F}|^2 - 2 \sqrt{3} |\mathbf{F}| + 3 
\end{align*}
since $\text{Tr}(\mathbf{F}^T\mathbf{Q}) \leq |\mathbf{F}| |\mathbf{Q}| \leq \sqrt{3} |\mathbf{F}|$ for all $\mathbf{Q} \in SO(3)$. At least one of the following is true: (i) $ 2 \sqrt{3}|\mathbf{F}| \leq \frac{1}{2} |\mathbf{F}|^2$, (ii) $2 \sqrt{3}|\mathbf{F}| \geq \frac{1}{2} |\mathbf{F}|^2$. If (ii), then $|\mathbf{F}|$ is $\leq4\sqrt{3}$ and thus $|\mathbf{F}|^2 - 2 \sqrt{3} |\mathbf{F}| + 3 \geq |\mathbf{F}|^2 - 21 \geq | \mathbf{F}|^2 -21$. Otherwise, (i) holds and thus $|\mathbf{F}|^2 - 2 \sqrt{3} |\mathbf{F}| + 3 \geq \frac{1}{2} |\mathbf{F}|^2 + 3 \geq \frac{1}{2}| \mathbf{F}|^2 -21$. Evidently then (\ref{eq:firstAppendProof}) holds with $c_1= \frac{1}{2}$ and $c_2 = 21$. 
\end{proof}
\begin{lem}\label{AppendLemma22}
There are constants $c_1( \lambda_\emph{f},\lambda_0), c_2> 0$ such that
\begin{align*}
\emph{dist}^2((\boldsymbol{\ell}_{\mathbf{v}}^{\emph{f}})^{-1/2} \mathbf{F} (\boldsymbol{\ell}_{\mathbf{v}_0}^{0})^{1/2}, SO(3)) \geq c_1(\lambda_0, \lambda_{\emph{f}}) |\mathbf{F}|^2 - c_2, \quad \text{ for all } \mathbf{F} \in \mathbb{R}^{3\times3}, \mathbf{v} \in \mathbb{S}^2, \mathbf{v}_0 \in \mathbb{S}^1. 
\end{align*}
\end{lem}
\begin{proof}
In view of Lemma \ref{AppendLemma2}, we only need to show that $|(\boldsymbol{\ell}_{\mathbf{v}}^{\emph{f}})^{-1/2} \mathbf{F} (\boldsymbol{\ell}_{\mathbf{v}_0}^{0})^{1/2}|^2 \geq c(\lambda_\text{f}, \lambda_{0}) |\mathbf{F}|^2$ for some $c(\lambda_\text{f}, \lambda_{0})>0$. This result follows because 
\begin{align*}
|(\boldsymbol{\ell}_{\mathbf{v}}^{\text{f}})^{-1/2} \mathbf{F} (\boldsymbol{\ell}_{\mathbf{v}_0}^{0})^{1/2}| \geq \sigma_{\text{min}}\big(( \boldsymbol{\ell}_{\mathbf{v}}^{\text{f}})^{-1/2}\big) \sigma_{\text{min}}\big(( \boldsymbol{\ell}_{\mathbf{v}_0}^{\text{0}})^{1/2}\big) | \mathbf{F}| \geq \lambda_{\text{f}}^{-1/2} \lambda_0^{-1/4} |\mathbf{F}|,
\end{align*}
for $\sigma_{\text{min}}(\mathbf{A})$ the minimum singular value of the tensor $\mathbf{A}$. 
\end{proof}

For the next lemma, observe that since $SO(3)$ is a smooth and compact manifold of $\mathbb{R}^{3\times3}$, there is an $\epsilon >0$ and a tubular neighborhood $N_{\epsilon}(SO(3)) := \{ \mathbf{F} \in \mathbb{R}^{3\times3} \colon \text{dist}( \mathbf{F}, SO(3)) < \epsilon \}$ such that the nearest point projection from $N_{\epsilon}(SO(3))$ to $SO(3)$ is defined uniquely for each point in $N_{\epsilon}(SO(3))$ and given by a smooth function $\boldsymbol{\pi} \colon N_{\epsilon}(SO(3)) \rightarrow SO(3)$. Let
\begin{align*}
\boldsymbol{\pi}_{\text{ext}}(\mathbf{F}) := \begin{cases}
\boldsymbol{\pi}(\mathbf{F}) & \text{ if } \mathbf{F} \in N_{\epsilon}(SO(3)) \\
\mathbf{I} & \text{ otherwise on $\mathbb{R}^{3\times3}$}.
\end{cases}
\end{align*}
\begin{lem}\label{projectionLemma}
There is a $C(\epsilon) > 0$ such that
\begin{equation}
\begin{aligned}\label{eq:distFunEst}
|\boldsymbol{\pi}_{\emph{ext}}(\mathbf{F}) - \mathbf{F}|^2 \leq C(\epsilon) \emph{dist}^2( \mathbf{F}, SO(3)) \quad \text{ for all } \mathbf{F} \in \mathbb{R}^{3\times3}.
\end{aligned}
\end{equation}
\end{lem}
\begin{proof}
If $\mathbf{F}$ is in $N_{\epsilon}(SO(3))$, then $|\boldsymbol{\pi}_{\text{ext}}(\mathbf{F}) - \mathbf{F}| = \text{dist}( \mathbf{F}, SO(3))$. So (\ref{eq:distFunEst}) holds trivially on $N_{\epsilon}(SO(3))$. On the exceptional set, $\text{dist}(\mathbf{F}, SO(3)) \geq \epsilon$ and 
\begin{align*}
|\boldsymbol{\pi}_{\text{ext}}(\mathbf{F}) - \mathbf{F}|^2 \leq 2\big(|\mathbf{I} - \mathbf{Q}|^2 + |\mathbf{Q} - \mathbf{F}|^2\big) \leq 2 \big( 2 (|\mathbf{I}|^2 + |\mathbf{Q}|^2) + |\mathbf{Q} - \mathbf{F}|^2\big) \leq 24 + 2 |\mathbf{Q} - \mathbf{F}|^2
\end{align*}
for any choice of $\mathbf{Q} \in SO(3)$. It follows after taking the infimum over such rotations that, when $\mathbf{F}$ is in $\mathbb{R}^{3\times3} \setminus N_{\epsilon}(SO(3))$, 
\begin{align*}
|\boldsymbol{\pi}_{\text{ext}}(\mathbf{F}) - \mathbf{F}|^2 \leq 24 + 2\text{dist}^2( \mathbf{F}, SO(3)) \leq (24 \epsilon^{-2} + 2) \text{dist}^2( \mathbf{F}, SO(3)),
\end{align*}
which completes the proof. 
\end{proof}

In our next result, it is useful to introduce a "mixed" step length tensor defined by 
\begin{equation}
\begin{aligned}\label{eq:mixedStepLength}
\boldsymbol{\ell}_{\mathbf{v}_0}^{\text{f}} := \lambda_{\text{f}} \begin{pmatrix} \mathbf{v}_0 \\ \mathbf{0} \end{pmatrix} \otimes \begin{pmatrix} \mathbf{v}_0 \\ \mathbf{0} \end{pmatrix} + \lambda_\text{f}^{-1/2} \left( \mathbf{I} - \begin{pmatrix} \mathbf{v}_0 \\ \mathbf{0} \end{pmatrix} \otimes \begin{pmatrix} \mathbf{v}_0 \\ \mathbf{0} \end{pmatrix} \right), \quad \mathbf{v}_0 \in \mathbb{S}^1,
\end{aligned}
\end{equation}
for the purpose of manipulating various identities on the deformation gradient and director field that arise at the bending energy scale. 
\begin{lem}\label{linAlgLemma}
Let $\mathbf{v} \in \mathbb{S}^2$, $\mathbf{v}_0 \in \mathbb{S}^1$, $\mathbf{F} \in \mathbb{R}^{3\times3}$, $\mathbf{Q} \in SO(3)$ and $\sigma \in \{ \pm1\}$. The identities 
\begin{equation}
\begin{aligned}\label{eq:firstLinIdents}
\mathbf{v} = \sigma \mathbf{Q} \begin{pmatrix} \mathbf{v}_0 \\ 0 \end{pmatrix} \quad \text{ and } \quad (\boldsymbol{\ell}_{\mathbf{v}}^{\emph{f}})^{-1/2} \mathbf{F} (\boldsymbol{\ell}_{\mathbf{v}_0}^0)^{1/2} = \mathbf{Q}
\end{aligned}
\end{equation}
are equivalent to 
\begin{equation}
\begin{aligned}\label{eq:secLinIdents}
\mathbf{v} = \sigma \frac{[\mathbf{F}]_{3\times2} \mathbf{v}_0}{|[\mathbf{F}]_{3\times2} \mathbf{v}_0|}, \quad [ \mathbf{F}^T \mathbf{F}]_{2\times2} = \mathbf{g}_{\mathbf{v}_0} 
\quad \text{ and } \quad \mathbf{F} \mathbf{e}_3 = \lambda_0^{1/4} \lambda_{\emph{f}}^{-1/4} \frac{\mathbf{F} \mathbf{e}_1 \times \mathbf{F}\mathbf{e}_2}{|\mathbf{F} \mathbf{e}_1 \times \mathbf{F}\mathbf{e}_2|}
\end{aligned}
\end{equation}
for $\mathbf{g}_{\mathbf{v}_0}$ given in (\ref{eq:metricDef}).
\end{lem}
\begin{rem}\label{lotsOfIdentsRem}
Under the assumption in (\ref{eq:firstLinIdents}), it also holds that 
\begin{align*}
\mathbf{Q}^T =( \boldsymbol{\ell}_{\mathbf{v}_0}^{\emph{f}} )^{-1/2}( \boldsymbol{\ell}_{\mathbf{v}_0}^{\emph{0}} )^{1/2} \mathbf{F}^T
\end{align*}
\end{rem}
\begin{proof}
First we show (\ref{eq:firstLinIdents}) $\Rightarrow$ (\ref{eq:secLinIdents}). Substituting the second identity into the first in (\ref{eq:firstLinIdents}) gives 
\begin{align*}
\mathbf{v} = \sigma (\boldsymbol{\ell}_{\mathbf{v}}^{\text{f}})^{-1/2} \mathbf{F} (\boldsymbol{\ell}_{\mathbf{v}_0}^0)^{1/2} \begin{pmatrix} \mathbf{v}_0 \\ 0 \end{pmatrix} = \sigma \lambda_0^{1/2} (\boldsymbol{\ell}_{\mathbf{v}}^{\text{f}})^{-1/2} [\mathbf{F}]_{3\times2} \mathbf{v}_0.
\end{align*}
Premultiplying this equation by $(\boldsymbol{\ell}_{\mathbf{v}}^{\text{f}})^{1/2}$ then leads to $\lambda_{\text{f}}^{1/2} \mathbf{v} = \sigma \lambda_0^{1/2} [\mathbf{F}]_{3\times2} \mathbf{v}_0$. The first identity in (\ref{eq:secLinIdents}) is follows because $\mathbf{v} \cdot \mathbf{v} = 1$. 
Next, observe that the first identity in (\ref{eq:firstLinIdents}) gives $(\boldsymbol{\ell}^{\text{f}}_{\mathbf{v}}) = \mathbf{Q} ( \boldsymbol{\ell}_{\mathbf{v}_0}^{\text{f}} )\mathbf{Q}^T$.
The second can then be manipulated as $\mathbf{Q} = \mathbf{Q} ( \boldsymbol{\ell}_{\mathbf{v}_0}^{\text{f}} )^{-1/2} \mathbf{Q}^T \mathbf{F} (\boldsymbol{\ell}_{\mathbf{v}_0}^{0})^{1/2} \Leftrightarrow \mathbf{F} =\mathbf{Q}( \boldsymbol{\ell}_{\mathbf{v}_0}^{\text{f}} )^{1/2}(\boldsymbol{\ell}_{\mathbf{v}_0}^0)^{-1/2}$. Thus,
\begin{equation}
\begin{aligned}\label{eq:FtransposeF}
\mathbf{F}^T \mathbf{F} = (\boldsymbol{\ell}_{\mathbf{v}_0}^0)^{-1/2} ( \boldsymbol{\ell}_{\mathbf{v}_0}^{\text{f}} ) (\boldsymbol{\ell}_{\mathbf{v}_0}^0)^{-1/2} = \begin{pmatrix} \mathbf{g}_{\mathbf{v}_0} & \mathbf{0} \\ \mathbf{0}^T & \lambda_{\text{f}}^{-1/2} \lambda_0^{1/2} \end{pmatrix} .
\end{aligned}
\end{equation} 
It follows that $[ \mathbf{F}^T \mathbf{F}]_{2\times2} = \mathbf{g}_{\mathbf{v}_0}$ as desired. For the final identity, note that $\mathbf{F}\mathbf{e}_1$ and $\mathbf{F} \mathbf{e}_2$ are linearly independent since $\det \mathbf{g}_{\mathbf{v}_0} \neq 0$. We therefore parameterize $\mathbf{F}\mathbf{e}_3$ generically as $\mathbf{F}\mathbf{e}_3 = \alpha \mathbf{F} \mathbf{e}_1 + \beta \mathbf{F} \mathbf{e}_2 + \gamma \mathbf{F} \mathbf{e}_1 \times \mathbf{F} \mathbf{e}_2$, and it follows from (\ref{eq:FtransposeF}) that $\alpha = \beta = 0$ and $\gamma = \lambda_0^{1/4} \lambda_{\text{f}}^{-1/4}|\mathbf{F} \mathbf{e}_1 \times \mathbf{F}\mathbf{e}_2|^{-1}$. Thus, (\ref{eq:firstLinIdents}) $\Rightarrow$ (\ref{eq:secLinIdents}).

Now we show (\ref{eq:secLinIdents}) $\Rightarrow$ (\ref{eq:firstLinIdents}). To start observe from (\ref{eq:secLinIdents}) that $|[\mathbf{F}]_{3\times2} \mathbf{v}_0|^2 = \mathbf{v}_0 \cdot \mathbf{g}_{\mathbf{v}_0} \mathbf{v}_0 = \lambda_{\text{f}} \lambda_0^{-1}$, that 
\begin{align*}
(\boldsymbol{\ell}_{\mathbf{v}}^{\text{f}})^{-1/2} = \lambda_{\text{f}}^{-3/2} \lambda_0\mathbf{F} \begin{pmatrix} \mathbf{v}_0 \\ 0 \end{pmatrix} \otimes \begin{pmatrix} \mathbf{v}_0 \\ 0 \end{pmatrix} \mathbf{F}^T + \lambda_{\text{f}}^{1/4} \Big( \mathbf{I} - \lambda_{\text{f}}^{-1} \lambda_0 \mathbf{F} \begin{pmatrix} \mathbf{v}_0 \\ 0 \end{pmatrix} \otimes \begin{pmatrix} \mathbf{v}_0 \\ 0 \end{pmatrix} \mathbf{F}^T \Big),
\end{align*}
that (\ref{eq:FtransposeF}) holds, and that 
\begin{equation}
\begin{aligned}\label{eq:interestingCommute}
(\boldsymbol{\ell}_{\mathbf{v}}^{\text{f}})^{-1/2} \mathbf{F} = \mathbf{F} \Big[ \lambda_{\text{f}}^{-1/2} \begin{pmatrix} \mathbf{v}_0 \\ 0 \end{pmatrix} \otimes \begin{pmatrix} \mathbf{v}_0 \\ 0 \end{pmatrix} + \lambda_{\text{f}}^{1/4} \Big( \mathbf{I} - \begin{pmatrix} \mathbf{v}_0 \\ 0 \end{pmatrix} \otimes \begin{pmatrix} \mathbf{v}_0 \\ 0 \end{pmatrix} \Big) \Big] = \mathbf{F} (\boldsymbol{\ell}_{\mathbf{v}_0}^{\text{f}})^{-1/2}.
\end{aligned}
\end{equation}
Thus, from (\ref{eq:FtransposeF}) and (\ref{eq:interestingCommute}),
\begin{align*}
\big[(\boldsymbol{\ell}_{\mathbf{v}}^{\text{f}})^{-1/2} \mathbf{F} (\boldsymbol{\ell}_{\mathbf{v}_0}^0)^{1/2}\big]^T (\boldsymbol{\ell}_{\mathbf{v}}^{\text{f}})^{-1/2} \mathbf{F} (\boldsymbol{\ell}_{\mathbf{v}_0}^0)^{1/2} &= \big[ \mathbf{F} (\boldsymbol{\ell}_{\mathbf{v}_0}^{\text{f}})^{-1/2} (\boldsymbol{\ell}_{\mathbf{v}_0}^0)^{1/2}\big]^T ( \mathbf{F} \boldsymbol{\ell}_{\mathbf{v}_0}^{\text{f}})^{-1/2}(\boldsymbol{\ell}_{\mathbf{v}_0}^0)^{1/2} = \mathbf{I}.
\end{align*}
We therefore conclude that $(\boldsymbol{\ell}_{\mathbf{v}}^{\text{f}})^{-1/2} \mathbf{F} (\boldsymbol{\ell}_{\mathbf{v}_0}^0)^{1/2}$ is a rotation since $\det \big[ (\boldsymbol{\ell}_{\mathbf{v}}^{\text{f}})^{-1/2} \mathbf{F} (\boldsymbol{\ell}_{\mathbf{v}_0}^0)^{1/2}\big] = \det \big[ \mathbf{F} \big]$ and since $\det \big[ \mathbf{F} \big]$ is clearly positive because $\mathbf{F} \mathbf{e}_3 = \lambda_0^{1/4} \lambda_{\text{f}}^{-1/4} \frac{\mathbf{F} \mathbf{e}_1 \times \mathbf{F}\mathbf{e}_2}{|\mathbf{F} \mathbf{e}_1 \times \mathbf{F}\mathbf{e}_2|}$. Next observe that 
\begin{align*}
\sigma (\boldsymbol{\ell}_{\mathbf{v}}^{\text{f}})^{-1/2} \mathbf{F} (\boldsymbol{\ell}_{\mathbf{v}_0}^0)^{1/2} \begin{pmatrix} \mathbf{v}_0 \\ 0 \end{pmatrix} = \sigma \mathbf{F} (\boldsymbol{\ell}_{\mathbf{v}_0}^{\text{f}})^{-1/2}(\boldsymbol{\ell}_{\mathbf{v}_0}^0)^{1/2} \begin{pmatrix} \mathbf{v}_0 \\ 0 \end{pmatrix} = \sigma \lambda_{\text{f}}^{-1/2} \lambda_0^{1/2} [\mathbf{F}]_{3\times2} \mathbf{v}_0 = \sigma \frac{[\mathbf{F}]_{3\times2} \mathbf{v}_0}{|[\mathbf{F}]_{3\times2} \mathbf{v}_0|} = \mathbf{v}
\end{align*}
since $|[\mathbf{F}]_{3\times2} \mathbf{v}_0| = \lambda_{\text{f}}^{1/2} \lambda_0^{-1/2}$. Thus, (\ref{eq:secLinIdents}) $\Rightarrow$ (\ref{eq:firstLinIdents}).
\end{proof}

Now consider the quadratic form relevant to incompressible plates 
\begin{align*}
Q_2(\mathbf{A}) := 2 \mu \Big\{ \big|\text{sym}([ \mathbf{A} ]_{2\times2}) \big|^2 + \text{Tr}\big( \text{sym}( [ \mathbf{A} ]_{2\times2})\big)^2\Big\} , \quad \mathbf{A} \in \mathbb{R}^{3\times2}. 
\end{align*}
\begin{lem}\label{Q3toQ2Lemma}
$Q_2$ and $Q_3$ are related via 
\begin{equation}
\begin{aligned}\label{eq:Q3toQ2}
Q_2(\mathbf{A}) = \min_{\mathbf{d} \in \mathbb{R}^{3}} \{ Q_3\big((\mathbf{A}, \mathbf{d}) \big) \colon \emph{Tr}\big( ( \mathbf{A}, \mathbf{d}) \big) = 0 \Big\} , \quad \mathbf{A} \in \mathbb{R}^{3\times2}.
\end{aligned}
\end{equation}
In particular, $\mathbf{Q}_3( \mathbf{A}) \geq Q_2( [\mathbf{A}]_{3\times2})$ for all $\mathbf{A} \in \mathbb{R}^{3\times3}$ such that $\emph{Tr}(\mathbf{A}) = 0$. 
\end{lem}
\begin{proof}
Notice that $Q_3\big((\mathbf{A}, \mathbf{d}) \big) = 2\mu \Big( \big| \text{sym} ( [\mathbf{A}]_{2\times2} ) \big|^2 + (\mathbf{d} \cdot \mathbf{e}_3)^2 + \frac{1}{2}( \mathbf{d} \cdot \mathbf{e}_1 + \mathbf{e}_3 \cdot \mathbf{A} \mathbf{e}_1)^2 + \frac{1}{2}( \mathbf{d} \cdot \mathbf{e}_2 + \mathbf{e}_3 \cdot \mathbf{A} \mathbf{e}_2)^2 \Big) + \lambda \text{Tr}\big( \text{sym}(\mathbf{A}, \mathbf{d}) \big)^2 $. Thus, for any $\mathbf{d} \in \mathbb{R}^3$ such that $\text{Tr}(( \mathbf{A}, \mathbf{d})) = 0$, we have $\lambda \text{Tr}\big( \text{sym}(\mathbf{A}, \mathbf{d}) \big)^2 = 0$ and $(\mathbf{d} \cdot \mathbf{e}_3)^2 = \text{Tr}\big( \text{sym}( [ \mathbf{A} ]_{2\times2})\big)^2$. So, 
\begin{align*}
Q_3\big((\mathbf{A}, \mathbf{d}) \big)= &2 \mu \Big\{ \big|\text{sym}([ \mathbf{A} ]_{2\times2}) \big|^2 + \text{Tr}\big( \text{sym}( [ \mathbf{A} ]_{2\times2})\big)^2\Big\}
\\
&+\frac{1}{2}\big(( \mathbf{d} \cdot \mathbf{e}_1 + \mathbf{e}_3 \cdot \mathbf{A} \mathbf{e}_1)^2
+ ( \mathbf{d} \cdot \mathbf{e}_2 + \mathbf{e}_3 \cdot \mathbf{A} \mathbf{e}_2)^2 \big) \quad \text{ if } \text{Tr}(( \mathbf{A}, \mathbf{d})) = 0.
\end{align*}
Minimizing out the unconstrained components of $\mathbf{d}$, i.e., $\mathbf{d} \cdot \mathbf{e}_{1}$ and $\mathbf{d} \cdot \mathbf{e}_2$, gives the result in (\ref{eq:Q3toQ2}). The inequality stated in the lemma is a trivial consequence of this result. 
\end{proof}

\bibliographystyle{plain}
\bibliography{references}

\end{document}